\documentclass[reqno]{amsart}
\usepackage{comment,amssymb,latexsym,upref,enumerate,fouridx}
\usepackage{mathrsfs,color}
\usepackage{centernot}
\usepackage[colorlinks,linkcolor=blue,citecolor=blue]{hyperref}
\usepackage{graphicx}

\allowdisplaybreaks

\numberwithin{equation}{section}

\theoremstyle{plain}
\newtheorem{thm}{Theorem}[section]
\newtheorem{lem}[thm]{Lemma}
\newtheorem{prop}[thm]{Proposition}

\theoremstyle{definition}

\theoremstyle{remark}
\newtheorem{rem}[thm]{Remark}

\newcommand{\ep}{\epsilon}
\newcommand{\vep}{\varepsilon}

\newcommand{\Bcv}{\boldsymbol{\Bc}{}}

\newcommand{\Kttt}{\tilde{\Ktt}{}}
\newcommand{\gttt}{\tilde{\gtt}{}}

\newcommand{\Hcb}{\bar{\Hc}}
\newcommand{\alphab}{\bar{\alpha}}

\newcommand{\Hct}{\tilde{\mathcal{H}}}

\newcommand{\Mft}{\tilde{\mathfrak{M}}}
\newcommand{\Bft}{\tilde{\mathfrak{B}}}
\newcommand{\Cft}{\tilde{\mathfrak{C}}}
\newcommand{\Hft}{\tilde{\mathfrak{H}}}
\newcommand{\Dft}{\tilde{\mathfrak{D}}}
\newcommand{\Uft}{\tilde{\mathfrak{U}}}

\newcommand{\gttb}{\bar{\gtt}}
\newcommand{\Kttb}{\bar{\Ktt}}

\newcommand{\Ufb}{\bar{\mathfrak{U}}{}}
\newcommand{\Pfb}{\bar{\mathfrak{P}}{}}
\newcommand{\Ucb}{\bar{\mathcal{U}}{}}
\newcommand{\Uct}{\tilde{\mathcal{U}}{}}
\newcommand{\Pft}{\tilde{\mathfrak{P}}{}}
\newcommand{\Sft}{\tilde{\mathfrak{S}}{}}
\newcommand{\Vcb}{\bar{\mathcal{V}}{}}

\newcommand{\alphabr}{\breve{\alpha}{}}
\input Tdef.def

\newcounter{mnotecount}[section]
\let\oldmarginpar\marginpar
\renewcommand\marginpar[1]{\-\oldmarginpar[\raggedleft\footnotesize #1]%
 {\raggedright\footnotesize #1}}

\begin{document}

\title[Localized past stability of the subcritical Kasner-scalar field spacetimes]{Localized past stability of the subcritical Kasner-scalar field spacetimes}

\author[F.~Beyer]{Florian Beyer}
\address{Dept of Mathematics and Statistics\\
730 Cumberland St\\
University of Otago, Dunedin 9016\\ New Zealand}
\email{fbeyer@maths.otago.ac.nz }

\author[T.A.~Oliynyk]{Todd A. Oliynyk}
\address{School of Mathematical Sciences\\
9 Rainforest Walk\\
Monash University, VIC 3800\\ Australia}
\email{todd.oliynyk@monash.edu}

\author[W.~Zheng]{Weihang Zheng}
\address{School of Mathematical Sciences\\
9 Rainforest Walk\\
Monash University, VIC 3800\\ Australia}
\email{weihang.zheng@monash.edu }

\begin{abstract}
We prove the nonlinear stability, in the contracting direction, of the entire subcritical family of Kasner–scalar field solutions to the Einstein–scalar field equations in four spacetime dimensions. Our proof relies on a zero-shift, orthonormal frame decomposition of a conformal representation of the Einstein-scalar field equations. To synchronise the big bang singularity, we use the time coordinate $\tau = \exp\bigl(\frac{2}{\sqrt{3}}\phi\bigr)$, where $\phi$ is the scalar field, which coincides with a conformal harmonic time slicing. 

We show that the perturbed solutions are asymptotically
pointwise Kasner, geodesically incomplete to the past and terminate at quiescent, crushing big bang singularities located at $\tau=0$, which are characterised by curvature blow up. Specifically, we establish two stability theorems. The first is a global-in-space stability result where the perturbed spacetimes are of the form $M = \bigcup_{t\in (0,t_0]} \tau^{-1}(\{t\}) \cong (0,t_0] \times \mathbb{T}^{3}$. The second is a localised version where the perturbed spacetimes are given by $M=\bigcup_{t\in (0,t_0]}\tau^{-1}(\{t\})\cong \bigcup_{t\in (0,t_0]} \{t\}\times \mathbb{B}_{\rho(t)}$ with time-dependent radius function
$\rho(t) = \rho_0 + (1-\vartheta) \rho_0 \bigl( \bigl(\frac{t}{t_0}\bigr)^{1-\epsilon}-1\bigr)$. Spatial localisation is achieved through our choice of zero-shift, harmonic time slicing that leads to hyperbolic evolution equations with a finite propagation speed. 
\end{abstract}

\maketitle

\section{Introduction}
The study of cosmological solutions near big bang singularities is essential for understanding the early universe and the nature of spacetime singularities predicted by general relativity. Remarkable progress in this area has been made in recent years. Pioneering contributions appeared in \cite{RodnianskiSpeck:2018b,RodnianskiSpeck:2018c,Speck:2018}, where the past nonlinear stability of the Friedmann-Lema\^{i}tre-Robertson-Walker (FLRW) family of solutions to the Einstein-scalar field equations was rigorously established. These works demonstrated that perturbed FLRW solutions terminate in the past at spacelike big bang singularities that are quiescent (non-oscillatory) and characterised by curvature blow up. Notably, this marked the first rigorous demonstration of big bang stability without imposing symmetry assumptions.

The stability results established in \cite{RodnianskiSpeck:2018b,RodnianskiSpeck:2018c,Speck:2018} were extended in \cite{RodnianskiSpeck:2022} to demonstrate the past stability of moderately anisotropic Kasner-scalar field solutions. This line of research culminated in \cite{Fournodavlos_et_al:2023} where the past stability of the entire subcritical range of Kasner-scalar field solutions was established. Building on these methods, \cite{Groeniger_et_al:2023} further extended these stability results to allow for non-vanishing scalar potentials and a significantly broader class of subcritical initial data, showing that such configurations lead to quiescent big bang singularities.

The past stability proofs in \cite{Fournodavlos_et_al:2023} and \cite{Groeniger_et_al:2023} rely on foliating spacetime by level sets of a time function \( t \) that are spacelike and of constant mean curvature (CMC). This time function possesses remarkable geometric and analytic properties that play a central role in proof of these stability results. A key feature this time function is its ability to \textit{synchronise} the big bang singularity, aligning it with the hypersurface $t = 0$. This synchronisation is crucial because it enables statements to be made about the behaviour of the physical fields as the singularity is approached, ie.~in the limit $t\searrow 0$, that are uniform across the whole singular hypersurface.

The use of CMC foliations introduces a non-local dependence of perturbed solutions on the initial data. Consequently, the stability proofs in \cite{Fournodavlos_et_al:2023,Groeniger_et_al:2023} cannot be localised, that is, restricted to truncated cone-shaped domains that extend to the singularity.  By contrast, certain gauge choices, such as wave gauges, ensure that solutions of the Einstein equations depend locally on the initial data. This raises a natural question: Can a localised version of the stable big bang formation results in \cite{Fournodavlos_et_al:2023} be established?

A partial positive answer to this question was provided in \cite{BeyerOliynyk:2024b}, which demonstrated that FLRW solutions to the Einstein-scalar field equations, along with their big bang singularities, are locally stable to the past in spacetime dimensions $n \geq 3$. The main aim of this article is to show in four spacetime dimensions that the family of subcritical Kasner scalar field spacetimes and their big bang singularities are locally stable to the past. This confirms that the big bang stability results of 
\cite{Fournodavlos_et_al:2023} can indeed be localised, at least in four spacetime dimensions.

\subsection{Einstein-scalar field equations}
The Einstein-scalar field equations on a four dimensional spacetime $(M,\gb)$ are given by\footnote{See Section \ref{indexing} for our indexing conventions.}
\begin{align}
\Gb_{\mu\nu}&=\Tb_{\mu\nu}, 
\label{ESF.1}\\
\Box_{\gb} \phi &=0, \label{ESF.2}
\end{align}
where $\nablab_\mu$ is the Levi-Civita connection of the physical spacetime metric $\gb_{\mu\nu}$, $\Gb_{\mu\nu}$ is the Einstein tensor of $\gb_{\mu\nu}$, $\Box_{\gb} =\gb^{\mu\nu}\nablab_\mu\nablab_\nu$ is the wave operator,
and
\begin{equation} \label{stress-energy}
\Tb_{\mu\nu} = 2\nablab_\mu\phi\nablab_\nu\phi - |\nablab\phi|^2_{\gb} \gb_{\mu\nu}
\end{equation}
is the stress-energy tensor of the scalar field $\phi$. For use below, we recall that the Einstein equation \eqref{ESF.1} can be equivalently expressed as 
\begin{equation} \label{ESF.3}
\Rb_{\mu\nu} = \Tb_{\mu\nu}-\frac{1}{2}\gb^{\alpha\beta}\Tb_{\alpha\beta}\gb_{\mu\nu} = 2\nablab_\mu \phi \nablab_\nu\phi
\end{equation}
where $\Rb_{\mu\nu}$ is the Ricci curvature tensor of $\gb_{\mu\nu}$. 

\subsection{Kasner-scalar field spacetimes\label{KSFS}}
The \textit{Kasner-scalar field spacetimes} $(\Rbb_{>0}\times\Tbb^3,\mathring{\gb})$ are a family of solutions to the Einstein-scalar field equations \eqref{ESF.1}-\eqref{ESF.2} defined by  
\begin{equation} \label{Kasner-solns-A} 
\biggl\{\mathring{\gb}=t\mathring{\gt},\,
\mathring{\phi}=\frac{\sqrt{3}}{2}\ln(t)\biggr\}
\end{equation}
where 
\begin{equation*}
\mathring{\gt} = -t^{r_0}dt\otimes dt + \sum_{\Omega=1}^{3} t^{r_\Omega} dx^\Omega \otimes dx^\Omega.
\end{equation*}
Here, $r_0$ and $r_\Omega$, $\Omega=1,2,3$, are constants defined by
\begin{equation}
    r_0 = \frac{\sqrt{3}}{P}-3 \AND
    r_\Omega =\frac{\sqrt{3}}{P} q_\Omega
    -1 \label{Kasner-exps-A}
\end{equation}
where  $P\in (0,\sqrt{1/3}]$ and the $q_\Omega\in \Rbb$, $\Omega=1,2,3$, known as the \textit{Kasner exponents}, satisfy the \textit{Kasner relations} 
\begin{equation}\label{Kasner-rels-A}
    \sum_{\Omega=1}^{3}q_\Omega =1 \AND 
    \sum_{\Omega=1}^{3}q_\Omega^2 = 1-2 P^2. 
\end{equation}
The Kasner relations can be equivalently expressed as
\begin{equation} \label{Kasner-rels-B}
\sum_{\Omega=1}^3 r_\Omega = r_0  \AND  \sum_{\Omega=1}^3 r_\Omega^2 = (r_0+2)^2 - 4,
\end{equation}
where we note, by \eqref{Kasner-exps-A} and $0<P\leq\sqrt{1/3}$, that 
\begin{equation*} 
r_0\geq 0.
\end{equation*}
We will refer to the constants $r_\Omega$ as the \textit{conformal Kasner exponents} and \eqref{Kasner-rels-B} as the \textit{conformal Kasner relations}. It is shown in \cite[Section~1.3]{BeyerOliynyk:2024b} that $P$ can be interpreted as the \emph{scalar field strength} where $P=0$ corresponds to vacuum (the scalar field vanishes; this case is excluded here) and where the largest possible value $P= \sqrt{1/3}$ implies by \eqref{Kasner-rels-A} that all Kasner exponents have the same value $q_{\Omega}=1/3$.

In four spacetime dimensions, the Kasner exponents are said to be \textit{subcritical} \cite[Eq.~(1.8)] {Fournodavlos_et_al:2023}  provided they satisfy 
\begin{equation} \label{SCR}
    \max\limits_{\substack{\Omega,\Lambda,\Gamma \in \{1,2,3\} \\ \Omega < \Lambda}}\{q_\Omega+q_\Lambda-q_\Gamma\} < 1.
\end{equation}

\subsection{Informal statement of the main results}
Theorems \ref{glob-stab-thm} and \ref{loc-stab-thm} provide precise statements of the main results of this article. The first stability theorem, Theorem \ref{glob-stab-thm}, establishes the  nonlinear stability in the contracting direction of the subcritical Kasner-scalar field solutions \eqref{Kasner-solns-A} to the Einstein-scalar field equations in four spacetime dimensions. Since the initial data for Theorem \ref{glob-stab-thm} is specified on a closed hypersurface, this result represents a past global-in-space stability result. While this stability was first proven in \cite{Fournodavlos_et_al:2023}, the approach and proof presented here are new. A key advantage of this new proof is that it can be localised in space.  

Below is an informal statement of Theorem \ref{glob-stab-thm}.
As we discuss in detail below, we call an initial data set \emph{synchronised} provided the initial data for the scalar field is constant, see \eqref{eq:synchroID} below in Remark \ref{rem:constraints}. We further recall that the term \emph{asymptotically pointwise Kasner} is defined in \cite[Def.~1.1]{BeyerOliynyk:2024b}.

\begin{thm}[Past stability of the subcritical Kasner-scalar field spacetimes]
Solutions $\{\gb,\phi\}$ of Einstein-scalar field equations \eqref{ESF.1}-\eqref{ESF.2} that are 
generated from sufficiently differentiable, synchronised initial data that is imposed on $\{t_0\}\times \Tbb^{3}$ and suitably close to subcritical Kasner-scalar field initial data exist on the spacetime region $M \cong \bigcup_{t\in (0,t_0]}\tau^{-1}(\{t\})\cong (0,t_0]\times \Tbb^{3}$ where $\tau=e^{\frac{2}{\sqrt{3}}\phi}$. Moreover, these solutions are asymptotically pointwise Kasner, $C^2$-inextendible through the $\tau=0$ boundary of $M$, past timelike geodesically incomplete and terminate at quiescent, crushing big bang singularities, which are located at $\tau=0$ and characterised by curvature blow up. 
\end{thm}

The second stability theorem, Theorem~\ref{loc-stab-thm}, is a localised version of Theorem~\ref{glob-stab-thm} and a new result.  Here, initial data is specified on an open centred ball $\{t_0\}\times \mathbb{B}_{\rho_0}$ in $\{t_0\}\times \Tbb^{3}$. An informal statement of this theorem is as follows.

\begin{thm}[Localised past stability of subcritical Kasner-scalar field spacetimes]
Given $\vartheta\in (0,1)$ and an open centred ball $\mathbb{B}_{\rho_0}\subset \Tbb^{3}$, there exist a $t_0>0$ and $\ep\in (0,1)$ such that solutions $\{\gb,\phi\}$ of the  Einstein-scalar field equations \eqref{ESF.1}-\eqref{ESF.2} that are 
generated from sufficiently differentiable, synchronised initial data that is imposed on $\{t_0\}\times \mathbb{B}_{\rho_0}$ and suitably close to subcritical Kasner-scalar field initial data exist on the spacetime region
$M \cong \bigcup_{t\in (0,t_0]}\tau^{-1}(\{t\})\cong \bigcup_{t\in (0,t_0]} \{t\}\times \mathbb{B}_{\rho(t)}$ where $\tau=e^{\frac{2}{\sqrt{3}}\phi}$ and $\rho(t) = \rho_0 + (1-\vartheta) \rho_0 \bigl( \bigl(\frac{t}{t_0}\bigr)^{1-\ep}-1\bigr)$. Moreover, these solutions are asymptotically pointwise Kasner, $C^2$-inextendible through the $\tau=0$ boundary of $M$, past timelike geodesically incomplete and terminate at quiescent, crushing big bang singularities, which are located at $\tau=0$ and characterised by curvature blow up.  
\end{thm}
It is also worth mentioning that the perturbed solutions from the above two theorems exhibit the expected asymptotically velocity term dominated (AVTD) behaviour \cite{Eardley:1972,isenberg1990}

\subsection{Prior and related work}
Among the earliest results in the literature on big bang asymptotics are \cite{chruscielStrongCosmicCensorship1990,isenberg1990}, which cover the case of vacuum solutions with polarised Gowdy symmetry. A subsequent major breakthrough was Ringstr\"{o}m’s comprehensive analysis of the big bang asymptotics of vacuum Gowdy spacetimes in \cite{ringstrom2008,ringstrom2009a}. More recently, the stability of a large family of polarised $\Tbb^2$-symmetric vacuum solutions (allowing for a cosmological constant) was treated in \cite{ABIO:2022_Royal_Soc,ABIO:2022}. We also mention the related work by Li \cite{li2024b,li2024a} on aspects of the singular dynamics of solutions with symmetries.

As discussed above, pioneering studies on big bang stability without symmetry restrictions were carried out in \cite{RodnianskiSpeck:2018b,RodnianskiSpeck:2018c,RodnianskiSpeck:2022,Speck:2018}, and more recently, the past stability of the entire subcritical range of Kasner–scalar field solutions was established in \cite{Fournodavlos_et_al:2023}.  Building on this and on the general framework for the Einstein–scalar field equations introduced in \cite{ringstrom2017,ringstrom2021a,ringstrom2021}, it was shown in \cite{Groeniger_et_al:2023} that these stability results can be extended to allow for non-vanishing scalar potentials and a significantly broader class of subcritical initial data. We also recall the closely related big bang stability results from \cite{FajmanUrban:2022,FajmanUrban:2024,Urban:2024} for the Einstein equations coupled to scalar fields and Vlasov matter. While all of these works employ a CMC foliation to synchronise the singularity, the idea of using the scalar field as a time function, which is also used in this paper, was introduced in \cite{BeyerOliynyk:2024b} and led to a localised past stability result for FLRW solutions of the Einstein–scalar field equations, as well as a past stability result for FLRW solutions of the Einstein–Euler–scalar field system in \cite{BeyerOliynyk:2024c}.

Complementing these stability results, many families of big bang solutions parameterised by data on the singularity have been constructed by solving a Fuchsian \emph{singular initial value problem}. The first results in this direction appeared in the contexts of Gowdy symmetry \cite{ames2017,beyer2017,kichenassamy1998,rendall2000,stahl2002}, $\Tbb^2$-symmetry \cite{ames2013a,isenberg1999}, and $\Tbb^1$-symmetry \cite{choquet-bruhat2006,choquet-bruhat2004}.

Among the first works that do not impose any symmetry restrictions in solving the singular initial value problem was \cite{andersson2001}, which addressed the Einstein–scalar field equations under an analyticity restriction. The corresponding problem for the vacuum equations was solved later in \cite{klinger2015} (still under an analyticity restriction) and in \cite{AthanasiouFournodavlos:2024,FournodavlosLuk:2023} without assuming analyticity. In vacuum, the resulting solution families are expected to be non-generic due a non-oscillatory condition being enforced \cite[Eq.~(1.4)]{AthanasiouFournodavlos:2024}.
Notably, the results from \cite{AthanasiouFournodavlos:2024} are localised in space and closely related to the results established here in the sense that they provided a singular initial value version of the localised stability results established in this article, albeit for the vacuum Einstein equations. 
Finally, we point out that a general framework for prescribing data on the singularity for the Einstein–scalar field equations has been developed in \cite{ringstrom2022a,ringstrom2022}.

\subsection{Proof summary}

\subsubsection{A tetrad formulation of the conformal Einstein-scalar field equations}
The big bang stability results from Theorems \ref{glob-stab-thm} and \ref{loc-stab-thm} are derived using an orthonormal frame formulation of the conformal Einstein-scalar field equations introduced in Section \ref{sec:tetrad}. In this formulation, the coordinate and frame gauge freedom is fixed by employing a zero-shift $3+1$-decomposition, evolving the spatial frames via Fermi-Walker transport, and using the scalar field $\tau = e^{\frac{2}{\sqrt{3}}\phi}$ as a time function (see \eqref{tau-def} and \eqref{t-fix}). These choices yield gauge-reduced evolution equations that are hyperbolic with a finite propagation speed. This is in contrast to 
\cite{Fournodavlos_et_al:2023,Groeniger_et_al:2023} where the inverse mean curvature $(K_A^{A})^{-1}$ was used as a time function and results in gauge-reduced evolution equations with an infinite propagation speed. 

The finite propagation speed for our gauge-reduce equations is the key property that allows us to establish the localised past stability of the subcritical Kasner-scalar field spacetimes. As is standard, these equations can be separated into evolution and constraint equations, and are presented in Section \ref{sec:complete-evolve-constraint} with the evolution and constraint equations given by \eqref{EEb.1}-\eqref{EEb.7} and \eqref{A-cnstr-2}-\eqref{H-cnstr-2}, respectively.

\subsubsection{Local-in-time existence, uniqueness, continuation and constraint propagation}
The local-in-time existence and uniqueness of solutions to our evolution equations, along with a continuation principle and constraint propagation, is established in Section \ref{sec:local-existence-theory} on spacetime cylinders of the form $(t_1,t_0]\times\Tbb^3$ and truncated cone domains (see Section \ref{domains}). Precise statements of these results are presented in Propositions \ref{prop:locA} and \ref{prop:locC}.

We prove the local-in-time existence, uniqueness, and continuation of solutions by showing in Section \ref{sec:sym-hyp} that our evolution equations can be cast into symmetric hyperbolic form (see \eqref{EE-matform}). Standard hyperbolic PDE theory then guarantees these results hold. To establish constraint propagation, we demonstrate in Section \ref{const-hyp} that the evolution equations for the constraints are strongly hyperbolic (see \eqref{cprop-sys} and Lemma \ref{lem:chyp}). This enables us to apply standard hyperbolic PDE theory to show that if the constraints vanish on the initial hypersurface, they continue to hold throughout the evolution.

\subsubsection{Global estimates via Fuchsian techniques}
While the tetrad formulation of the conformal Einstein-scalar field equations employed in Section \ref{sec:local-existence-theory} is effective for establishing local-in-time existence of solutions, it is not suitable for analysing the behaviour of nonlinear perturbations near $t=0$. For this, we use a Fuchsian formulation of the evolution equations, given by \eqref{FuchFinal-A}, to establish bounds that hold globally to the past. This formulation, derived in Section \ref{sec:Fuch-global-existence-theory}, consists of low and high order equations given by \eqref{FuchB} and \eqref{FuchE}, respectively.

The significance of the Fuchsian equation \eqref{FuchFinal-A} is that, as shown in Section \ref{coeff-props}, it satisfies all the necessary conditions to apply the Fuchsian global existence theory developed in \cite{BOOS:2021} provided the Kasner exponents satisfy the subcritical condition \eqref{SCR}. Specifically, applying \cite[Thm.~3.8]{BOOS:2021} ensures the past global stability of the trivial solution to \eqref{FuchFinal-A} on spacetime cylinders of the form \((0, t_0] \times \Tbb^3\) and truncated cone domains intersecting \( t = 0 \). The precise statements of these stability results can be found in Propositions \ref{prop:Fuch-global} and \ref{prop:Fuch-global-loc}.

\subsubsection{Past nonlinear stability of the subcritical Kasner-scalar field spacetimes}
Using the local-in-time existence theory from Propositions \ref{prop:locA} and \ref{prop:locC} in conjunction with the Fuchsian stability results from Propositions \ref{prop:Fuch-global} and \ref{prop:Fuch-global-loc}, we prove in Section \ref{sec:stable-big-bang} the past nonlinear stability of the subcritical Kasner-scalar field spacetimes on both spacetime cylinders of the form $(0,t_0]\times\Tbb^3$ and truncated cone domains that intersect $t=0$. Additionally, we show the perturbed spacetimes are asymptotically pointwise Kasner, past geodesically incomplete, and terminate at quiescent, crushing big bang singularities, which are located at $t=0$ and characterised by curvature blow up. The precise statements of these stability results are presented in Theorems \ref{glob-stab-thm} and \ref{loc-stab-thm}.

\section{Preliminaries\label{prelim}}

\subsection{Data availability statement}

This article has no associated data.

\subsection{Indexing conventions\label{indexing}}
In the article, we consider four dimensional spacetime manifolds of the form
\begin{equation*} 
    M_{t_1,t_0}= (t_1,t_0]\times \Tbb^{3}
\end{equation*}
where $t_0>0$, $t_1<t_0$, and $\Tbb^{3}$ is
the $3$-torus defined by
\begin{equation*}
    \Tbb^{3} = [-L,L]^{3}/\sim
\end{equation*} 
with $\sim$ the equivalence relation obtained
from identifying the sides of the box $[-L,L]^{3}\subset \Rbb^{3}$.

\subsubsection{Coordinates and partial derivatives}
On $M_{t_1,t_0}$,
$(x^\mu)=(x^0,x^\Omega)$ are spacetime coordinates 
where $(x^\Omega)$ are periodic spatial coordinates
on $\Tbb^{3}$ and $x^0$ is a time coordinate
on the interval $(t_1,t_0]$. Lower case Greek letters, eg. $\mu,\nu,\gamma$, run from $0$ to $3$ and label spacetime coordinate indices while upper case Greek letters, eg. $\Lambda,\Omega,\Gamma$, run from $1$ to $3$ and label spatial coordinate indices. Partial derivatives with respect to the coordinates $(x^\mu)$ are denoted by $\del{\mu} = \frac{\del{}\;}{\del{}x^\mu}$.
We often use $t$ to denote the time coordinate $x^0$, ie.~$t=x^0$, and use the notion $\del{t} = \del{0}$ for the partial derivative with respect to the coordinate $x^0$.

\subsubsection{Frames}
\label{sec:frames}
Orthonormal frames $e_a= e_a^\mu \del{\mu}$ are employed in this article. Lower case Latin letters, eg. $a,b,c$, label frame indices and they run from $0$ to $3$ while spatial frame indices are labelled using upper case Latin letters, eg. $A,B,C$, that run from $1$ to $3$.

\subsubsection{Multiindices\label{multiindex}} 
Multiindices are used to enumerate higher order spatial derivatives and we use lower case calligraphic letters, eg. $\ac,\bc,\cc$, to denote these indices.
In $\Rbb^3$, a multiindex is a $3$-tuple 
$\bc = (\bc_1,\bc_2,\bc_3)\in \Nbb_0^3$ of non-negative integers, and its length is defined by $|\bc| = \bc_1+\bc_2+\bc_3$. If $\ac$ and $\bc$ are two multiindices, then $\ac\leq\bc$ provided $\ac_j\leq\bc_j$ for $1\leq j\leq 3$. In this case, $\bc-\ac$ is a multiindex. 

For a given multiindex $\bc$, we use $\del{}^\bc$ to denote the $|\bc|$th-order spatial partial differential operator
\begin{equation*}
    \del{}^\bc = \del{1}^{\bc_1}\del{2}^{\bc_2}\del{3}^{\bc_3}.
\end{equation*}
Also, for $k\in \Nbb_0$ and a function $u(x)$,
\begin{equation*}
\del{}^k u =\{\, (\del{}^\bc u)\,|\, |\bc|=k\,\}
\end{equation*}
denotes the \textit{$k$th-order gradient}, that is, the collection of all spatial derivatives of $u(x)$ of order $k$.

\subsubsection{Index operations\label{Index-op}}
On two component spatial tensors fields, the \textit{symmetrisation}, \textit{anti-symmetrisation} and \textit{symmetric trace free} operations are defined, respectively, by
\begin{gather*}
\sigma_{(AB)} = \frac{1}{2}(\sigma_{AB}+\sigma_{BA}), \quad
\sigma_{[AB]} = \frac{1}{2}(\sigma_{AB}-\sigma_{BA}) \AND
\sigma_{\expval{AB}}=\sigma_{(AB)}-\frac{1}{3}\delta^{CD}\sigma_{CD}\delta_{BA}.
\end{gather*}
We employ the $*$-notation to denote products of spatial tensor for which it is not necessary to know the exact values of the constant coefficients appearing in the product. For example, given two spatial tensors  $\sigma_{AB}$ and $\nu_A$, 
$[\sigma*\nu]_A$ denotes a product of the form
\begin{align*}
[\sigma*\nu]_{A}= C_{A}^{BCD}\sigma_{BC}\nu_{D}
\end{align*}
where the $C_{A}^{BCD}$ are constants. We also use $[(\sigma+\nu)*(\rho+\kappa)]$ to denote products of the form
\begin{equation*}
[(\sigma+\nu)*(\rho+\kappa)]=[\sigma*\rho]+[\sigma*\kappa]+[\nu*\rho]+[\nu*\kappa].
\end{equation*}

\subsection{Inner-products and matrices}
The Euclidean inner-product is denoted by $\ipe{\xi}{\zeta} = \xi^{\tr} \zeta$, $\xi,\zeta \in \Rbb^N$, while
$|\xi| = \sqrt{\ipe{\xi}{\xi}}$
defines the Euclidean norm. The set of $N\times N$ matrices is denoted by  $\Mbb{N}$, and $\Sbb{N}$ is the subspace of symmetric $N\times N$-matrices. The operator norm of a matrix $A\in \Mbb{N}$ is, by definition,
\begin{equation*}
  |A|_{\op} = \sup_{\xi\in \Rbb^N_\times} \frac{|A\xi|}{|\xi|},
\end{equation*}
where $\Rbb^N_\times = \Rbb^N\setminus\{0\}$, while, 
for any $A,B\in \Mbb{N}$, we define 
\begin{equation*}
    A\leq B \quad \Longleftrightarrow \quad \ipe{\xi}{A\xi} \leq \ipe{\xi }{B\xi}, \quad \forall \; \xi \in \Rbb^N.
\end{equation*}
Additionally, we employ the Frobenius inner-product 
\begin{equation}\label{Frobenius}
\ipe{A}{B}=\text{Tr}(A^{\tr}B), \quad A,B\in \Mbb{N},
\end{equation} 
on square matrices and use  $|A| = \sqrt{\ipe{A}{A}}$ 
to denote the associated norm. We recall that this norm satisfies
\begin{equation} \label{Frobenius-inq}
|AB|\leq |A| |B|, \quad \forall \, A,B\in \Mbb{N}.
\end{equation}

Similar notation will be employed on $\Cbb^N$ except that the transpose is replaced by the Hermitian adjoint $\dagger$ and the Euclidean inner product is replaced by the Hermitian inner product $\ipe{\xi}{\zeta} = \xi^{\dagger} \zeta$, $\xi,\zeta \in \Cbb^N$. In particular, the operator norm of a complex $N\times N$-matrix $A$ is defined in the same way, ie.~$|A|_{\op} = \sup_{\xi\in \Cbb^N_\times} \frac{|A\xi|}{|\xi|}$, while the notation 
\begin{equation*}
    A\leq B \quad \Longleftrightarrow \quad \ipe{\xi}{A\xi} \leq \ipe{\xi }{B\xi}, \quad \forall \; \xi \in \Cbb^N,
\end{equation*}
only makes sense for self-adjoint complex $N\times N$-matrices (ie. $A^\dagger=A$ and $B^\dagger=B$). 
\subsection{Balls and truncated cone domains\label{domains}}
Inside $\Tbb^{3}$, we define, for $0<\rho<L$, the \textit{centred ball of
radius} $\rho$ by
\begin{equation*}
    \mathbb{B}_{\rho} = \{\, x\in (-L,L)^{3}\, |\, |x|<\rho \, \} \subset \Tbb^{3},
\end{equation*}
where $|x|=\sqrt{\delta_{\Lambda\Sigma}x^\Lambda x^{\Sigma}}$.
Given constants $\rho_1>0$, $0\leq t_1<t_0$,  $0<\rho_0<L$ and $0\leq \ep<1$ satisfying
\begin{equation*}
\rho_0 -\frac{\rho_1 t_0^{1-\ep}}{1-\ep}>0,    
\end{equation*}
we define the \textit{truncated cone domain} $\Omega_{\qv}\subset M_{t_1,t_0}$ via
\begin{equation} \label{Omega-def}
    \Omega_{\qv} = \Bigl\{\, (t,x)\in (t_1,t_0]\times (-L,L)^{3}\, \Bigl| \,|x|< \frac{\rho_1(t^{1-\ep}-t_0^{1-\ep})}{1-\ep}+\rho_0 \,  \Bigr\}, \quad \qv = (t_1,t_0,\rho_0,\rho_1,\ep).
\end{equation}
The ``side'' component of the boundary of $\Omega_{\qv}$, which we denote by 
$\Gamma_{\qv}$, is determined by the vanishing of the function $\chi =|x|-\frac{\rho_1(t^{1-\ep}-t_0^{1-\ep})}{1-\ep}-\rho_0$.
Applying the exterior derivative to $\chi$ yields the outward pointing co-normal
\begin{equation} \label{n-def}
    n =  \frac{1}{|x|}x_\Lambda dx^\Lambda-\frac{\rho_1}{t^\ep}dt
\end{equation}
to $\Gamma_{\qv}$.
The boundary of $\Omega_{\qv}$ can be decomposed as the disjoint union
\begin{equation} \label{Omega-bndry}
    \del{}\Omega_{\qv}= \bigl(\{t_0\}\times\mathbb{B}_{\rho_0}\bigr) \cup \Gamma_{\qv}\cup \bigl(\{t_1\}\times\mathbb{B}_{\tilde{\rho}_1}\bigr)
\end{equation}
where $\tilde{\rho}_1 =\rho_0- \frac{\rho_1 t_0^{1-\ep}}{1-\ep}$
and the balls $\{t_0\}\times\mathbb{B}_{\rho_0}$ and
$\{t_1\}\times\mathbb{B}_{\tilde{\rho}_1}$ cap $\Omega_{\qv}$ on top and bottom, respectively.  Additionally, we let
\begin{equation} \label{Omega-grave-def}
\grave{\Omega}_{\qv_1} = \Bigl\{\, (t,x)\in (t_1,t_0]\times (-L,L)^{3}\, \Bigl| \,|x|\leq \frac{\rho_1(t^{1-\ep}-t_0^{1-\ep})}{1-\ep}+\rho_0 \,  \Bigr\}
\end{equation}
denote the cone domain with closed spatial slices. 

\subsection{Sobolev spaces, extension operators and order notation\label{sec:Sobolev}}
The $W^{k,p}$, $k\in \Zbb_{\geq 0}$, norm of a map $u\in C^\infty(U,\Rbb^N)$ with $U\subset \Tbb^{3}$ open is defined by
\begin{equation*}
\norm{u}_{W^{k,p}(U)} = \begin{cases} \begin{displaystyle}\biggl( \sum_{0\leq |\bc|\leq k} \int_U |\del{}^{\bc} u|^p \, d^{3} x\biggl)^{\frac{1}{p}}  \end{displaystyle} & \text{if $1\leq p < \infty $} \\
 \begin{displaystyle} \max_{0\leq \ell \leq k}\sup_{x\in U}|\del{}^{\ell} u(x)|  \end{displaystyle} & \text{if $p=\infty$}
\end{cases}.
\end{equation*}
The Sobolev space $W^{k,p}(U,\Rbb^N)$ is then defined to be the completion of $C^\infty(U,\Rbb^N)$ with respect to the norm
$\norm{\cdot}_{W^{k,p}(U)}$. When $N=1$ or the dimension $N$ is clear from the context, we simplify notation and write $W^{k,p}(U)$ instead of $W^{k,p}(U,\Rbb^N)$, and we employ the standard notation $H^k(U,\Rbb^N)=W^{k,2}(U,\Rbb^N)$ throughout.

To each centred ball $\mathbb{B}_\rho\subset \Tbb^{3}$, $0<\rho<L$, we
assign a (non-unique) total
extension operator
\begin{equation} \label{Ebb-def}
\Ebb_\rho \: : \: H^k(\mathbb{B}_\rho,\Rbb^N)\longrightarrow H^k(\Tbb^{3},\Rbb^N), \qquad k \in \Zbb_{\geq 0},
\end{equation}
that satisfies
\begin{equation}\label{Ebb-prop}
\Ebb_\rho(u)\bigl|_{\mathbb{B}_\rho} = u, \AND
\norm{\Ebb_{\rho}(u)}_{H^{k}(\Tbb^{3)}} \leq C\norm{u}_{H^k(\mathbb{B}_\rho)}
\end{equation}
for some constant $C=C(k,\rho)>0$ independent of $u\in H^k(\mathbb{B}_\rho)$. The existence of such an operator is established in
\cite{AdamsFournier:2003}; see Theorems 5.21 and 5.22, and Remark 5.23 for details.

For $t_0>0$, $U$ an open set in  $\Tbb^3$, and $g\in C^0((0,t_0],\Rbb_{>0})$, we say that a time dependent function $f\in C^0((0,t_0],H^k(U))$ is \textit{order $g(t)$} and write $f=\Ord_{H^k(U)}(g(t))$ if there exists a constant $C>0$ such that 
\begin{equation*} 
\norm{f(t)}_{H^k(U)}\leq C g(t), \quad \forall\, t\in (0,t_0].
\end{equation*}
The following lemma provides a simple but useful result that will be employed repeatedly in this article. 

\begin{lem}\label{lem:asymptotic}
Suppose  $t_0>0$, $U$ is an open set in  $\Tbb^3$, $g\in C^0((0,t_0],\Rbb_{>0})$, and
$f\in C^1((0,t_0],H^k(U))$. If $\del{t}f=\Ord_{H^k(U)}(g(t))$ and $\int_0^{t_0} g(s)\,ds< \infty$, then there exists a unique $F \in C^0([0,t_0],H^k(U))$ such that $F(t)=f(t)$ for all $t\in (0,t_0]$, and 
$F(t)=F(0)+\Ord_{H^k(U)}(h(t))$ where $h(t)=\int_0^t g(s)\, ds$.
\end{lem}
\begin{proof}
The existence of the extension $F\in C^0([0,t_0],H^k(U))$
is an immediate consequence of \cite[\S5.9.2 Thm.~2]{Evans:2010} as is the representation
\begin{equation*}
F(t)=F(\tilde{t}) + \int_{\tilde{t}}^t \del{s}f(s)\, ds, \quad 0\leq \tilde{t}\leq t \leq t_0,
\end{equation*}
together with corresponding representation formulas for all spatial derivatives of $F$ of order less or equal than $k$.
Setting $\tilde{t}=0$, we find with the help of the triangle inequality that 
\begin{equation*}
\norm{F(t)-F(0)}_{H^k(U)}\leq \int_0^t \norm{\del{s}f(s)}_{H^k(U)}\, ds \leq C  \int_0^t g(s)\, ds, \quad \forall\, t\in (0,t_0],
\end{equation*}
which completes the proof. 
\end{proof}

\subsection{Constants and inequalities}
We use the standard notation $a \lesssim b$
for inequalities of the form
$a \leq Cb$
in situations where the precise value or dependence on other quantities of the constant $C$ is not required.
On the other hand, when the dependence of the constant needs to be specified, for
example if the constant depends on the norm $\norm{u}_{L^\infty}$, we use the notation
$C=C(\norm{u}_{L^\infty})$.
Constants of this type will always be non-negative, non-decreasing, continuous functions of their arguments.

\section{Tetrad formulation of the conformal Einstein-scalar field equations\label{sec:tetrad}}
Our approach to establishing the stability of the Kasner solutions and their big bang singularities relies on the tetrad (orthonormal frame) formulation of the conformal Einstein equations introduced in \cite{RohrUggla:2005}, combined with specific gauge choices described in detail below. We also recall the closely related tetrad formulations of the Einstein equations from \cite{GarfinkleGundlach:2005,Uggla_et_al:2003,vanElstUggla:1997}.

\subsection{Conformal metric and orthonormal frame}
Following \cite{RohrUggla:2005}, we employ a conformal metric $\gt=\gt_{\mu\nu}dx^\mu\otimes dx^\nu$ defined via
\begin{equation}\label{conf-metricA}
\gb = e^{2\Phi}\gt,
\end{equation}
where the conformal scalar $\Phi$ will be fixed below. We also introduce an orthonormal frame $\et_a = \et^\mu_a \del{\mu}$
for the conformal metric where the associated frame metric
\begin{equation} \label{f-metric} 
\eta_{ab} = \gt(\et_a,\et_b) = \gt_{\mu\nu}\et_a^\mu \et^\nu_b
\end{equation}
is determined by
\begin{equation*}
\eta_{ab} = -\delta_a^0\delta_b^0 + \delta_a^A \delta_{AB}\delta^B_b.
\end{equation*}

We partially fix the coordinates by assuming that relative to the foliation determined by the coordinate time $t=x^0$ the shift of the conformal metric vanishes. Using $\alphat$ to denote the lapse, the $3+1$ decomposition of the conformal metric is then given by
\begin{equation} \label{conformal-metric}
\gt= - \alphat^2 dt\otimes dt + \gt_{\Sigma \Omega}dx^\Sigma \otimes dx^\Omega.
\end{equation}
Next, we fix the frame vector $\et_0$ setting it equal to the future pointing normal vector to the $t$-foliation, that is, 
\begin{equation}\label{et0-def}
    \et_0 = \frac{1}{\alphat}\del{t}.
\end{equation}
To fix the remaining freedom in the spatial vector fields $e_A$, we propagate them using Fermi-Walker transport\footnote{This is consistent with our orthonormal frame assumption because Fermi-Walker transport preserves orthonormality of frames.} defined by
\begin{equation*}
    \nablat_{\et_0}\et_A = -\frac{\gt(\nablat_{\et_0}\et_0, \et_A)}{\gt(\et_0, \et_0)}\et_0
\end{equation*}
where $\nablat$ is the Levi-Civita connection of the conformal metric $\gt$. 
Since the shift vanishes and the frame is orthonormal, the spatial frame vectors are tangent to the $t=\text{constant}$ hypersurfaces and can be expressed as
\begin{equation}\label{etA0=0}
\et_A =\et^\Omega_A \del{\Omega}.
\end{equation}
This allows us to express the coordinate components of the (conformal) spatial metric as
\begin{equation} \label{spatial-conformal-metric}
(\gt_{\Sigma\Omega})=\bigl(\delta^{AB}\et_A^\Sigma\et_B^\Omega\bigr)^{-1}.
\end{equation}

\begin{rem} \label{rem:frame-components}
For the remainder of this article, all tensorial frame components, including for physical tensors, will be relative to the (conformal) orthonormal frame $\et_{a}$. For example, for the physical stress energy tensor $\Tb = \Tb_{\mu\nu}dx^\mu \otimes dx^\nu$, see \eqref{stress-energy}, $\Tb_{ab}=\Tb(\et_a,\et_b)=\Tb_{\mu\nu}\et_a^\mu\et^\nu_b$ are the spacetime frame components while the spatial frame components are given by $T_{AB}=\Tb(\et_A,\et_B)=\Tb_{\Lambda\Sigma}\et_A^\Lambda \et_B^\Sigma$.
\end{rem}

\subsection{Connection and commutator coefficients}
The connection coefficients $\omega_a{}^c{}_b$ of the Levi-Civita connection $\nablat$ of the conformal metric $\gt$ are defined in the usual manner by
\begin{equation*}
\nablat_{\et_a}\et_b = \omega_a{}^c{}_b \et_c.
\end{equation*}
On the other hand, the Lie brackets $[\et_a,\et_b]$ define the commutator coefficients $c_a{}^c{}_b$ via
\begin{equation*}
[\et_a,\et_b]=c_a{}^c{}_b \et_c.
\end{equation*}
Since the frame $\et_a$ is orthonormal, the connection and commutator coefficients uniquely determine each other through the relations
\begin{equation}\label{coeff-relations}
c_a{}^c{}_b=\omega_a{}^c{}_b-\omega_b{}^c{}_a \AND \omega_{abc}= \frac{1}{2}(c_{bac}-c_{cba}-c_{acb}).
\end{equation}

\begin{rem}
Here and in the following, we use the conformal frame metric $\eta_{ab}$ to raise and lower frame indices, eg. $\omega_{acb}=\eta_{cd}\omega_a{}^d{}_b$ and  $\omega_a{}^c{}_b = \eta^{cd}\omega_{adb}$. As a consequence of this convention, all spatial frame indices are raised and lowered with the Euclidean frame metric $\delta_{AB}$, eg. $S_{AB}=\delta_{AC}\delta_{BD}S^{CD}$ and $S_{A}{}^{B}=\delta^{BC}S_{AB}$. 
\end{rem}

By \cite[Eqns.~(25)-(26)]{RohrUggla:2005} and our choice of a zero-shift coordinate gauge and Fermi-Walker transported spatial frame\footnote{This gauge and frame choice correspond to setting $M_i=0$ and $\Mcal_\alpha=W_\alpha=R_\alpha=0$ in the formulas from \cite{RohrUggla:2005}.}, we can express the commutators $[\et_0,\et_A]$ and $[\et_A,\et_B]$ as
\begin{align}
[\et_0,\et_A]&= \Ut_A \et_0-(\Hct\delta^B_A+\Sigmat_A{}^B)\et_B,\label{comm-decomp-0} \\
[\et_A,\et_B]& = (2 \At_{[A}\delta_{B]}^C+\ep_{ABD}\Nt^{DC})\et_C, \label{comm-decomp-A}
\end{align}
where
$\Nt_{AB}$ is symmetric and $\Sigmat_{AB}$ is symmetric and trace free, that is,
\begin{equation}\label{index-sym}
\Sigmat_{\expval{AB}}=\Sigmat_{AB} \AND \Nt_{(AB)}=\Nt_{AB}.
\end{equation}
Here, $\ep_{ABC}$ is the completely anti-symmetric symbol, ie. $\ep_{\sigma(1)\sigma(2)\sigma(3)}=\text{sgn}(\sigma)$ where $\sigma$ is any permutation of $\{1,2,3\}$. 

Now, the commutators \eqref{comm-decomp-0}-\eqref{comm-decomp-A} completely determine the commutation coefficients $c_i{}^k{}_j$, and these, in turn, determine the connection coefficients $\omega_i{}^k{}_j$ by \eqref{coeff-relations}. From this observation and \eqref{comm-decomp-0}-\eqref{comm-decomp-A}, it is not difficult to verify\footnote{Alternatively, these formulas follow from setting $R_\alpha=W_\alpha$ in \cite[Eqns.~(35)-(36)]{RohrUggla:2005}. }  that 
\begin{equation} \label{connect-form}
\begin{gathered} 
\omega_{0A0}=-\omega_{00A}=\Ut_A, \qquad \omega_{A0B}=-\omega_{AB0}=-(\Hct\delta_{AB}+\Sigmat_{AB}), \\
\omega_{ABC}= 2\At_{[B}\delta_{C]A}+\ep_{BCD}\Nh^D{}_A,\qquad  \Nh^D{}_A = \Nt^D{}_A-\frac{1}{2}\Nt^E{}_E \delta^D_A, \\
\omega_{000}=\omega_{A00}=\omega_{0BC}=0.
\end{gathered}
\end{equation}

\subsection{Conformal Einstein field equations}
By \cite[Eqn.~(15)]{RohrUggla:2005} or \cite[Eqn.~(1.4)]{BeyerOliynyk:2024b}, the Einstein equations \eqref{ESF.3}, when expressed in terms of the orthonormal frame $\et_a$ and the conformal metric $\gt$, 
become
\begin{equation} \label{cEEa}
\Rt_{ab}-\Upsilon_{ab}=\Tb_{ab}-\frac{1}{2}\Tb_{c}{}^{c}\eta_{ab}
\end{equation}
where $\Rt_{ab}$ is the Ricci curvature tensor of the conformal metric $\gt$, $\Tb_c{}^c = \eta^{cd}\Tb_{cd}$,
\begin{align}
\Upsilon_{ab} &= 2\nablat_{(a}\psi_{b)}-2 \psi_a \psi_b+\eta_{ab}(\nablat_c \psi^c + 2 \psi_c \psi^c), \label{Upsilon-def}\\
\psi_a &:=\nablat_a\Phi = \et_a(\Phi) \label{ri-def}
\intertext{and}
\Tb_{ab} &= 2 \et_a(\phi)\et_b(\phi)-\eta^{cd}\et_c(\phi)\et_d(\phi)\eta_{ab}. \label{Tbij-form}
\end{align}
For use below, we note from \cite[Eqns.~(57)-(60)]{RohrUggla:2005} that the components of the symmetric $2$-tensor \eqref{Upsilon-def} can be calculated using
\begin{align}
\Upsilon_{00}&= 3\et_0(\psi_0)+3\Hct \psi_0 -\bigl(\et_A(\psi^A)+2 \psi_A \psi^A +3 \Ut^A \psi_A -2 \At^A \psi_A\bigr), \label{Upsilon00}\\
\Upsilon_{0A}&= 2\et_A(\psi_0)-2\psi_A \psi_0-2(\Hct \psi_A+\Sigmat_A{}^B \psi_B), \label{Upsilon0A}\\
\Upsilon_{\expval{AB}}&=-2\Sigmat_{AB}\psi_0 +2\et_{\langle A}\psi_{B\rangle}+ 2 \At_{\langle A}\psi_{B\rangle}+2\epsilon^C{}_{D\langle A} \Nt^D{}_{B\rangle} \psi_C-2 \psi_{\langle A}\psi_{B\rangle}, \label{UpsilonAB-tr} \\
\Upsilon_{00}+\Upsilon_A{}^A &= 
-6(2\Hct+\psi_0)\psi_0+4\et_A(\psi^A)+2 \psi_A \psi^A-8\At^A \psi_A. \label{Upsilon00+AA}
\end{align}

By \cite[Eqns.~(41)-(60)]{RohrUggla:2005}, the conformal Einstein equations \eqref{cEEa} imply the following equations for the tetrad variables\footnote{Since our zero-shift coordinate gauge and Fermi-Walker transported spatial frame corresponds to $M_i=0$ and $\Mcal_\alpha=W_\alpha=R_\alpha=0$ in the article \cite{RohrUggla:2005}, \cite[Eqns.~(42) \& (53)]{RohrUggla:2005} are automatically satisfied. Moreover, \cite[Eqns.~(41) \& (46)]{RohrUggla:2005}, which are viewed there as evolution equations, become the constraint equation \eqref{D-cnstr-1} and \eqref{B-cnstr-1}, respectively, here.} $\{\alphat,\Hct,\et^\Sigma_A,\At_A,\Ut_A,\Sigmat_{AB},\Nt_{AB}\}$, which we separate into evolution and constraint equations\footnote{There is a sign error in the term $\epsilon^{\gamma\delta}{}_{\langle \alpha}(2\Sigma_{\beta\rangle \gamma} R_\delta - N_{\beta\rangle \gamma}\dot{U}_\delta)$ from \cite[Eqn.~(45)]{RohrUggla:2005}. The `$-$' sign in front of this term should be replaced with a `$+$' sign; cf. \cite[Eqn.~(34)]{vanElstUggla:1997}.} as follows: 

\subsubsection*{Evolution equations}

\begin{align}
\et_{0}(\et_{A}^{\Omega}) &= -(\Hct\delta_{A}^{B} + \Sigmat_{A}{}^{B})\et_{B}^{\Omega}, \label{EEa.1}\\
\et_{0}(\At_{A}) &= -\et_{A}(\Hct) + \frac{1}{2}\et_{B}(\Sigmat_{A}{}^{B}) - \Hct(\Ut_{A}+\At_{A}) + \Sigmat_A{}^B\Bigl(\frac{1}{2}\Ut_{B}-\At_{B}\Bigr), \label{EEa.2}\\
\et_{0}(\Nt^{AB}) &= -\Hct\Nt^{AB} + 2\Nt^{(A}{}_{C}\Sigmat^{B)C} - \epsilon^{CD(A}\et_{C}(\Sigmat_D{}^{B)})-\Ut_{C}\epsilon^{CD(A}\Sigmat_D{}^{B)},\label{EEa.3} \\
\et_{0}(\Hct) &= -\Hct^{2}+\frac{1}{3}\et_{A}(\Ut^{A}) + \frac{1}{3}\Ut_{A}(\Ut^{A}-2\At^{A})-\frac{1}{3}\Sigmat_{AB}\Sigmat^{AB}\nonumber \\
&\quad-\frac{1}{6}(\Tb_{00}+\Tb_A{}^A)-\frac{1}{3}\Upsilon_{00}, \label{EEa.4}\\
\et_{0}(\Sigmat_{AB}) &= -3\Hct\Sigmat_{AB} +\et_{\langle A}(\Ut_{B\rangle}) +\Ut_{\langle A}\Ut_{B\rangle} -\et_{\langle A}(\At_{B\rangle})+\At_{\langle A}\Ut_{B\rangle} \nonumber \\
&\quad+\epsilon_{CD\langle A}\et^{C}(\Nt_{B\rangle }{}^D)+\epsilon_{CD\langle A}\Nt_{B\rangle}{}^D\Ut^{C}-2\epsilon_{CD\langle A}\Nt_{B\rangle}{}^D\At^{C} \nonumber \\
&\quad -2\Nt_{\langle A}{}^C\Nt_{B\rangle C} +\Nt_C{}^C\Nt_{\expval{AB}} 
+\Tb_{\expval{AB}} + \Upsilon_{\expval{AB}}. \label{EEa.5}
\end{align}

\subsubsection*{Constraint equations}
\begin{align}
\Af_{AB}^\Omega  &:= 2\et_{[A}(\et_{B]}^{\Omega}) - 2\At_{[A}\et^{\Omega}_{B]} - \epsilon_{ABD}\Nt^{DC}\et^{\Omega}_{C}=0, \label{A-cnstr-1}\\
\Bf_{AB}  &:= -2\et_{[A}(\Ut_{B]}) + 2\At_{[A}\Ut_{B]} + \epsilon_{ABC}\Nt^{CD}\Ut_{D}=0, \label{B-cnstr-1}\\
\Cf^A  &:= \et_{B}(\Nt^{BA}) - \epsilon^{BAC}\et_{B}(A_{C}) - 2A_{B}N^{BA}=0, \label{C-cnstr-1}\\
\Df_A  &:= \Ut_{A} - \frac{1}{\alphat}\et_{A}(\alphat)=0, \label{D-cnstr-1} \\
\Mf_A  &:= \et_{B}(\Sigmat_A{}^B) -2\et_{A}(\Hct) -3\Sigmat_A{}^B \At_{B} - \epsilon_{ABC}\Nt^{BD}\Sigmat_D{}^C - \Tb_{0A} -\Upsilon_{0A}=0, \label{M-cnstr-1}\\
\Hf  &:= 4\et_{A}(\At^{A}) +6\Hct^{2} - 6\At^{A}\At_{A} -\Nt^{AB}\Nt_{AB} +\frac{1}{2}(\Nt_B{}^B)^{2} -\Sigmat_{AB}\Sigmat^{AB} \nonumber \\
&\quad -2\Tb_{00} -(\Upsilon_{00}+\Upsilon_A{}^A)=0. \label{H-cnstr-1}
\end{align}
Equations \eqref{M-cnstr-1} and \eqref{H-cnstr-1} are the Momentum and Hamiltonian constraints, respectively. 

\subsection{Scalar field equation}
Following \cite{BeyerOliynyk:2024b} while restricting to $n=4$, we replace the scalar field $\phi$ in favour of the scalar field $\tau$, cf.~\cite[Eqns.~(1.13)-(1.14)]{BeyerOliynyk:2024b}, defined by
\begin{equation}\label{tau-def}
\tau = e^{\frac{2}{\sqrt{3}}\phi}\quad \Longleftrightarrow \quad \phi = \frac{\sqrt{3}}{2}\ln(\tau),
\end{equation}
and we fix the conformal scalar $\Phi$, cf.~\cite[Eqns.~(1.11)-(1.14)]{BeyerOliynyk:2024b}, by setting
\begin{equation} \label{Phi-fix}
\Phi = \frac{1}{2}\ln(\tau).
\end{equation}
Substituting this into \eqref{conf-metricA}
yields
\begin{equation}\label{conf-metricB}
\gb = \tau \gt.
\end{equation}
Moreover, from the calculations carried out in \cite[\S 1.1]{BeyerOliynyk:2024b}, we know that the scalar field equation \eqref{ESF.2}, when expressed in terms of the conformal metric \eqref{conf-metricB} and scalar $\tau$, becomes
\begin{equation}\label{SFa}
\Box_{\gt} \tau =0,
\end{equation}
and the pair $\{\gt,\tau\}$ satisfies the \textit{conformal Einstein-scalar field equations} 
\begin{align}
\Gt_{\mu\nu} &= \frac{1}{\tau}\nablat_\mu\nablat_\nu \tau, \label{cESF.1}\\
\Box_{\gt} \tau &=0. \label{cESF.2}
\end{align}

\subsection{Fixing the lapse}
The last remaining gauge freedom is to determine the lapse. Here, we again follow \cite{BeyerOliynyk:2024b} and fix the lapse by using the scalar field $\tau$ as our time coordinate, that is, we set 
\begin{equation}\label{t-fix}
\tau = t.
\end{equation}
As discussed in \cite{BeyerOliynyk:2024b}, the importance of this choice of time coordinate is that it synchronises the big bang singularity. Now, by \eqref{SFa}, we have
\begin{equation*}
\Box_{\gt}t =0 \quad \Longleftrightarrow\quad \eta^{ab}\bigl(\et_a(\et_b(t))-\omega_a{}^c{}_b e_c(t) \bigr) =0,
\end{equation*}
or in other words, the time coordinate $t$ is harmonic with respect to the conformal metric. 
Using \eqref{et0-def}-\eqref{etA0=0} and \eqref{connect-form}, we see that this time slicing condition implies that the lapse $\alphat$ is determined via the evolution equation  
\begin{equation} \label{alphat-ev}
\et_0(\alphat) = 3\Hct \alphat. 
\end{equation}

\subsection{An evolution equation for $\Ut_A$}
In order to close our system of equations, we need an evolution equation for $\Ut_A$. To derive it, we divide \eqref{alphat-ev} by $\alphat$ and apply $\et_A$ to the resulting expression to get
\begin{align}
3\et_A(\Hct) & = \et_A(\et_0(\ln(\alphat))) \notag \\
&= [\et_A,\et_0](\ln(\alphat)) + \et_0(\et_A(\ln(\alphat))) \notag \\
&=-\frac{1}{\alphat}\bigl(\Ut_A \et_0(\alphat)-(\Hct \delta^B_A + \Sigmat_A{}^B)\et_B(\alphat)\bigr)+
\et_0\biggl(\frac{1}{\alphat}\et_A(\alphat)\biggr) && \text{(by \eqref{comm-decomp-0})} \notag \\ 
&=-\bigl(\Ut_A 3\Hct \alphat-(\Hct \delta^B_A + \Sigmat_A{}^B)\Ut_B\bigr)+
\et_0(\Ut_A) && \text{(by \eqref{D-cnstr-1} \& \eqref{alphat-ev})}.\notag 
\end{align}
Rearranging yields the desired evolution equation:
\begin{equation}\label{Ut-ev}
\et_0(\Ut_A) = 3\et_A(\Hct)+2\Hct \Ut_A-\Sigmat_A{}^B\Ut_B.
\end{equation}

\subsection{Complete system of evolution and constraint equations\label{sec:complete-evolve-constraint}}
From \eqref{ri-def}-\eqref{Tbij-form}, \eqref{Upsilon00}-\eqref{Upsilon00+AA}, \eqref{tau-def}-\eqref{Phi-fix} and \eqref{t-fix}, we observe that the components of $\Tb_{ab}$ and $\Upsilon_{ab}$ are determined by
\begin{equation}\label{Tb-Upsilon-cmpts}
\begin{gathered} 
    \Tb_{00} = \frac{3}{4\alphat^2 t^2}, \quad \Tb_{0A} = 0, \quad \Tb_A{}^A = \frac{9}{4\alphat^2 t^2}, \quad \Tb_{\expval{AB}} = 0, \\
    \Upsilon_{00} = -\frac{3\Hct}{\alphat t} - \frac{3}{2\alphat^2 t^2}, \quad \Upsilon_{0A} = -\frac{1}{\alphat t}\Ut_A, \quad \Upsilon_A{}^A = -\frac{3\Hct}{\alphat t}, \quad \Upsilon_{\expval{AB}} = -\frac{1}{\alphat t}\Sigmat_{AB}.
\end{gathered}
\end{equation}
Substituting \eqref{Tb-Upsilon-cmpts} into the evolution \eqref{EEa.1}-\eqref{EEa.5} and constraint \eqref{A-cnstr-1}-\eqref{H-cnstr-1} equations and including \eqref{alphat-ev} and \eqref{Ut-ev} with the evolutions equations, we arrive, after using \eqref{et0-def} to express $\et_0$ in terms of the time derivative $\del{t}$, at the following tetrad formulation of the conformal Einstein-scalar field equations:
\subsubsection*{Evolution equations}
\begin{align}
\del{t}\et_{A}^{\Omega} &= -\alphat(\Hct\delta_{A}^{B} + \Sigmat_{A}{}^{B})\et_{B}^{\Omega}, \label{EEb.1}\\
\del{t}\At_{A} &= -\alphat \et_{A}(\Hct) + \frac{1}{2}\alphat\et_{B}(\Sigmat_{A}{}^{B}) - \alphat\Hct(\Ut_{A}+\At_{A}) + \alphat\Sigmat_A{}^B\Bigl(\frac{1}{2}\Ut_{B}-\At_{B}\Bigr), \label{EEb.2}\\
\del{t}\Nt^{AB} &= -\alphat\Hct\Nt^{AB} + 2\alphat\Nt^{(A}{}_{C}\Sigmat^{B)C} - \alphat\epsilon^{CD(A}\et_{C}(\Sigmat_D{}^{B)})-\alphat\Ut_{C}\epsilon^{CD(A}\Sigmat_D{}^{B)},\label{EEb.3} \\
\del{t}\Hct &= -\alphat\Hct^{2}+\frac{1}{3}\alphat\et_{A}(\Ut^{A}) + \frac{1}{3}\alphat\Ut_{A}(\Ut^{A}-2\At^{A})-\frac{1}{3}\alphat\Sigmat_{AB}\Sigmat^{AB}+\frac{1}{t}\Hct, \label{EEb.4}\\
\del{t}\Sigmat_{AB} &= -3\alphat\Hct\Sigmat_{AB} +\alphat\et_{\langle A}(\Ut_{B\rangle}) +\alphat\Ut_{\langle A}\Ut_{B\rangle} -\alphat\et_{\langle A}(\At_{B\rangle})+\alphat\At_{\langle A}\Ut_{B\rangle} \nonumber \\
&\quad+\alphat\epsilon_{CD\langle A}\et^{C}(\Nt_{B\rangle }{}^D)+\alphat\epsilon_{CD\langle A}\Nt_{B\rangle}{}^D\Ut^{C}-2\alphat\epsilon_{CD\langle A}\Nt_{B\rangle}{}^D\At^{C} \nonumber \\
&\quad -2\alphat\Nt_{\langle A}{}^C\Nt_{B\rangle C} +\alphat\Nt_C{}^C\Nt_{\expval{AB}} -\frac{1}{t}\Sigmat_{AB}, \label{EEb.5} \\
\del{t}\alphat &= 3\Hct \alphat^2, \label{EEb.6}\\
\del{t}\Ut_A &= 3\alphat\et_A(\Hct)+2\alphat\Hct \Ut_A-\alphat\Sigmat_A{}^B\Ut_B. \label{EEb.7}
\end{align}

\subsubsection*{Constraint equations}
\begin{align}
\tilde{\Af}_{AB}^\Omega  &:= 2\et_{[A}(\et_{B]}^{\Omega}) - 2\At_{[A}\et^{\Omega}_{B]} - \epsilon_{ABD}\Nt^{DC}\et^{\Omega}_{C}=0, \label{A-cnstr-2}\\
\tilde{\Bf}_{AB}  &:= -2\et_{[A}(\Ut_{B]}) + 2\At_{[A}\Ut_{B]} + \epsilon_{ABC}\Nt^{CD}\Ut_{D}=0, \label{B-cnstr-2}\\
\tilde{\Cf}^A  &:= \et_{B}(\Nt^{BA}) - \epsilon^{BAC}\et_{B}(A_{C}) - 2A_{B}N^{BA}=0, \label{C-cnstr-2}\\
\tilde{\Df}_A  &:= \Ut_{A} - \frac{1}{\alphat}\et_{A}(\alphat)=0, \label{D-cnstr-2} \\
\tilde{\Mf}_A  &:= \et_{B}(\Sigmat_A{}^B) -2\et_{A}(\Hct) -3\Sigmat_A{}^B \At_{B} - \epsilon_{ABC}\Nt^{BD}\Sigmat_D{}^C+\frac{1}{\alphat t}\Ut_A=0, \label{M-cnstr-2}\\
\tilde{\Hf}  &:= 4\et_{A}(\At^{A}) +6\Hct^{2} - 6\At^{A}\At_{A} -\Nt^{AB}\Nt_{AB} +\frac{1}{2}(\Nt_B{}^B)^{2} -\Sigmat_{AB}\Sigmat^{AB} +\frac{6}{\alphat t}\Hct=0. \label{H-cnstr-2}
\end{align}

\begin{rem}\label{rem:constraints}
Initial data sets $(\Sigma, \gttt, \Kttt, \tau_0, \tau_1)$ for the Einstein-scalar field equations consist of a 3-dimensional spatial manifold $\Sigma$, a Riemannian metric $\gttt$ on $\Sigma$, a symmetric 2-tensor $\Kttt$ on $\Sigma$, and two scalar functions $\tau_0$ and $\tau_1$ on $\Sigma$. To synchronise the big bang singularity, we take $\Sigma$ as a level set of $\tau$, ie.~$\Sigma = \tau^{-1}(\tau_0)$, which implies 
\begin{equation}
    \label{eq:synchroID}
    \tau_0 := \tau|_{\Sigma} = t_0.
\end{equation} The remaining initial data are specified as follows: $\gttt$ is the restriction of the conformal metric $\gt$ to $\Sigma$, $\Kttt$ is the second fundamental form of $\Sigma$, and $\tau_1=e_0(\tau)|_{\Sigma}$. 
To solve the Einstein-scalar field equations, the geometric initial data $(\gttt, \Kttt, \tau_0, \tau_1)$ must satisfy the Momentum and Hamiltonian constraints \eqref{M-cnstr-2} and \eqref{H-cnstr-2} on $\Sigma$. Furthermore, for the solutions of the evolution equations \eqref{EEb.1}-\eqref{EEb.7} to generate solutions of the conformal Einstein-scalar field equations, the remaining constraints \eqref{A-cnstr-2}-\eqref{D-cnstr-2} must also be satisfied on $\Sigma$. Importantly, these constraints do not involve the geometric initial data $(\gttt,\Kttt,\tau_0,\tau_1)$ and they can always be satisfied. Furthermore, since we restrict our attention to the past evolution of small perturbations of Kasner-scalar field spacetime initial data, no generality is lost by assuming the synchronisation condition $\tau_0 = t_0$; see \cite[\S 5.7]{BeyerOliynyk:2024b} for details.
\end{rem}

\subsection{Kasner-scalar field solutions}
Before proceeding, we note here that the Kasner-scalar field solutions \eqref{Kasner-solns-A} can be expressed in terms of the tetrad variables  $\{\alphat,\Hct,\et^\Sigma_A,\At_A,\Ut_A,\Sigmat_{AB},\Nt_{AB}\}$ as follows:
\begin{equation}\label{Kasner-solns-B}
\begin{gathered}
    \mathring{\et}_A^\Omega = t^{-r_A/2}\delta^\Omega_A, \quad  \mathring{\At}_A = 0, \quad \mathring{\Nt}^{AB} = 0, \quad \mathring{\Hct} = \frac{r_0}{6}t^{-r_0/2-1}, \\
    \mathring{\Sigmat}_{AB} = \biggl(\frac{r_{AB}}{2}-\frac{r_0\delta_{AB}}{6}\biggr)t^{-r_0/2-1}, \quad \mathring{\alphat} = t^{r_0/2}, \quad \mathring{\Ut}_A = 0,
\end{gathered}
\end{equation}
where
\begin{equation} \label{rAB-def}
    r_{AB} = \begin{cases}
                 r_A, &\text{ if } A=B  \\ 
                  0,  &\text{ if } A \neq B 
             \end{cases}.
\end{equation}

\section{Local-in-time existence, uniqueness, continuation and constraints propagation\label{sec:local-existence-theory}}
We now focus on establishing the local-in-time existence and uniqueness of solutions to a modified version of the evolution equations \eqref{EEb.1}-\eqref{EEb.7}, obtained by adding multiples of the constraints. Additionally, we establish a continuation principle for the solutions and verify that the constraints \eqref{A-cnstr-2}-\eqref{H-cnstr-2} propagate.

\subsection{A symmetric hyperbolic formulation of the evolution equations\label{sec:sym-hyp}}
The modified system used to establish the existence of local-in-time solutions is obtained by adding the constraints $\rho\alphat\tilde{\Mf}_A$, $-\frac{1}{3}\alphat\tilde{\Hf}$, and $\gamma\alphat\tilde{\Mf}_A$ to \eqref{EEb.2}, \eqref{EEb.4}, and \eqref{EEb.7}, respectively. Here, $\rho$ and $\gamma$ are arbitrary constants, while $\tilde{\Mf}_A$ and $\tilde{\Hf}$ are the Momentum and Hamiltonian constraints, defined above by \eqref{M-cnstr-2} and \eqref{H-cnstr-2}, respectively. This leads to the following system of evolution equations:
\begin{align}
    \del{t}\et_{A}^{\Omega} &= -\alphat(\Hct\delta_{A}^{B} + \Sigmat_{A}{}^{B})\et_{B}^{\Omega}, \label{mEEb.1} \\
    \del{t}\At_{A} &= -(1+2\rho)\alphat \et_{A}(\Hct) + \biggl(\frac{1}{2}+\rho\biggr)\alphat\et_{B}(\Sigmat_{A}{}^{B}) + \alphat\Sigmat_A{}^B\Bigl(\frac{1}{2}\Ut_{B}-(1+3\rho)\At_{B}\Bigr) \notag \\
    &\quad - \alphat\Hct(\Ut_{A}+\At_{A}) - \rho\alphat\ep_{ABC}\Nt^{BD}\Sigmat_D{}^C + \frac{\rho}{t}\Ut_A, \label{mEEb.2} \\
    \del{t}\Nt^{AB} &= -\alphat\Hct\Nt^{AB} + 2\alphat\Nt^{(A}{}_{C}\Sigmat^{B)C} - \alphat\epsilon^{CD(A}\et_{C}(\Sigmat_D{}^{B)})-\alphat\Ut_{C}\epsilon^{CD(A}\Sigmat_D{}^{B)}, \label{mEEb.3}\\
    \del{t}\Hct &= -\frac{4}{3}\alphat\et_A(\At^A) + \frac{1}{3}\alphat\et_{A}(\Ut^{A}) - 3\alphat\Hct^{2}  + \frac{1}{3}\alphat\Ut_{A}(\Ut^{A}-2\At^{A}) + 2\alphat\At^A\At_A \notag \\
    &\quad + \frac{1}{3}\alphat\Nt^{AB}\Nt_{AB} - \frac{1}{6}\alphat(\Nt_B{}^B)^2 - \frac{1}{t}\Hct, \label{mEEb.4}\\
    \del{t}\Sigmat_{AB} &= -3\alphat\Hct\Sigmat_{AB} +\alphat\et_{\langle A}(\Ut_{B\rangle}) +\alphat\Ut_{\langle A}\Ut_{B\rangle} -\alphat\et_{\langle A}(\At_{B\rangle})+\alphat\At_{\langle A}\Ut_{B\rangle} \nonumber \\
    &\quad+\alphat\epsilon_{CD\langle A}\et^{C}(\Nt_{B\rangle }{}^D)+\alphat\epsilon_{CD\langle A}\Nt_{B\rangle}{}^D\Ut^{C}-2\alphat\epsilon_{CD\langle A}\Nt_{B\rangle}{}^D\At^{C} \nonumber \\
    &\quad -2\alphat\Nt_{\langle A}{}^C\Nt_{B\rangle C} +\alphat\Nt_C{}^C\Nt_{\expval{AB}} -\frac{1}{t}\Sigmat_{AB},\label{mEEb.5} \\
    \del{t}\alphat &= 3\Hct \alphat^2,\label{mEEb.6} \\
    \del{t}\Ut_A &= (3-2\gamma)\alphat\et_A(\Hct) + \gamma\alphat\et_B(\Sigmat_A{}^B) + 2\alphat\Hct \Ut_A-\alphat\Sigmat_A{}^B\Ut_B \notag \\
    &\quad - 3\gamma\alphat\Sigmat_A{}^B \At_B - \gamma\alphat\ep_{ABC}\Nt^{BD}\Sigmat_D{}^C + \frac{\gamma}{t}\Ut_A. \label{mEEb.7}
\end{align}

We claim that this system is symmetrisable for specific values of the constants $\rho$ and $\gamma$. Once verified, the local-in-time existence and uniqueness of solutions, as well as the continuation principle, follow directly from the application of standard existence theory for hyperbolic systems. It is worth noting that the idea of obtaining a symmetric hyperbolic formulation of the tetrad equations by adding multiples of the constraints to the tetrad evolution equations originates from \cite{GarfinkleGundlach:2005}. The calculations below can be viewed as a generalisation of this symmetrisation method.

To verify the system \eqref{mEEb.1}-\eqref{mEEb.7} is symmetrisable, we first express it in matrix form as follows:
\begin{equation} \label{EE-matform}
    \del{t}\Wt = \alphat B_{\upsilon}^C \et_C^\Lambda \del{\Lambda}\Wt + \Gt_\upsilon, \quad \upsilon=(\rho,\gamma),
\end{equation}
where
\begin{gather} 
    \Wt = \bigl(\et_P^\Sigma, \alphat, \At_P, \Ut_P, \Hct, \Sigmat_{PQ}, \Nt_{PQ}\bigr)^{\tr}, \label{Wt-def}\\
   B_{\upsilon}^C = \begin{bmatrix}
        0&  0&  0&  0&  0&  0&  0 \\ 
        0&  0&  0&  0&  0&  0&  0 \\ 
        0&  0&  0&  0&  -(1+2\rho)\delta_A^C&  (\frac{1}{2}+\rho)\delta^{C\langle P}\delta_A^{Q\rangle}&  0 \\ 
        0&  0&  0&  0&  (3-2\gamma)\delta_A^C&  \gamma\delta^{C\langle P}\delta_A^{Q\rangle }&  0 \\
        0&  0&  -\frac{4}{3}\delta^{CP}&  \frac{1}{3}\delta^{CP}&  0&  0&  0 \\
        0&  0&  -\delta_{\langle A}^C\delta_{B\rangle} ^P&  \delta_{\langle A}^C\delta_{B\rangle}^P&  0&  0&  \vep^{C}{}_{AB}{}^{PQ} \\
        0&  0&  0&  0&  0& \vep^{CPQ}{}_{AB}&  0
    \end{bmatrix}, \label{B-rho-gamma}\\
\vep^{CPQ}{}_{AB}= -\ep^{C\langle P}{}_{(  A}\delta_{B) }^{Q\rangle},  \label{vep-def}
\intertext{and}
\Gt_\upsilon = (\Gt_1, \Gt_2, \Gt_3, \Gt_4, \Gt_5, \Gt_6, \Gt_7)^{\tr} \label{Gt-def}
\end{gather}
with
\begin{align}
    \Gt_1 &= -\alphat(\Hct\delta_{A}^{B} + \Sigmat_{A}{}^{B})\et_{B}^{\Omega}, \label{Gt1-def} \\
    \Gt_2 &= 3\Hct \alphat^2, \label{Gt2-def} \\
    \Gt_3 &= \alphat\Sigmat_A{}^B\Bigl(\frac{1}{2}\Ut_{B}-(1+3\rho)\At_{B}\Bigr) - \alphat\Hct(\Ut_{A}+\At_{A}) - \rho\alphat\ep_{ABC}\Nt^{BD}\Sigmat_D{}^C + \frac{\rho}{t}\Ut_A, \label{Gt3-def}\\
    \Gt_4 &= 2\alphat\Hct \Ut_A-\alphat\Sigmat_A{}^B\Ut_B - 3\gamma\alphat\Sigmat_A{}^B \At_B - \gamma\alphat\ep_{ABC}\Nt^{BD}\Sigmat_D{}^C + \frac{\gamma}{t}\Ut_A, \label{Gt4-def}\\
    \Gt_5 &= - 3\alphat\Hct^{2}  + \frac{1}{3}\alphat\Ut_{A}(\Ut^{A}-2\At^{A}) + 2\alphat\At^A\At_A + \frac{1}{3}\alphat\Nt^{AB}\Nt_{AB} - \frac{1}{6}\alphat(\Nt_B{}^B)^2 - \frac{1}{t}\Hct, \label{Gt5-def}\\
    \Gt_6 &= -3\alphat\Hct\Sigmat_{AB} + \alphat\Ut_{\langle A}\Ut_{B\rangle} + \alphat\At_{\langle A}\Ut_{B\rangle} + \alphat\epsilon_{CD\langle A}\Nt_{B\rangle}{}^D\Ut^{C}-2\alphat\epsilon_{CD\langle A}\Nt_{B\rangle}{}^D\At^{C} \nonumber \\
    &\quad -2\alphat\Nt_{\langle A}{}^C\Nt_{B\rangle C} +\alphat\Nt_C{}^C\Nt_{\expval{AB}} -\frac{1}{t}\Sigmat_{AB}, \label{Gt6-def} \\
    \Gt_7 &= -\alphat\Hct\Nt^{AB} + 2\alphat\Nt^{(A}{}_{C}\Sigmat^{B)C} - \alphat\Ut_{C}\epsilon^{CD(A}\Sigmat_D{}^{B)}. \label{Gt7-def}
\end{align}

\begin{rem}
Viewing  $\At_A$, $\Ut_A$ as vectors in $\Rbb^3$, and $\et^\Omega_A$, $\Nt_{AB}$, and $\Sigma_{AB}$ as $3\times3$-matrices where $\Sigma_{AB}$ is symmetric and trace free and $\Nt^{AB}$ is symmetric, we note that the vector-valued map $\Wt$ defined by \eqref{Wt-def} takes its values in the vector space
\begin{equation}\label{Wbb-def}
\Wbb = \Mbb{3}\times \Rbb \times \Rbb^3 \times \Rbb^3 \times \Rbb\times \mathbb{S}^{\text{TF}}_{3} \times  \Sbb{3}, 
\end{equation}
where here we are using $\Sbb{3}$ and $\mathbb{S}^{\text{TF}}_{3}$ to denote the set of symmetric and symmetric trace-free $3\times 3$-matrices, respectively, and $\Mbb{3}$ to denote the set of $3\times 3$-matrices. It is clear from the structure of the system \eqref{EE-matform}, in particular, the matrix \eqref{B-rho-gamma} and the nonlinear source terms \eqref{Gt6-def} and \eqref{Gt7-def}, that it is well-defined on the vector space $\Wbb$. For use below, we define an open subset $\Uc\subset \Wbb$ via
\begin{equation}\label{Uc-def}
\Uc = \textrm{GL}^+(3)\times (0,\infty) \times \Rbb^3 \times \Rbb^3 \times \Rbb\times \mathbb{S}^{\text{TF}}_{3} \times  \Sbb{3}
\end{equation}
where $\textrm{GL}^+(3)$ is set of $3\times 3$-matrices with positive determinant, ie.~the identity component of the general linear group on $\Rbb^3$.
\end{rem}

To proceed, we introduce a change of variables via
\begin{equation}\label{Wh-def}
\Wh = V^{-1}\Wt 
\end{equation}
where
\begin{equation} \label{V-def}
    V = \begin{bmatrix}
        \delta^P_A\delta^\Omega_\Sigma  & 0 & 0 & 0 & 0 & 0 & 0\\
        0 & 1 & 0 & 0 & 0 & 0 & 0 \\
        0 & 0 & m\delta^P_A & n\delta^P_A & 0 & 0 & 0 \\
        0 & 0 & p\delta^P_A & q\delta^P_A & 0 & 0 & 0 \\
        0 & 0 & 0 & 0 & 1 & 0 & 0 \\
        0 & 0 & 0 & 0 & 0 & \delta^P_A\delta^Q_B & 0 \\
        0 & 0 & 0 & 0 & 0 & 0 & \delta^P_A\delta^Q_B
    \end{bmatrix}
\end{equation}
and the constants $m,p,n,q$ satisfy $mq-pn \neq 0$ 
so that $V^{-1}$ is well-defined.
Then multiplying \eqref{EE-matform} on the left by 
the matrix
\begin{equation} \label{Sc-def}
    \Sc = \begin{bmatrix}
        \delta^P_A \delta^\Omega_\Sigma  & 0 & 0 & 0 & 0 & 0 & 0\\
            0 & 1 & 0 & 0 & 0 & 0 & 0 \\
            0 & 0 & a\delta^P_A & b\delta^P_A & 0 & 0 & 0 \\
            0 & 0 & c\delta^P_A & d\delta^P_A & 0 & 0 & 0 \\
            0 & 0 & 0 & 0 & h & 0 & 0 \\
            0 & 0 & 0 & 0 & 0 & l\delta^P_A\delta^Q_B & 0 \\
            0 & 0 & 0 & 0 & 0 & 0 & s\delta^P_A\delta^Q_B
    \end{bmatrix}
\end{equation}
with $a,b,c,d\in \Rbb$ and $h,l,s>0$
yields
\begin{equation} \label{EE-symform-A}
    \Sc V \del{t}\Wh = \alphat \Sc B_{\upsilon}^C V  \et_C^\Lambda\del{\Lambda}\Wh + \Sc\Gt_\upsilon
\end{equation}
where
\begin{equation} \label{ScV-def}
    \Sc V = \begin{bmatrix}
        \delta^P_A\delta^\Omega_\Sigma & 0 & 0 & 0 & 0 & 0 & 0\\
        0 & 1 & 0 & 0 & 0 & 0 & 0 \\
        0 & 0 & (am+bp)\delta^P_A & (an+bq)\delta^P_A & 0 & 0 & 0 \\
        0 & 0 & (cm+dp)\delta^P_A & (cn+dq)\delta^P_A & 0 & 0 & 0 \\
        0 & 0 & 0 & 0 & h & 0 & 0 \\
        0 & 0 & 0 & 0 & 0 & l\delta^P_A\delta^Q_B & 0 \\
        0 & 0 & 0 & 0 & 0 & 0 & s\delta^P_A\delta^Q_B 
    \end{bmatrix}
\end{equation}
and
\begin{equation}
    \Sc B_{\upsilon}^C V = \text{\scalebox{0.77}{$\begin{bmatrix}
        0&  0&  0&  0&  0&  0&  0 \\ 
        0&  0&  0&  0&  0&  0&  0 \\ 
        0&  0&  0&  0&  (-a(1+2\rho)+b(3-2\gamma))\delta_A^C&  \Bigl(a\Bigl(\frac{1}{2}+\rho\Bigr)+b\gamma\Bigr)\delta^{C\langle P}\delta_A^{Q\rangle} &  0 \\ 
        0&  0&  0&  0&  (-c(1+2\rho)+d(3-2\gamma))\delta_A^C&  \Bigl(c\Bigl(\frac{1}{2}+\rho\Bigr)+d\gamma\Bigr)\delta^{C\langle P}\delta_A^{Q\rangle}&  0 \\
        0&  0&  \Bigl(-\frac{4}{3}m+\frac{1}{3}p\Bigr)h\delta^{CP}&  \Bigl(-\frac{4}{3}n+\frac{1}{3}q\Bigr)h\delta^{CP}&  0&  0&  0 \\
        0&  0&  (-m+p)l\delta_{\langle A}^C\delta_{B\rangle}^P&  (-n+q)l\delta_{\langle A}^C\delta_{B\rangle}^P&  0&  0&  l\vep^{C}{}_{AB}{}^{PQ} \\
        0&  0&  0&  0&  0&  s\vep^{CPQ}{}_{AB}&  0
    \end{bmatrix}$ }}. \label{SBV-def}
\end{equation}

From \eqref{SBV-def}, it is clear that the matrices $\Sc B_\nu^C V$  will be symmetric provided that the following equations hold: 
\begin{gather*}
-a(1+2\rho)+b(3-2\gamma)=  \biggl(-\frac{4}{3}m+\frac{1}{3}p\biggr)h, \quad
a\biggl(\frac{1}{2}+\rho\biggr)+b\gamma = (-m+p)l, \\
-c(1+2\rho)+d(3-2\gamma) = \biggl(-\frac{4}{3}n+\frac{1}{3}q\biggr)h, \quad 
c\biggl(\frac{1}{2}+\rho\biggr)+d\gamma = (-n+q)l,\AND s=l.
\end{gather*}
Solving these equations for $m,n,p,q,s$, we find that
\begin{align}
    m &= -\frac{-ah-6al+18bl-2bh\gamma-12bl\gamma-2ah\rho-12al\rho}{6hl}, \label{m-fix} \\
    n &= -\frac{-ch-6cl+18dl-2dh\gamma-12dl\gamma-2ch\rho-12cl\rho}{6hl}, \label{n-fix} \\
    p &= -\frac{-2ah-3al+9bl-4bh\gamma-6bl\gamma-4ah\rho-6al\rho}{3hl}, \label{p-fix} \\
    q &= -\frac{-2ch-3cl+9dl-4dh\gamma-6dl\gamma-4ch\rho-6cl\rho}{3hl}, \label{q-fix} \\
    s &= l. \label{s-fix}
\end{align}
Given \eqref{m-fix}-\eqref{q-fix}, we observe from \eqref{ScV-def} that the matrix $\Sc V$ will be symmetric provided that
\begin{equation*}
    (cm+dp) - (an+bq) = \frac{(bc-ad)\big(h(-2+\gamma-4\rho)+6l(-2+\gamma-\rho)\big)}{3hl} = 0.
\end{equation*}
Solving this for $\gamma$ gives
\begin{equation} \label{rho-gamma}
    \gamma = \frac{4h+6l}{h+6l}\rho + 2.
\end{equation}
Moreover, by Sylvester's criterion for symmetric matrices and the fact that $h,l,s>0$, it follows that the matrix $\Sc V$ will be positive definite provided $a m+ bp >0$ and
$\det\begin{bmatrix} am+bp & cm+dp \\ cm +dp & cn +d q\end{bmatrix}>0$. Using \eqref{m-fix}-\eqref{q-fix}, these conditions become
\begin{align} \label{pos-c1}
    am + bp &= \frac{1}{6hl(h+6l)}\Bigl( a^2(h+6l)^2(1+2\rho) + 4ab(2h^2+15hl+18l^2)(1+2\rho) \notag \\
        &\quad + 2b^2\bigl(8h^2(1+2\rho)+18l^2(1+2\rho)+3hl(17+16\rho)\bigr) \Bigr) > 0,
\end{align}
and
\begin{equation} \label{pos-c2}
    (ad-bc)(mq-np) = \frac{3(ad-bc)^2 (1+2\rho)}{2hl} > 0.
\end{equation}
Since $h,l>0$, we can ensure \eqref{pos-c2} holds by requiring that
\begin{equation} \label{pos-def-A}
    1 + 2\rho > 0 \AND ad-bc \neq 0.
\end{equation}
Assuming \eqref{pos-def-A}, we have
\begin{align*}
    am + bp &> \frac{1+2\rho}{6hl(h+6l)}\Bigl( a^2(h+6l)^2 + 4ab(2h^2+15hl+18l^2) + 2b^2(8h^2+18l^2+24hl) \Bigr) \\
        &= \frac{1+2\rho}{6hl(h+6l)} \Bigl( a(h+6l) + 2b(2h+3l) \Bigr)^2 \geq 0,
\end{align*}
and hence, \eqref{pos-c1} holds.

To summarize, the above calculations show that the matrices $\Sc V$ and $S B_\nu^C V$ will be symmetric and $\Sc V$ will be positive definite provided $m,n,p,q,s$ are determined by \eqref{m-fix}-\eqref{s-fix}, $\gamma$ is determined by \eqref{rho-gamma},  $h,l$ are chosen positive, and $\rho$ and $a,b,c,d$ satisfy \eqref{pos-def-A}.

\subsection{A strongly hyperbolic formulation of the constraint propagation equations\label{const-hyp}}
It follows from the evolution equations \eqref{mEEb.1}-\eqref{mEEb.7} and a direct, but very lengthy and tedious calculation, that the constraints, defined by \eqref{A-cnstr-2}-\eqref{H-cnstr-2}, propagate according to 
\begin{align}
    \del{t}\tilde{\Af}_{AB}^\Omega =& -2\alphat\Hct\tilde{\Af}_{AB}^\Omega + \alphat\Sigmat_A{}^C\tilde{\Af}_{BC}^\Omega - \alphat\Sigmat_B{}^C\tilde{\Af}_{AC}^\Omega, \label{ev-A-cnstr} \\
    \del{t}\tilde{\Bf}_{AB} =& -\gamma\alphat\et_A(\tilde{\Mf}_B) + \gamma\alphat\et_B(\tilde{\Mf}_A) + (\rho+\gamma)\alphat(\Ut_B\tilde{\Mf}_A - \Ut_A\tilde{\Mf}_B) + \gamma\alphat(\At_A\tilde{\Mf}_B-\At_B\tilde{\Mf}_A) \notag \\
        & + \gamma\ep_{ABC}\Nt^{CD}\tilde{\Mf}_D + \alphat\Hct\tilde{\Bf}_{AB} + \alphat\Sigmat_A{}^C\tilde{\Bf}_{BC} - \alphat\Sigmat_B{}^C\tilde{\Bf}_{AC}, \label{ev-B-cnstr} \\
    \del{t}\tilde{\Cf}^A =& -\rho\alphat\ep^{BAC}\et_B(\tilde{\Mf}_C) - \rho\alphat\ep^{BAC}\Ut_B\tilde{\Mf}_C - 2\rho\alphat\Nt^{AB}\tilde{\Mf}_B \notag \\
        & + (\frac{1}{2}\alphat\ep^{BCD}\Sigmat_D{}^A + \frac{1}{2}\alphat\ep^{ABC}\Hct)\tilde{\Bf}_{BC} + (\alphat\Sigmat_B{}^A - 2\alphat\Hct\delta_B^A)\tilde{\Cf}^B, \label{ev-C-cnstr} \\
    \del{t}\tilde{\Df}_A =& -\alphat\Hct\tilde{\Df}_A - \alphat\Sigmat_A{}^B\tilde{\Df}_B + \gamma\alphat\tilde{\Mf}_A, \label{ev-D-cnstr} \\
    \del{t}\tilde{\Mf}_A =& -\frac{1}{t}\Hct\tilde{\Df}_A - \frac{1}{t}\Sigmat_A{}^B\tilde{\Df}_B - \frac{1-\gamma}{t}\tilde{\Mf}_A - (4\alphat\Hct\delta_A^B+(1+3\rho)\alphat\Sigmat_A{}^B)\tilde{\Mf}_B + \frac{1}{2}\alphat\et_A(\tilde{\Hf}) + \frac{2}{3}\alphat\Ut_A\tilde{\Hf} \notag \\
        & + \frac{1}{2}\alphat\Nt\tilde{\Cf}_A - \frac{3}{2}\alphat\Nt_{AB}\tilde{\Cf}^B - \frac{3}{2}\alphat\ep_{ABC}\At^B\tilde{\Cf}^C + \frac{1}{2}\alphat\ep_{ABC}\Ut^B\tilde{\Cf}^C + \frac{1}{2}\alphat\ep_{ABC}\et^B(\tilde{\Cf}^C) \notag \\
        & + \frac{1}{2}\alphat\et^B(\tilde{\Bf}_{AB}) + \frac{1}{4}\alphat\ep^{BCD}\Nt_{AC}\tilde{\Bf}_{BD} + \frac{3}{2}\alphat\Ut^B\tilde{\Bf}_{AB} - \frac{1}{2}\alphat\At^B\tilde{\Bf}_{AB}, \label{ev-M-cnstr} \\
    \del{t}\tilde{\Hf} =& -\frac{2}{t}\Ut^A\tilde{\Df}_A - \frac{2}{t}\tilde{\Hf} + (2+4\rho)\alphat\et_A(\tilde{\Mf}^A) + 4\rho\alphat\Ut_A\tilde{\Mf}^A + (4\alphat\Ut^A-(4+12\rho)\alphat\At^A)\tilde{\Mf}_A - 10\alphat\Hct\tilde{\Hf}. \label{ev-H-cnstr}
\end{align}

Given a solution $(\et_P^\Sigma, \alphat, \At_P, \Ut_P, \Hct, \Sigmat_{PQ}, \Nt_{PQ})$ of the evolution equations \eqref{mEEb.1}-\eqref{mEEb.7}, we first note that the constraint propagation equation \eqref{ev-A-cnstr} is a linear ODE for  $\tilde{\Af}_{AB}^\Omega$ that decouples from the remaining constraint propagation equations. Because \eqref{ev-A-cnstr} is homogeneous in $\tilde{\Af}_{AB}^\Omega$, it follows from
the uniqueness of solutions to ODEs that 
if  $\tilde{\Af}_{AB}^\Omega$ is initially zero, then it will remain so. This shows that the constraint $\tilde{\Af}_{AB}^\Omega=0$ propagates.
Given this, it is enough to consider the remaining constraint propagation equations \eqref{ev-B-cnstr}-\eqref{ev-H-cnstr} with $\tilde{\Af}_{AB}^\Omega=0$,
which we can express in matrix form as 
\begin{equation} \label{cprop-sys}
    \del{t}\tilde{\Uf} = \alphat \et_C^\Lambda \tilde{\Pf}^C \del{\Lambda} \tilde{\Uf} +  \Qsc(t,\Wt)\tilde{\Uf} 
\end{equation}
where
\begin{equation}\label{Uf-def}
    \tilde{\Uf} = (\tilde{\Mf}_P, \tilde{\Bf}_{PQ}, \tilde{\Cf}^P, \tilde{\Hf}, \tilde{\Df}_P)^{\tr},
\end{equation}
\begin{equation} \label{Nf-def}
    \tilde{\Pf}^C = \begin{bmatrix}
        0& -\frac{1}{2}\delta^{C[P}\delta_A^{Q]}& \frac{1}{2}\ep_A{}^C{}_P& \frac{1}{2}\delta_A^C& 0& \\
        -2\gamma\delta_{[A}^C\delta_{B]}^P& 0& 0& 0& 0& \\
        \rho\ep^{ACP}& 0& 0& 0& 0& \\
        (2+4\rho)\delta^{CP}& 0& 0& 0& 0& \\
        0& 0& 0& 0& 0
    \end{bmatrix},
\end{equation}
and $\Qsc(t,\Wt)$ a linear map that depends polynomially on $\Wt$ and $t^{-1}$. As shown in the following lemma, this system is strongly hyperbolic\footnote{That is, the matrix \eqref{Pf-def} has, for any $(\xit_C)\neq 0$, purely real eigenvalues and is diagonalisable with eigenspaces of constant dimension.}.

\begin{lem} \label{lem:chyp}
Suppose $(\xit_\Lambda)\in \Rbb^3_{\times}$, $\alphat> 0$, $\det(\et^\Lambda_C)\neq 0$, $\gamma>\rho$, $1+2\rho>0$, and $\gamma\neq 0$, and set  
\begin{equation} \label{Pf-def}
     \tilde{\Pf} := \tilde{\Pf}^C \xi_C = \begin{bmatrix}
         0& -\frac{1}{2}\xi^{[P}\delta_A^{Q]}& \frac{1}{2}\ep_{ACP}\xi^C& \frac{1}{2}\xi_A& 0& \\
         -2\gamma\xi_{[A}\delta_{B]}^P& 0& 0& 0& 0& \\
         \rho\ep^{ACP}\xi_C& 0& 0& 0& 0& \\
         (2+4\rho)\xi^P& 0& 0& 0& 0& \\
         0& 0& 0& 0& 0
     \end{bmatrix}
\end{equation}
where $\xi_C=\alphat \et_C^\Lambda \xit_\Lambda$. 
If $\gamma\neq 2+5\rho$ then
$\tilde{\Pf}$ 
has real eigenvalues
\begin{equation*}
    \biggl\{0, \dsp \sqrt{\frac{\gamma-\rho}{2}}\,|\xi|, -\sqrt{\frac{\gamma-\rho}{2}}\,|\xi|, \sqrt{1+2\rho}\,|\xi|,- \sqrt{1+2\rho}\, |\xi|\biggr\}
\end{equation*}
and the dimension of the corresponding eigenspaces are  $\bigl\{13,2,2,1,1\bigr\}$.
Otherwise, if $\gamma = 2+5\rho$, then 
$\tilde{\Pf}$ 
has real eigenvalues
\begin{equation*}
\biggl\{0,\sqrt{1+2\rho}\,|\xi|,- \sqrt{1+2\rho}\, |\xi|\biggr\}
\end{equation*}
and the dimension of the corresponding eigenspaces are  $\bigl\{13,3,3\bigr\}$.
In either case,  $\tilde{\Pf}$ is diagonalisable with real eigenvalues.
\end{lem}
\begin{proof}
From \eqref{Nf-def}, it is clear that $\tilde{\Pf}$ defines a linear operator on
\begin{equation*} 
\Vf = \mathbb{R}^3 \times \mathbb{M}_{3\times3} \times \mathbb{R}^3 \times \mathbb{R} \times \mathbb{R}^3,
\end{equation*}
where we note that
\begin{equation}\label{dim-Vf}
\dim_{\Rbb}(\Vf)=19.
\end{equation}
Now, suppose
\begin{equation*}
  \tilde{\Uf} = (\tilde{\Mf}_P, \tilde{\Bf}_{PQ}, \tilde{\Cf}^P, \tilde{\Hf}, \tilde{\Df}_P)^{\tr} \in \Vf
\end{equation*}
is an eigenvector of $\tilde{\Pf}$ with a possibly complex eigenvalue $\lambda$. Then 
$\tilde{\Pf} \Uft = \lambda \Uft$, or equivalently
\begin{align}
    -\frac{1}{2}\xi^B \Bft_{[BA]} + \frac{1}{2}\ep_A{}^{CB}\Cft_B \xi_C + \frac{1}{2}\xi_A\Hft &= \lambda \Mft_A, \label{eigen-1.1} \\
        -2\gamma \xi_{[A} \Mft_{B]} &= \lambda \Bft_{AB}, \label{eigen-1.2} \\
        \rho\ep^{ACB}\Mft_B \xi_C &= \lambda \Cft^A, \label{eigen-1.3} \\
        (2+4\rho)\xi^A \Mft_A &= \lambda \Hft, \label{eigen-1.4} \\
        0 &= \lambda \Dft_A. \label{eigen-1.5}
\end{align}
To help analyse these equations, we define
\begin{equation*}
    \xih_A = \frac{\xi_A}{|\xi|} \AND \pi_B^A = \delta_B^A - \xih_B\xih^A,
\end{equation*}
where $|\xi|^2 = \delta^{AB}\xi_A\xi_B$,
and decompose $\Mf_A$ as
\begin{equation} \label{Mft-decomp}
    \Mft_A = \mf \xih_A + \ell_A
\end{equation}
where
\begin{equation*}
    \mf = \xih^A\Mft_A \AND \ell_A = \pi_A^B\Mft_B.
\end{equation*}
Note that $\xih^A\ell_A=0$ and in line with our conventions all of the spatial frame indices are raised and lowered with the Euclidean metric, eg. $\hat{\xi}^A=\delta^{AB}\hat{\xi}_A$ and $\tilde{\Cf}_A = \delta_{AB}\tilde{\Cf}^A$.

Substituting \eqref{Mft-decomp} into \eqref{eigen-1.1}-\eqref{eigen-1.5} yields
\begin{align}
    -\frac{1}{2}\xi^B \Bft_{[BA]} + \frac{1}{2}\ep_A{}^{CB}\Cft_B \xi_C + \frac{1}{2}\xi_A\Hft &= \lambda \mf\xih_A + \lambda\ell_A, \label{eigen-2.1} \\
        -2\gamma \xi_{[A} \ell_{B]} &= \lambda \Bft_{AB}, \label{eigen-2.2} \\
        \rho\ep^{ACB}\ell_B \xi_C &= \lambda \Cft^A, \label{eigen-2.3} \\
        (2+4\rho)\mf |\xi| &= \lambda \Hft, \label{eigen-2.4} \\
        0 &= \lambda \Dft_A. \label{eigen-2.5}
\end{align}
Now, two cases follow.

\medskip

\noindent\underline{Case 1: $\lambda \neq 0:$}
If $\lambda \neq0$, then, from \eqref{eigen-2.2}-\eqref{eigen-2.5}, we have
\begin{align}
    \Bft_{AB} &= -\frac{2\gamma}{\lambda} \xi_{[A}\ell_{B]}, \label{Bft-fix-1} \\
    \Cft^A &= \frac{\rho}{\lambda}\ep^{ACB} \ell_B\xi_C, \label{Cft-fix-1} \\
    \Hft &= \frac{2}{\lambda}(1+2\rho)\mf |\xi|, \label{Hft-fix-1} \\
    \Dft_A &= 0. \label{Dft-fix-1}
\end{align}
Plugging \eqref{Bft-fix-1}-\eqref{Hft-fix-1} into \eqref{eigen-2.1} yields
\begin{equation*}
    \frac{\gamma}{2\lambda}|\xi|^2 \ell_A - \frac{\rho}{2\lambda}|\xi|^2 \ell_A + \frac{1+2\rho}{\lambda}\mf|\xi|^2 \xih_A = \lambda \mf\xih_A + \lambda\ell_A,
\end{equation*}
or after rearranging,
\begin{equation*}
    \biggl( \lambda^2 - \biggl(\frac{\gamma}{2}-\frac{\rho}{2}\biggr)|\xi|^2 \biggr)\ell_A + \bigl( \lambda^2 - (1+2\rho)|\xi|^2 \bigr)\mf \xih_A = 0.
\end{equation*}
Since, by assumption, $|\xi|> 0$,
$\gamma > \rho$ and $1+2\rho > 0$, we see that
\begin{equation}\label{eigenvalue-case-1}
   \gamma\neq 2+5\rho, \quad \lambda = \pm\sqrt{\frac{\gamma}{2}-\frac{\rho}{2}}\,|\xi| \AND \mf \xih_A = 0,
\end{equation}
or
\begin{equation}
\label{eigenvalue-case-2}
   \gamma\neq 2+5\rho, \quad \lambda = \pm\sqrt{1+2\rho}\,|\xi| \AND \ell_A = 0,
\end{equation}
or
\begin{equation} \label{eigenvalue-case-3}
\gamma=2+5\rho \AND \lambda = \pm\sqrt{1+2\rho}\,|\xi|.
\end{equation}
If either \eqref{eigenvalue-case-1} or \eqref{eigenvalue-case-2} holds, then we conclude from  \eqref{Mft-decomp} and \eqref{Bft-fix-1}-\eqref{Cft-fix-1} that  $\dsp \pm \sqrt{\frac{\gamma}{2}-\frac{\rho}{2}}\,|\xi|$  and $\pm\sqrt{1+2\rho}\,|\xi|$ are real eigenvalues of $\Pf$ with  corresponding eigenspaces 
\begin{equation*}
    \Ef_{\pm} = \biggl\{ \biggl(\ell_A,\, -\frac{2\gamma}{\lambda}\xi_{[A}\ell_{B]},\, \frac{\rho}{\lambda}\ep^{ACB}\ell_B\xi_C,\, 0,\, 0\biggr)^{\tr} \, \biggl| \, (\ell_A) \in \Rbb^3, \; \xi^B\ell_B=0, \; \lambda =  \pm \sqrt{\frac{\gamma}{2}-\frac{\rho}{2}}\,|\xi|\;  \biggr\}
\end{equation*}
and
\begin{equation*}
    \Ff_{\pm} = \biggl\{ \biggl(\mf\xih_A,\, 0,\, 0,\, \frac{2}{\lambda}(1+2\rho)\mf |\xi|,\, 0\biggr)^{\tr} \biggl| \, \mf \in \mathbb{R},\; \lambda = \pm\sqrt{1+2\rho}\,|\xi|\; \biggr\},
\end{equation*}
respectively. 
We further note that
\begin{equation} \label{dim-1}    \dim_{\Rbb}\Ef_{\pm} = 2 \AND \dim_{\Rbb}\Ff_{\pm} = 1.
\end{equation}

On the other hand if \eqref{eigenvalue-case-3} holds, then 
$\pm\sqrt{1+2\rho}\,|\xi|$ are two real eigenvalues with corresponding eigenspaces $\Ef_{+}\oplus \Ff_{+}$ 
and $\Ef_{-}\oplus \Ff_{-}$, both of which are $3$-dimensional.

\medskip

\noindent\underline{Case 2: $\lambda = 0$:} If $\lambda=0$, then, it follows from \eqref{eigen-2.1}-\eqref{eigen-2.5} that
\begin{align}
    -\frac{1}{2}\xi^B \Bft_{[BA]} + \frac{1}{2}\ep_A{}^{CB}\Cft_B \xi_C + \frac{1}{2}\xi_A\Hft &= 0, \label{eigen-3.1} \\
        -2\gamma \xi_{[A} \ell_{B]} &= 0, \label{eigen-3.2} \\
        \rho\ep^{ACB}\ell_B \xi_C &= 0, \label{eigen-3.3} \\
        (2+4\rho)\mf |\xi| &= 0. \label{eigen-3.4}
\end{align}
Since $\gamma \neq 0$ by assumption, \eqref{eigen-3.2} implies that $\xi_{[A} \ell_{B]} = 0$. Contracting this with $\ell^A$ yields
\begin{equation*}
    \ell^A\xi_A\ell_B - \ell^A\xi_B\ell_A = -|\ell|^2\xi_B = 0,
\end{equation*}
and hence that  $\ell_A = 0$ since $|\xi|>0$.
By \eqref{eigen-3.4}, we also have that $\mf=0$ since $1+\rho>0$, and so we deduce from  \eqref{Mft-decomp} that \begin{equation} \label{Mft-fix-2}
    \Mft_A = 0.
\end{equation}

Next, contracting \eqref{eigen-3.1} with $\xi^A$ yields $\frac{1}{2}|\xi|^2\Hft = 0$, and hence that
\begin{equation} \label{Hft-fix-2}
    \Hft = 0.
\end{equation}
Plugging this back into \eqref{eigen-3.1} gives
\begin{equation} \label{eigen-3.1-a}
    -\frac{1}{2}\xi^B \Bft_{[BA]} + \frac{1}{2}\ep_A{}^{CB}\Cft_B \xi_C = 0.
\end{equation}
To proceed, we decompose $\Bft_{AB}$ as
\begin{equation} \label{Bft-decomp-1}
    \Bft_{AB} = \chi_{AB} + \ep_{AB}{}^C \zeta_C
\end{equation}
where $\Bft_{(AB)} = \chi_{AB}$ and $\Bft_{[AB]} = \ep_{AB}{}^C \zeta_C$.
Using this decomposition, we can express \eqref{eigen-3.1-a} as
$\ep_{ABC}\xi^B(\zeta^C + \Cft^C) = 0$,
which will hold if and only if
\begin{equation} \label{zeta-Cft-a}
    \zeta^C + \Cft^C = \cf \xi^C
\end{equation}
for some $\cf\in\mathbb{R}$. So if we decompose $\zeta^C$ and $\Cft^C$ as $\zeta^C = \xih^B\zeta_B \xih^C + \pi_B^C \zeta^B$ and $\Cft^C = \xih^B\Cft_B \xih^C + \pi_B^C \Cft^B$,
respectively, then \eqref{zeta-Cft-a} can be expressed as 
\begin{equation} \label{zeta-Cft-b}
    \pi_B^C(\zeta^B + \Cft^B) = 0.
\end{equation}
It therefore follows from \eqref{Mft-fix-2}, \eqref{Hft-fix-2}, \eqref{Bft-decomp-1}, and \eqref{zeta-Cft-b}, that the kernel of $\tilde{\Pf}$, i.e.~the eigenspace corresponding to the eigenvalue $0$, is determined by
\begin{equation*}
\ker \tilde{\Pf} = \Bigl\{\, \bigl(0, \chi_{AB}+\ep_{AB}{}^C\zeta_C, \tilde{\Cf}^A, 0, \tilde{\Df}_A\bigr)^{\tr}\, \Bigl| \, (\chi_{AB})\in \Sbb{3},\; (\Cf^A), (\Df_A), (\zeta^C) \in \Rbb^3,  \pi_B^C(\zeta^B + \Cft^B) =  0\;\Bigr\},
\end{equation*}
where $\Sbb{3}$ is the set of $3\times 3$-symmetric matrices. We also note that
\begin{equation} \label{dim-KerNf}
    \dim_{\Rbb} \ker \tilde{\Pf} = 13.
\end{equation}

Now since
\begin{equation*}
\dim_{\Rbb}\Vf = \dim_{\Rbb} \Ef_+ + \dim_{\Rbb} \Ef_- + \dim_{\Rbb} \Ff_+ + \dim_{\Rbb} \Ff_- + \dim_{\Rbb} \ker \tilde{\Pf}
\end{equation*}
by \eqref{dim-Vf}, \eqref{dim-1} and \eqref{dim-KerNf}, we conclude that the vector space $\Vf$ can be decomposed as a direct sum of eigenspaces according to
\begin{equation*}
\Vf = \Ef_+ \oplus \Ef_- \oplus \Ff_+ \oplus \Ff_- \oplus \ker \tilde{\Pf},
\end{equation*}
which completes the proof.
\end{proof}

\subsection{Local-in-time existence and constraint propagation on $M_{t_1,t_0}$\label{Sec:loc-exist}}
By making the following choices for the constants
\begin{equation} \label{consts-fix}
    a=2, \ b=0, \ c=-\frac{2}{3}, \ d=\frac{1}{3}, \ m=3, \ n=0, \ p=6, \ q=1, \ h=1, \ l=\frac{1}{3}, \ s=\frac{1}{3}, \ \rho=0, \ \gamma=2,
\end{equation}
we know from the calculations carried out in Section \ref{sec:sym-hyp} that the system \eqref{mEEb.2}-\eqref{mEEb.7} can be cast, cf.~\eqref{EE-symform-A}, into the symmetric hyperbolic form
\begin{equation}\label{EE-symform-B}
\Bh^0\del{t}\Wh + \alphat \et^\Lambda_C\Bh^C \del{\Lambda}\Wh = \Gh(t,\Wh)
\end{equation}
where
\begin{align}
\Bh^0 &= \diag\Bigl(            \delta^P_A\delta^\Omega_\Sigma, 1, 6\delta^P_A, \frac{1}{3}\delta^P_A, 1,\frac{1}{3}\delta^P_A\delta^Q_B, \frac{1}{3}\delta^P_A\delta^Q_B \Bigr), \label{Bh0-def}\\
\Bh^C &=-\begin{bmatrix}
            0&  0&  0&  0&  0&  0&  0 \\ 
            0&  0&  0&  0&  0&  0&  0 \\ 
            0&  0&  0&  0&  -2\delta_A^C& \delta^{C \langle P}\delta_A^{Q \rangle}&  0 \\ 
            0&  0&  0&  0&  \frac{1}{3}\delta_A^C&  \frac{1}{3}\delta^{C \langle P}\delta_A^{Q \rangle}&  0 \\
            0&  0&  -2\delta^{CP}&  \frac{1}{3}\delta^{CP}&  0&  0&  0 \\
            0&  0&  \delta_{\langle A}^C\delta_{B \rangle}^P&  \frac{1}{3}\delta_{\langle A}^C\delta_{B \rangle}^P&  0&  0&  \frac{1}{3}\vep^{C}{}_{AB}{}^{PQ} \\
            0&  0&  0&  0&  0&  \frac{1}{3}\vep^{CPQ}{}_{AB}&  0
        \end{bmatrix}, \label{BhC-def}\\
 \Sc &= \begin{bmatrix}
        \delta^P_A & 0 & 0 & 0 & 0 & 0 & 0\\
            0 & 1 & 0 & 0 & 0 & 0 & 0 \\
            0 & 0 & 2\delta^P_A & 0 & 0 & 0 & 0 \\
            0 & 0 & -\frac{2}{3}\delta^P_A & \frac{1}{3}\delta^P_A & 0 & 0 & 0 \\
            0 & 0 & 0 & 0 & 1 & 0 & 0 \\
            0 & 0 & 0 & 0 & 0 & \delta^P_A\delta^Q_B & 0 \\
            0 & 0 & 0 & 0 & 0 & 0 & \frac{1}{3}\delta^P_A\delta^Q_B
    \end{bmatrix},\label{Sc-fix}\\
\Gh(t,\Wh) &= \St \Gt_{(0,2)}, \label{Gh-fix}
\end{align}
$\Gt$ is defined by \eqref{Gt-def} and $\Wt$, see \eqref{Wt-def}, is related to $\Wh$ via 
\begin{equation}\label{Wh-to-Wt}
\Wh:=V^{-1}\Wt= \biggl(\et_P^\Sigma, \alphat, \frac{1}{3}\At_P, -2\At_P + \Ut_P, \Hct, \Sigmat_{PQ}, \Nt_{PQ}\biggr)^{\tr}  
\end{equation}
with
\begin{equation}\label{V-fix}
  V = \begin{bmatrix}
        \delta^P_A & 0 & 0 & 0 & 0 & 0 & 0\\
        0 & 1 & 0 & 0 & 0 & 0 & 0 \\
        0 & 0 & 3\delta^P_A & 0 & 0 & 0 & 0 \\
        0 & 0 & 6\delta^P_A & \delta^P_A & 0 & 0 & 0 \\
        0 & 0 & 0 & 0 & 1 & 0 & 0 \\
        0 & 0 & 0 & 0 & 0 & \delta^P_A\delta^Q_B & 0 \\
        0 & 0 & 0 & 0 & 0 & 0 & \delta^P_A\delta^Q_B
    \end{bmatrix}.
\end{equation}
Note that $\Gh(t,\Wh)$ depends polynomially on $\Wh$ and $t^{-1}$, and for the choice of constants \eqref{consts-fix}, that the constraint propagation equations \eqref{cprop-sys} are strongly hyperbolic.

In the following proposition, we establish the local-in-time existence and uniqueness of solutions, along with a continuation principle, for the initial value problem (IVP) \begin{align}
\Bh^0\del{t}\Wh + \alphat \et^\Lambda_C\Bh^C \del{\Lambda}\Wh &= \Gh(t,\Wh)\hspace{0.5cm} \text{in $M_{t_1,t_0}=(t_1,t_0]\times \Tbb^3$,} \label{EE-loc-IVP.1}\\
\Wh &= V^{-1}\Wt_0 \hspace{0.5cm}  \text{in $\Sigma_{t_0}=\{t_0\}\times \Tbb^3$.}\label{EE-loc-IVP.2}
\end{align}
Additionally, we prove that the constraints propagate. 

\begin{prop}\label{prop:locA}
Suppose $t_0>0$, $k\in\Zbb_{>5/2}$, and $\Wt_0 \in H^k(\Tbb^3,\Wbb)$ where $\Wbb$ is defined by \eqref{Wbb-def}. Then there exists a $t_1 \in (0,t_0)$ and a unique classical solution
$\Wh \in C^1(M_{t_1,t_0})$  of the IVP \eqref{EE-loc-IVP.1}-\eqref{EE-loc-IVP.2}. Moreover, 
$\Wh \in \bigcap_{j=0}^k C^j\bigl((t_1,t_0], H^{k-j}(\Tbb^{3},\Wbb)\bigr)$
and the following hold:
\begin{enumerate}[(a)]
\item 
If $\sup_{t_1<t<t_0}\norm{\Wh(t)}_{W^{1,\infty}(\Tbb^{3})} < \infty$,
then  there exists a time $t_1^*\in [0, t_1)$ such that $\Wh$
can be uniquely continued as classical solution of \eqref{EE-loc-IVP.1} to $M_{t_1^*,t_0}$.
\item If $\Wt_0(x)\in \Uc$ for all $x\in \Tbb^3$, where $\Uc$ is defined by \eqref{Uc-def}, then $\Wt(t,x)=V \Wh(t,x)\in \Uc$ and  $\Wh(t,x)\in \Uc$ for all $(t,x)\in M_{t_1,t_0}$.
\item If $\Wt_0$ is chosen so that $\Wt_0(x)\in \Uc$ for all $x\in \Tbb^3$ and the constraint equations \eqref{A-cnstr-2}-\eqref{H-cnstr-2} hold on $\Sigma_{t_0}$, then the constraint equations continue to hold everywhere on $M_{t_1,t_0}$.
\end{enumerate} 
\end{prop}
\begin{proof}$\;$

\subsubsection*{Existence and uniqueness}
Suppose $t_0>0$, $k\in\Zbb_{>5/2}$ and $\Wt_0 \in H^k(\Tbb^3,\Wbb)$. Then, because the system \eqref{EE-loc-IVP.1} is symmetric hyperbolic, we can appeal to standard local-in-time existence and uniqueness results for systems of symmetric hyperbolic equations, eg. \cite[Thm.~10.1]{BenzoniSerre:2007}, \cite[Ch.~16, Prop.~2.1]{TaylorIII:1996} or \cite[Thm.~2.1]{Majda:1984},
to deduce the existence of a time $t_1 \in (0,t_0]$ and a unique classical solution\footnote{While \cite[Thm.~10.1]{BenzoniSerre:2007} only guarantees the existence and uniqueness of a solution where $\Wh \in \bigcap_{j=0}^1 C^j\bigl((t_1,t_0], H^{k-j}(\Tbb^{3},\Wbb)\bigr)$, the regularity \eqref{loc-sol-reg} of the solution $\Wh$ can be established by first using the evolution equations \eqref{EE-symform-B} to express the time derivatives $\del{t}^j\Wh$, $1\leq j\leq k$, in terms of $\Wh$ and its spatial derivatives. Using these representations for $\del{t}^j \Wh$, it is then not difficult to verify that \eqref{loc-sol-reg} is a consequence of
$\Wh\in \bigcap_{j=0}^1 C^j\bigl((t_1,t_0], H^{k-j}(\Tbb^{3},\Wbb)\bigr)$
and the calculus inequalities from Appendix \ref{calc}.} $\Wh \in C^1(M_{t_1,t_0})$  of the IVP \eqref{EE-loc-IVP.1}-\eqref{EE-loc-IVP.2} satisfying
\begin{equation} \label{loc-sol-reg}
\Wh \in \bigcap_{j=0}^k C^j\bigl((t_1,t_0], H^{k-j}(\Tbb^{3},\Wbb)\bigr).
\end{equation}

\subsubsection*{Proof of statement (a)} 

By the continuation principle for symmetric hyperbolic systems, eg. \cite[Thm.~10.3]{BenzoniSerre:2007} or \cite[Thm.~2.2]{Majda:1984}, we know that if the solution \eqref{loc-sol-reg} satisfies
$\sup_{t_1<t<t_0}\norm{\Wh(t)}_{W^{1,\infty}(\Tbb^{3})} < \infty$,
then there exists a time $t_1^*\in [0, t_1)$ such that it
can be uniquely continued as a solution of \eqref{EE-loc-IVP.1}-\eqref{EE-loc-IVP.2} to the larger time interval $(t_1^*,t_0]$. 

\subsubsection*{Proof of statement (b)} From \eqref{Wt-def}, \eqref{Gt-def}-\eqref{Gt2-def} and \eqref{Bh0-def}-\eqref{V-fix}, we observe that
\begin{align}
 \del{t}\et_{A}^{\Omega} &= -\alphat(\Hct\delta_{A}^{B} + \Sigmat_{A}{}^{B})\et_{B}^{\Omega}, \label{et-loc-ev} \\
\del{t}\alphat&= 3\Hct \alphat^2. \label{alphat-loc-ev}
\end{align}
Noting that \eqref{alphat-loc-ev} is ODE that is homogeneous in $\alphat$, it follows that $\alphat$ cannot cross zero. So if initially $\alpha|_{\Sigma_{t_0}} >0$, then this will continue to hold as long as the solution exists, and hence, it follows that $\alphat>0$ on $M_{t_1,t_0}$. Also letting $\et=(\et_A^\Omega)$ and noting that
$\Hc \delta_A^A+\Sigmat_A{}^A = 3 \Hc$,
we deduce from the form of the differential equation \eqref{et-loc-ev} and Lemma A.2
of \cite{BeyerOliynyk:2024b} that 
\begin{equation*}
\det(\et(t,x)) = e^{\int^{t_0}_{t}3\alphat(s,x)\Hct(s,x)\,ds }\det(\et(0,x)).
\end{equation*}
Consequently, $\det(\et)>0$ on $M_{t_1,t_0}$ provided that $\det(\et|_{\Sigma_0})>0$. 
Recalling the definition \eqref{Uc-def} of $\Uc$, the above arguments show that if $\Wt_0(x)\in \Uc$ for all $x\in \Tbb^3$, then $\Wt(t,x)=V \Wh(t,x)\in \Uc$ and  $\Wh(t,x)\in \Uc$ for all $(t,x)\in M_{t_1,t_0}$.

\subsubsection*{Proof of statement (c)} Assume now that $\Wt_0\in  H^k(\Sigma_{t_0},\Wbb)$ is chosen so that $\Wt_0(x)\in \Uc$ for all $x\in \Uc$, which, implies, by statement (b), that
\begin{equation} \label{Ubc-preserve}
    \Wt(t,x) \in \Uc, \quad \forall \; (t,x)\in M_{t_1,t_2} .
\end{equation}
Furthermore, suppose that the constraint quantities
$\tilde{\Af}_{AB}^\Omega$, $\tilde{\Bf}_{AB}$, $\tilde{\Cf}^A$, $\tilde{\Df}_A$, $\tilde{\Mf}_A$, and $\tilde{\Hf}$ defined above by \eqref{A-cnstr-2}-\eqref{H-cnstr-2} all vanish on the initial hypersurface $\Sigma_{t_0}$, that is,
$\bigl(\tilde{\Af}_{AB}^\Omega, \tilde{\Bf}_{AB}, \tilde{\Cf}^A, \tilde{\Df}_A, \tilde{\Mf}_A, \tilde{\Hf}\bigr)\bigr|_{\Sigma_{t_0}} = 0$.
Then from \eqref{ev-A-cnstr} and \eqref{Wh-to-Wt}, we see that $\tilde{\Af}^\Omega_{AB}$ satisfies a homogeneous linear ODE with coefficients that are quadratic in $\Wh$, while, by \eqref{loc-sol-reg} and \eqref{Wh-to-Wt}, we have
\begin{equation} \label{loc-sol-reg-A}
\Wt \in \bigcap_{j=0}^k C^j\bigl((t_1,t_0], H^{k-j}(\Tbb^{3},\Wbb)\bigr).
\end{equation}
From this and the vanishing of $\tilde{\Af}_{AB}^\Omega$ on the initial hypersurface $\Sigma_{t_0}$, we conclude via the uniqueness of solutions to systems of homogeneous linear ODEs that \begin{equation}\label{Af-propagate}
\tilde{\Af}_{AB}^\Omega=0 \hspace{0.5cm} \text{in $M_{t_1,t_0}$.}
\end{equation}
Then from the discussion in Section \ref{const-hyp}, \eqref{Wh-to-Wt} and the vanishing of the constraints  on the initial hypersurface $\Sigma_{t_0}$, it follows that the remaining constraint quantities, collected together in the  vector $\tilde{\Uf}$, see \eqref{Uf-def},  satisfy, recall \eqref{cprop-sys}, \eqref{Nf-def} and \eqref{consts-fix}, for any fixed $t_* \in (t_1,t_0)$ the IVP
\begin{align}
    \del{t}\tilde{\Uf} - \alphat \et_C^\Lambda \tilde{\Pf}^C \del{\Lambda} \tilde{\Uf} &= \Qsc(t,\Wt)\tilde{\Uf}  \hspace{0.5cm} \text{in $M_{t_*,t_0}$,}  \label{cprop-sys-fix.1}\\
    \tilde{\Uf} &= 0 \hspace{1.8cm} \text{in $\Sigma_{t_0}$,}\label{cprop-sys-fix.2}
\end{align}
where
\begin{equation} \label{Nf-fix}
    \tilde{\Pf}^C = \begin{bmatrix}
        0& -\frac{1}{2}\delta^{C[P}\delta_A^{Q]}& \frac{1}{2}\ep_A{}^C{}_P& \frac{1}{2}\delta_A^C& 0& \\
        -4\delta_{[A}^C\delta_{B]}^P& 0& 0& 0& 0& \\
        0& 0& 0& 0& 0& \\
        2\delta^{CP}& 0& 0& 0& 0& \\
        0& 0& 0& 0& 0
    \end{bmatrix}
\end{equation}
and $\Qsc(t,\Wt)$ a linear map that depends polynomially on $\Wt$ and $t^{-1}$.

Now, by \eqref{Ubc-preserve}, \eqref{loc-sol-reg-A} and Sobolev's inequality, Theorem \ref{Sobolev}, there exists a $c_*>0$ such that
\begin{equation} \label{alphat-lbnd}
\inf_{(t,x)\in M_{t_*,t_0}}\bigl\{\alphat(t,x),\det(\et(t,x))\bigr\} > c_* > 0.
\end{equation}
Furthermore, letting $\thetat=(\thetat^A_\Lambda)$ denote the inverse of the matrix $\et=(\et^A_\Lambda)$, we deduce from \eqref{alphat-lbnd}, \eqref{loc-sol-reg-A} and the calculus inequalities from Appendix \ref{calc} the existence of a constant $C_*>0$ such that
$\sup_{(t,x)\in M_{t_*,t_0}}|\thetat(t,x)|_{\op}< C_*$.
Thus if we define the open subset $\Uct\subset \Uc$, see \eqref{Uc-def}, by 
\begin{equation}\label{Uct-def}
\Wt \in \Uct \quad \Longleftrightarrow \quad \alphat > c_* \AND  |\thetat|_{\op}<C_*, 
\end{equation}
then $\Wt$ satisfies
\begin{equation}\label{Wt-in-Uct}
\overline{\bigl\{\, \Wt(t,x)\, \bigl|\, (t,x)\in \overline{M}_{t_*,t_0}\, \bigr\}}\subset \Uct. 
\end{equation}

On the other hand, we know, by Lemma \ref{lem:chyp} and the choice of constants \eqref{consts-fix}, that, for $\xit=(\xit_\Lambda)\in \Rbb^3_{\times}$, the eigenvalues of the matrix 
\begin{equation}\label{Pft-def}
\Pft(\Wt,\xit) = \Pft^C \alphat \et^\lambda_C \xit_\Lambda
\end{equation}
are real and given by
\begin{equation*}
\bigl\{\lambda_1(\Wt,\xit)=0,\;\lambda_2(\Wt,\xit)= |\xi|,\;\lambda_3(\Wt,\xit)= -|\xi|\bigr\}, 
\end{equation*}
where $\xi_C = \alphat \et^\Lambda_C \xit_\Lambda$. Moreover, this matrix is
diagonalisable and the dimension of its eigenspaces are independent of $\xit\neq 0$. Letting
\begin{equation*}
\nu = \inf\Bigl\{ \bigl|\lambda_i(\Wt,\xit)-\lambda_j(\Wt,\xit)\bigr|\,\Bigl| \, \Wt\in\Ucb, \; |\xib|=1 \Bigr\}
\end{equation*}
denote the spectral gap,  it follows from \eqref{Uct-def} and the inequality
\begin{equation*}
0<c_*=c_*|\xit|\leq |\alphat\xit|=|\thetat \alphat\et\xit|< C_*|\alphat\et\xit| = C_*|\xi|    
\end{equation*}
that $\nu > 0$. Due to this lower bound, we can construct a symbolic symmetriser
$\Sft \in C^\infty(\Uct \times \Rbb^3_\times, \Mbb{19})$
for the matrix \eqref{Pft-def} using the construction from the proof of \cite[Thm.~2.3]{BenzoniSerre:2007}. This symbolic symmetriser along with \eqref{loc-sol-reg-A} and \eqref{Wt-in-Uct}  allows us 
to conclude from  \cite[Thm.~2.7]{BenzoniSerre:2007} that solutions to the IVP \eqref{cprop-sys-fix.1}-\eqref{cprop-sys-fix.2} are  unique, from which we deduce that $\Uft=0$ in $M_{t_*,t_1}$. But $t_*\in (t_1,t_0]$ was chosen arbitrarily, and it follows that $\Uft=0$ in $M_{t_1,t_0}$. This together with \eqref{Af-propagate} guarantees that the constraints quantities $\tilde{\Af}_{AB}^\Omega$, $\tilde{\Bf}_{AB}$, $\tilde{\Cf}^A$, $\tilde{\Df}_A$, $\tilde{\Mf}_A$, and $\tilde{\Hf}$ all vanish in $M_{t_0,t_1}$ provided they vanish on the the initial hypersurface $\Sigma_{t_0}$, which completes the proof of statement (c).

\end{proof}

\subsection{\label{loc-exist-Omega}Local-in-time existence and constraint propagation on $\Omega_{\qv}$}
While Proposition \ref{prop:locA} establishes the existence of solutions to the tetrad formulation of the conformal Einstein-scalar field equation on the spacetime regions of the form $M_{t_1,t_0}$, what we need, in order to prove the localised stability of the Kasner-scalar field spacetime, is a localised version of this existence result. Specifically,  we need to establish the local-in-time existence and uniqueness of solutions, along with a continuation principle and showing constraint propagation, to the IVP\footnote{Since the initial value problem is formulated on a spatially bounded domain, we could also refer to it as in initial boundary value problem that happens to require no boundary conditions.} \begin{align}
\Bh^0\del{t}\Wh + \alphat \et^\Lambda_C\Bh^C \del{\Lambda}\Wh &= \Gh(t,\Wh)\hspace{0.5cm} \text{in $\Omega_{q_1}$,} \label{EE-loc-IBVP-A.1}\\
\Wh &= V^{-1}\Wt_0 \hspace{0.5cm}  \text{in $\{t_0\}\times \mathbb{B}_{\rho_0}$,}\label{EE-loc-IBVP-A.2}
\end{align}
where $q_1=(t_1,t_0,\rho_0,\rho_1,\ep)$ and  $\Omega_{q_1}$ is the truncated cone domain defined by \eqref{Omega-def}.

\begin{prop}\label{prop:locC}
Suppose $k\in\Zbb_{>5/2}$,  $t_0>0$, 
$\rho_1>0$,  $0<\rho_0<L$, $0\leq \ep<1$, $\rho_0 -\frac{\rho_1 t_0^{1-\ep}}{1-\ep}>0$, $\Wt_0=\bigl(\et_{0P}^\Sigma, \alphat_0, \At_{0P}, \Ut_{0P}, \Hct_0, \Sigmat_{0PQ}, \Nt_{0PQ}\bigr)^{\tr} \in H^k\bigl(\mathbb{B}_{\rho_0},\Wbb\bigr)$, and  $\sup_{x\in \mathbb{B}_{\rho_0}}  t_0^{\ep}|\alphat_0(x) \et_0(x)|<\frac{\rho_1}{9}$. Then there exists a $t_1 \in (0,t_0]$ and a unique classical solution $\Wh \in C^1(\grave{\Omega}_{\qv_1})$ of  \eqref{EE-symform-B} on the truncated cone domain $\Omega_{\qv_1}$, $\qv_1=(t_1,t_0,\rho_0,\rho_1,\ep)$, satisfying the initial condition
$\Wh|_{\{t_0\}\times\mathbb{B}_{\rho_0}}=V^{-1}\Wt_0$ 
and $t^\ep |\alphat(t,x) \et(t,x)|<\frac{\rho_1}{9}$ for all $(t,x)\in \grave{\Omega}_{\qv_1}$. 
Moreover, letting $\rho(t)=\frac{\rho_1(t^{1-\ep}-t_0^{1-\ep})}{1-\ep}+ \rho_0$, then  $\sup_{t< \tilde{t}\leq t_0} \norm{\del{t}^j\Wh(\tilde{t})}_{H^{k-j}(\mathbb{B}_{\rho(\tilde{t})})}<\infty$  for all $(t,j)\in (t_1,t_0]\times\{1,2,\ldots,k\}$, and the following hold:
\begin{enumerate}[(a)]
\item If the solution $\Wh$ satisfies $\sup_{t_1<t<t_0}\norm{\Wh(t)}_{H^{k}(\mathbb{B}_{\rho(t)})} < \infty$ and $  \sup_{(t,x)\in \Omega_{\qv_1}} t^{\ep}|\alphat(t,x) \et(t,x)|<\frac{\rho_1}{9}$,
then there exists a time $t_2\in [0, t_1)$ such that $\Wh$
can be uniquely continued as classical solution of \eqref{EE-symform-B} to the cone domain $\Omega_{\qv_2}$, $\qv_2=(t_2,t_0,\rho_0,\rho_1,\ep)$.
\item If $\Wt_0(x)\in \Uc$ for all $x\in \overline{\mathbb{B}}_{\rho_0}$, then $\Wt(t,x)=V \Wh(t,x)\in \Uc$ and  $\Wh(t,x)\in \Uc$ for all $(t,x)\in \Omega_{\qv_1}$.
\item If $\Wt_0$ is chosen so that $\Wt_0(x)\in \Uc$ for all $x\in \overline{\mathbb{B}}_{\rho_0}$ and the constraint equations \eqref{A-cnstr-2}-\eqref{H-cnstr-2} hold on $\{t_0\}\times\mathbb{B}_{\rho_0}$, then the constraint equations continue to hold everywhere on $\Omega_{\qv_1}$.
\end{enumerate} 
\end{prop}
\begin{proof}$\;$
\subsubsection*{Existence and uniqueness, and a proof of statement (a)} Consider a map $\Phi$ defined by
$x^\lambda = \Phi^\lambda(\xb^0,\xb)$, where $\xb=(\xb^\Lambda)$, $\tb=\xb^0$, 
\begin{equation*} 
t=x^0 = \Phi^0(\tb,\xb) :=\xb^0,\quad x^\Lambda = \Phi^\Lambda(\tb,\xb) := \frac{\rho(\tb)}{\rho_0}\xb^\Lambda, 
\end{equation*}
and $\rho(t)= \frac{\rho_1(t^{1-\ep}-t_0^{1-\ep})}{1-\ep}+ \rho_0$.
Since $\Omega_{\qv_1}=\bigl\{\, (t,x)\in (t_1,t_0]\times (-L,L)^3\, \bigl| \, |x|<\rho(t)\,\bigr\}$, we have that $\Omega_{\qv_1}=\Phi\bigl((t_1,t_0]\times \mathbb{B}_{\rho_0}\bigr)$.
Also, letting $\delb{\mu}=\fdel{\;}{\xb^\mu}$ denote the partial derivatives with respect to the coordinates $(\xb^\mu)$, a short calculation shows that  
\begin{equation}\label{pd-trans}
\del{t}=\del{\tb}-\frac{\rho_1}{\tb^\ep \rho(\tb)}\xb^\Lambda \delb{\Lambda} \AND \del{\Lambda}= \frac{\rho_0}{\rho(\tb)}\delb{\Lambda}.
\end{equation}
Using these transformation laws, we can use $\Phi$ to pull-back the IVP \eqref{EE-loc-IBVP-A.1}-\eqref{EE-loc-IBVP-A.1} to get 
\begin{align}
\Bh^0\del{\tb}\Wb + \Bb^\Lambda(\tb,\xb,\Wb) \delb{\Lambda}\Wb &= \Gb(\tb,\Wb)\hspace{0.5cm} \text{in $(t_1,t_0]\times \mathbb{B}_{\rho_0}$,} \label{EE-loc-IBVP-B.1}\\
\Wb &= V^{-1}\Wt_0 \hspace{0.5cm}  \text{in $\{t_0\}\times \mathbb{B}_{\rho_0}$,}\label{EE-loc-IBVP-B.2}
\end{align}
where 
\begin{gather}
\Wb = (\eb_P^\Sigma, \alphab, \Ab_P, \Ub_P, \Hcb, \Sigmab_{PQ}, \Nb_{PQ})^{\tr}:= \tb^\ep \Wh\circ\Phi, \label{Wb-def}\\
\Bb^\Lambda(\tb,\xb,\Wb)= \frac{1}{\tb^\ep \rho(\tb)}\Bigl( -\rho_1\xb^\Lambda \Bh^0 + \rho_0 \Bh^C \alphab \eb^\Lambda_C\Bigr)\AND \Gb(\tb,\Wb) = \ep \Bh^0\Wb + \Gh\bigl(\tb,\tb^{-2\ep}\Wb\bigr). \label{BbLambda-def}
\end{gather}
The outward pointing unit normal to the boundary component 
    $\Gamma_{t_1,t_0} = (t_,t_0]\times \partial \mathbb{B}_{\rho_0}$  
of the cylinder $(t_1,t_0]\times \mathbb{B}_{\rho_0}$
is  
\begin{equation}\label{nb-def}
    \nb = \frac{1}{|\xb|}\xb_\Lambda d\xb^\Lambda,
\end{equation}
and so, the normal matrix is given by  
\begin{equation*}
\Bb_{\nb}:=\Bb^\Lambda\nb_\Lambda|_{\Gamma_{t_1,t_0}}=\frac{\rho_0}{\tb^\ep \rho(\tb)}\Bigl( -\rho_1\Bh^0 + \Bh^C \mb_C\Bigr), \quad \mb_C = \alphab \eb^\Lambda_C \frac{\xb_{\Lambda}}{|\xb|}.
\end{equation*}

From \eqref{Bh0-def}-\eqref{BhC-def}, we then have
\begin{align} \label{Bn-ubnd-B}
\frac{\tb^\ep \rho(\tb)}{\rho_0}\Wb^T\Bb_{\nb}\Wb &=-\rho_1\biggl( |\eb|^2 + \alphab^2 + 6 |\Ab|^2 + \frac{1}{3}|\Ub|^2 + \Hcb^2 +\frac{1}{3}(|\Sigmab|^2+|\Nb|^2)\biggr)
\notag \\
&\qquad + \biggl( 4\Ab^C-\frac{2}{3}\Ub^C \biggr)\mb_C\Hcb - \biggl(2\Ab_P + \frac{2}{3}\Ub_P \biggr)\mb_C\Sigmab^{CP}+\frac{2}{3}\mb_C\ep^{CP}{}_A \Nb^{AQ}\Sigmab_{PQ}.
\end{align}
Noting that $\mb_C \ep^{C}{}_{PQ} \mb_D \ep^{DPQ}= \mb_{C}\mb_{D}2\delta^{CD}=2|\mb|^2$,
we can, with the help of the  matrix inequalities \eqref{mat-inq-A} and \eqref{mat-inq-B}, bound $\Wb^T\Bb_{\nb}\Wb$ from above by 
\begin{align*}
\frac{\tb^\ep \rho(\tb)}{\rho_0}\Wb^T\Bb_{\nb}\Wb &\leq -\rho_1 \bigl( |\eb|^2 + \alphab^2 \bigr)-\rho_1 \biggl( 6 |\Ab|^2 + \frac{1}{3}|\Ub|^2 + \Hcb^2 +\frac{1}{3}(\Sigmab|^2+|\Nb|^2)\biggr) \\ 
&\qquad 
+|\mb|\biggl(\biggl(4|\Ab|+\frac{2}{3}|\Ub|\biggr)|\Hcb|+\biggl(2|\Ab|+\frac{2}{3}|\Ub|\biggr)|\Sigmab|+ \frac{4}{3}|\Nb||\Sigmab|\biggr),
\end{align*}
and hence, with the help of Young's inequity, $ab \leq \frac{1}{2}a^2 + \frac{1}{2}b^2$, by
\begin{align}
\frac{\tb^\ep \rho(\tb)}{\rho_0}\Wb^T\Bb_{\nb}\Wb &\leq  -\rho_1 \bigl( |\eb|^2 + \alphab^2 \bigr)+ ( 3|\mb|-  6\rho_1) |\Ab|^2 + \biggl(\frac{2|\mb|}{3}-\frac{\rho_1}{3} \biggr)|\Ub|^2\notag \\
&\quad + (|\mb|-\rho_1)\Hcb^2 +\biggl(\frac{5|\mb|}{3}-\frac{\rho_1}{3} \biggr)|\Sigmab|^2+\biggl(\frac{2|\mb|}{3}-\frac{\rho_1}{3}\biggr)|\Nb|^2. \label{Bn-ubnd-A} 
\end{align}

Next, we define an open set $\Vc\in \Wbb$ by
\begin{equation*}
\Vc = \Bigl\{ \Wb\in \Wbb\, \Bigl|\, |\alphab\eb|<\frac{\rho_1}{9}\, \Bigr\}.
\end{equation*}
Then, by \eqref{Bn-ubnd-A}, for any open set $\Wc\ssubset \Vc$, there exists a $\delta>0$ such that 
\begin{equation}\label{spacelike}
\Bb_{\nb}(\tb,\xb,\Wb)\leq -\delta  \id, \quad \forall \, (\tb,\xb,\Wb)\in [t_1,t_0]\times\del{}\mathbb{B}_{\rho_0}\times \Wc.
\end{equation}
As can be verified from this inequality, as long as a solution $\Wb$ of \eqref{EE-loc-IBVP-B.1}-\eqref{EE-loc-IBVP-B.2} starts inside a open set $\Wc\ssubset \Vc$, then the boundary terms will not contribute to energy estimates derived from \eqref{EE-loc-IBVP-B.2} as long as $\Wb$ does not exit every open set $\grave{\Wc}$ satisfying  $\grave{\Wc}\ssubset \Vc$. Because of this, local-in-time existence and uniqueness proofs for symmetric hyperbolic equations, eg. \cite[Thm.~10.1]{BenzoniSerre:2007} or \cite[Ch.~16, Prop.~2.1]{TaylorIII:1996}, as well as the continuation principle, eg. \cite[Thm.~10.3]{BenzoniSerre:2007}, continue to apply to the IVP \eqref{EE-loc-IBVP-B.1}-\eqref{EE-loc-IBVP-B.2} mutatis mutandis. 

Alternatively, the local-in-time existence and uniqueness of solutions to \eqref{EE-loc-IBVP-B.1}-\eqref{EE-loc-IBVP-B.2} follows directly from an application of\footnote{Note that all the conditions of \cite[Thm.~A2]{Schochet:1986} are trivially satisfied due to the inequality \eqref{spacelike} which leads to \textit{(i)} no boundary conditions, ie. $M(x)=0$ in the notation of \cite{Schochet:1986}, for the system \eqref{EE-loc-IBVP-B.1} and \textit{(ii)} no compatibility conditions for the initial data \eqref{EE-loc-IBVP-B.2}.} \cite[Thm.~A2]{Schochet:1986}. Thus, given 
$\Wt_0 \in H^k(\mathbb{B}_{\rho_0},\Wbb)$, where $k\in \Zbb_{>5/2}$, there exists a $t_1\in (0,t_0]$ and a unique solution
\begin{equation} \label{loc-sol-reg-B}
\Wb \in \bigcap_{j=0}^k C^j\bigl((t_1,t_0], H^{k-j}(\mathbb{B}_{\rho_0},\Wbb)\bigr) \subset C^1\bigl((t_1,t_0]\times \overline{\mathbb{B}}_{\rho_0}\bigr)
\end{equation}
of the IVP \eqref{EE-loc-IBVP-B.1}-\eqref{EE-loc-IBVP-B.2}.

To see that the continuation principle hold, assume that
\begin{equation} \label{loc-sol-cont-B}
\sup_{t_1<\tb<t_0}\norm{\Wb(\tb)}_{W^{1,\infty}(\mathbb{B}_{\rho_0})} < \infty \AND \sup_{(\tb,\xb)\in (t_1,t_0)\times \mathbb{B}_{\rho_0}}|\alphab(\tb,\xb)\eb(\tb,\xb)|<\frac{\rho_1}{9},
\end{equation}
which implies that the boundary matrix satisfies \eqref{spacelike}.
As noted above, \eqref{spacelike} implies that the boundary terms do not contribute when integrating by parts to obtain energy estimates. Consequently, we can, with the help of the calculus inequalities, in particular the Moser estimates, from Appendix \ref{calc}, derive energy inequalities satisfied by the solution \eqref{loc-sol-reg-B} in the standard fashion for symmetric hyperbolic systems. This yields the differential energy inequality 
\begin{equation*}
\del{t}\norm{\Wb(\tb)}_{H^k(\mathbb{B}_{\rho_0})} \leq C\bigl(\norm{\Wb(\tb)}_{W^{1,\infty}(\mathbb{B}_{\rho_0})}\bigr)\norm{\Wb(\tb)}_{H^k(\mathbb{B}_{\rho_0})}, \quad t_1<t\leq t_0.
\end{equation*}
From this inequality and the bound $\sup_{t_1<\tb<t_0}\norm{\Wb(\tb)}_{W^{1,\infty}(\mathbb{B}_{\rho_0})} < \infty$,
it then follows via a Gr\"{o}nwall inequality argument that
$\sup_{t_1<\tb<t_0}\norm{\Wb(\tb)}_{H^{k}(\mathbb{B}_{\rho_0})} < \infty$. By using the evolution equation \eqref{EE-loc-IBVP-B.1} to express $\del{\tb}\Wb$ in terms of $\Wb$ and its spatial derivative, the bound on $\Wb$ and the calculus inequalities from Appendix \ref{calc} can be used to show that
$\sup_{t_1<\tb<t_0}\norm{\del{\tb}\Wb(\tb)}_{H^{k-1}(\mathbb{B}_{\rho_0})} < \infty$. By repeatedly differentiating \eqref{EE-loc-IBVP-B.1} with respect to $\tb$, we can, in a similar fashion, obtain bounds
\begin{equation*}
\sup_{t_1<\tb<t_0}\norm{\del{\tb}^\ell\Wb(\tb)}_{H^{k-\ell}(\mathbb{B}_{\rho_0})} < \infty, \quad \ell=0,1,\ldots,k,
\end{equation*}
on all the time derivatives. Once we have this, then by the argument contained in the first paragraph of the proof of \cite[Thm.~A6]{Schochet:1986} there exists a $t_1^*\in [0,t_1)$ such that the solution \eqref{loc-sol-reg-B} can be uniquely continued to the larger time interval $(t_1^*,t_0]$ while satisfying $|\alphab(\tb,\xb)\eb(\tb,\xb)|<\frac{\rho_1}{9}$ for all $(\tb,\xb)\in (t_1^*,t_0]\times \overline{\mathbb{B}}_{\rho_0}$.

Regarding uniqueness of classical solutions, suppose that $\Wb_*\in C^1((t_1,t_0]\times \overline{\mathbb{B}}_{\rho_0})$ is another classical solution of \eqref{EE-loc-IBVP-B.1}-\eqref{EE-loc-IBVP-B.2} satisfying $|\alphab(\tb,\xb)\eb(\tb,\xb)|<\frac{\rho_1}{9}$ for all $(\tb,\xb)\in (t_1,t_0]\times \overline{\mathbb{B}}_{\rho_0}$. Then for any $t_*\in (t_1,t_0]$ the difference $\Zb=\Wb-\Wb_*$ satisfies a IVP of the form 
\begin{align*}
\Bh^0\del{\tb}\Zb + \Bb^\Lambda(\tb,\xb,\Wb) \delb{\Lambda}\Zb &= \Hb(\tb,\Wb,\Wb_*,\delb{}\Wb_*,\Zb)\hspace{0.5cm} \text{in $(t_*,t_0]\times \mathbb{B}_{\rho_0}$,} \label{EE-loc-IBVP-B.1}\\
\Zb &= 0 \hspace{3.65cm}  \text{in $\{t_0\}\times \mathbb{B}_{\rho_0}$,}
\end{align*}
where $\Hb$ is smooth in its arguments and linear in $\Zb$. Now, $\Wb,\Wb_*,\Zb\in C^1([t_*,t_0]\times \overline{\mathbb{B}}_{\rho_0}\bigr)$, and so, using the fact that we can ignore boundary terms when integrating by parts due to \eqref{spacelike}, we can derive a standard $L^2$-energy inequality for $\Zb$ of the form
\begin{equation*}
\del{t}\norm{\Zb(\tb)}_{L^2(\mathbb{B}_{\rho_0})}\leq C \norm{\Zb(\tb)}_{L^2(\mathbb{B}_{\rho_0})}, \quad t_*< t \leq t_0.
\end{equation*}
Gr\"{o}nwall's inequality plus the vanishing of $\Zb$ at $\tb=t_0$ implies that $\norm{\Zb(\tb)}_{L^2(\mathbb{B}_{\rho_0})}$ for $t_*<t\leq t_0$, which in turn implies that $\Wb(\tb,\xb)=\Wb_*(\tb,\xb)$ for all $(\tb,\xb)\in (t_1,t_0]\times \mathbb{B}_{\rho_0}$ since $t_*\in (t_1,t_0)$ was chosen arbitrary. This completes the proof of the uniqueness of classical solutions. 

\subsubsection*{Proof of statement (b)}
For a given classical solution $\Wh\in C^1(\Omega_{\qv_1})$ of \eqref{EE-loc-IBVP-A.1}-\eqref{EE-loc-IBVP-A.2}, $\alphat$ and $\et_A^\Omega$, see \eqref{Wh-to-Wt}, satisfy the same evolution equations \eqref{et-loc-ev}-\eqref{alphat-loc-ev} as from the proof Proposition \ref{prop:locA}.(b). Since, we can again view these equations as ODEs for $(\alphat,\et_A^\Omega)$, the same arguments that were employed to prove Proposition \ref{prop:locA}.(b) also hold in this new setting. In particular, if $\alpha(t_0,x)>0$ and $\det(\et(t_0,x))>0$ for $x\in \overline{\mathbb{B}}_{\rho_0}$, then 
$\alpha(t,x)>0$ and $\det(\et(t,x))>0$ for all $(t,x)\in \Omega_{\qv_1}$. Recalling the definition \eqref{Uc-def} of $\Uc$, we conclude that  $\Wt(t,x)=V \Wh(t,x)\in \Uc$ and  $\Wh(t,x)\in \Uc$ for all $(t,x)\in \Omega_{\qv_1}$ provided
$\Wt_0(x)\in \Uc$ for all $x\in \overline{\mathbb{B}}_{\rho_0}$.

\subsubsection*{Proof of statement (c)} 
Assuming that $\Wt_0\in H^k(\mathbb{B}_{\rho_0},\Wbb)$ satisfies $\Wt_0(x)\in \Uc$ for all $x\in \overline{\mathbb{B}}_{\rho}$, we know from part (b) that the unique classical solution  $\Wb\in C^1\bigl((t_1,t_0]\times \mathbb{B}_{\rho_0}\bigr)$ of \eqref{EE-loc-IBVP-B.1}-\eqref{EE-loc-IBVP-B.2} satisfies $\alphab(\tb,\xb)>0$ and $\det(\eb(\tb,\xb))>0$ for all $(t_1,t_0]\times \mathbb{B}_{\rho_0}$ in addition to \eqref{loc-sol-reg-B}. In particular, fixing $t_*\in (t_1,t_0)$, there exists a $c_*>0$ such that
\begin{equation} \label{alphab-lbnd}
\inf_{(\tb,\xb)\in (t_*,t_0]\times \mathbb{B}_{\rho_0}}\bigl\{\alphab(\tb,\xb),\det(\eb(\tb,\xb))\bigr\} > c_* > 0.
\end{equation}
Letting $\thetab=(\thetab^A_\Lambda)$ denote the inverse of the matrix $\eb=(\eb^A_\Lambda)$, it then follows from the lower bound \eqref{alphab-lbnd} on $\det(\eb)$, the regularity \eqref{loc-sol-reg-B} of the solution $\Wb$, and the calculus inequalities from Appendix \ref{calc} that 
$\sup_{(\tb,\xb)\in (t_*,t_0]\times \mathbb{B}_{\rho_0}}|\thetab(\tb,\xb)|_{\op}< C_*$
for some $C_*>0$. Thus if we define the open subset $\Ucb\subset \Uc$, see \eqref{Uc-def}, by 
\begin{equation}\label{Ucb-def}
\Wb \in \Ucb \quad \Longleftrightarrow \quad \alphab > c_* \AND  |\thetab|_{\op}<C_*, 
\end{equation}
the solution $\Wb$ satisfies
\begin{equation} \label{Wb-in-Ucb}
\overline{\bigl\{\, \Wb(\tb,\xb)\, \bigl|\, (\tb,\xb)\in [t_*,t_0]\times \mathbb{B}_{\rho_0} \, \bigr\}}\subset \Ucb. 
\end{equation}

Next, by \eqref{Wb-def} and \eqref{loc-sol-reg-B}, $\Wh\in C^1(\Omega_{\qv_1})$ and solves \eqref{EE-loc-IBVP-A.1}-\eqref{EE-loc-IBVP-A.2}. Furthermore, by \eqref{ev-A-cnstr} and \eqref{Wh-to-Wt}, we know that $\Af^\Omega_{AB}$ satisfies a homogeneous linear ODE with coefficients that are quadratic in $\Wh$. Assuming that $\tilde{\Af}_{AB}^\Omega$ vanishes on
$\{t_0\}\times \mathbb{B}_{\rho_0}$, it follows immediately from the uniqueness of solutions to ODEs that $\tilde{\Af}_{AB}^\Omega=0$ in $\Omega_{\qv_1}$. Then, by the same arguments as in the proof of Proposition \ref{prop:locA}.(b), the remaining constraint quantities $\tilde{\Uf}$, recall \eqref{Uf-def}, satisfy the IVP 
\begin{align}
    \del{t}\tilde{\Uf} - \alphat \et_C^\Lambda \tilde{\Pf}^C \del{\Lambda} \tilde{\Uf} &= \Qsc(t,\Wt)\tilde{\Uf}  \hspace{0.5cm} \text{in $\Omega_{\qv_1}$,}  \label{cprop-sys-A.1}\\
    \tilde{\Uf} &= 0 \hspace{1.8cm} \text{in $\{t_0\}\times \mathbb{B}_{\rho_0}$,}\label{cprop-sys-A.2}
\end{align}
where $\tilde{\Pf}^C$ is defined by \eqref{Nf-fix} and 
$\Qsc(t,\Wt)$ a linear map that depends polynomially on $\Wt$ and $t^{-1}$. 
Setting
$\Ufb = \tilde{\Uf}\circ \Phi$,
we see, with the help of \eqref{pd-trans}, that $\Ufb$ satisfies
\begin{align}
    \del{\tb}\Ufb - \Pfb^\Lambda(\tb,\xb,\Wb) \delb{\Lambda} \Ufb &= \Qsc(\tb,\tb^{-\ep}\Wb)\Ufb  \hspace{0.5cm} \text{in $(t_*,t_0]\times\mathbb{B}_{\rho_0}$,}  \label{cprop-sys-A.1}\\
    \Ufb &= 0 \hspace{2.3cm} \text{in $\{t_0\}\times\mathbb{B}_{\rho_0}$,}\label{cprop-sys-A.2}
\end{align}
where
\begin{equation*}
\Pfb^\Lambda(\tb,\xb,\Wb) =\frac{1}{\tb^\ep \rho(\tb)}\bigl(\rho_1 \xb^\Lambda \id + \rho_0 \alphab \eb^\Lambda_C \tilde{\Pf}^C\bigr). 
\end{equation*}
Then for $\xib=(\xib_\Lambda)\in \Rbb^3_{\times}$, it follows from Lemma \ref{lem:chyp} and the choice of constants \eqref{consts-fix} the eigenvalues of the matrix
\begin{equation}\label{Pfb-def}
\Pfb(\tb,\xb,\Wb,\xib) = \Pfb^\Lambda(\tb,\xb,\Wb)\xib_\Lambda
\end{equation}
are real and given by
\begin{equation*}
\biggl\{\lambda_1(\tb,\xb,\Wb,\xib)= \frac{\rho_1}{\tb^\ep \rho(\tb)} \xb^\Lambda \xib_\Lambda,\;\lambda_2(\tb,\xb,\Wb,\xib)= \frac{1}{\tb^\ep \rho(\tb)}\bigl(\rho_1 \xb^\Lambda \xib_\Lambda + \rho_0|\xi|\bigr),\;\lambda_3(\tb,\xb,\Wb,\xib)= \frac{1}{\tb^\ep \rho(\tb)}\bigl(\rho_1 \xb^\Lambda \xib_\Lambda - \rho_0|\xi|\bigr)\biggr\}, 
\end{equation*}
where $\xi_C = \alphab \eb^\Lambda_C \xib_\Lambda$,
and that the matrix is diagonalisable where, as long as the eigenvalues do not cross, the eigenspaces are of constant dimension. To see that the eigenvalues do not cross, we let
\begin{equation*}
\nu = \inf\Bigl\{ \bigl|\lambda_i(\tb,\xb,\Wb,\xib)-\lambda_j(\tb,\xb,\Wb,\xib)\bigr|\,\Bigl| \, (\tb,\xb)\in [t_*,t_0]\times \mathbb{B}_{\rho_0}, \; \Wb\in\Ucb, \; |\xib|=1 \Bigr\}
\end{equation*}
denote the spectral gap. Then it follows from \eqref{Ucb-def} and 
\begin{equation*}
0<c_*=c_*|\xib|\leq |\alphab\xib|=|\thetab \alphab\eb\xib|< C_*|\alphab\eb\xib| = C_*|\xi|    
\end{equation*}
that $\nu > 0$, which guarantees that the eigenvalues remain separated.  
Defining an open subset of $\Ucb$
\begin{equation*}
\Vcb = \Bigl\{\, \Wb \in \Ucb \, \Bigl| \, |\ab\eb|< \frac{\rho_1}{9}\, \Bigr\}.
\end{equation*}
the smallest eigenvalue of \eqref{Pfb-def}, when evaluated at $\xib=\nb(\xb)$ and $\Wb \in \Vcb$, where $\nb$ is the outward pointing unit normal to the boundary of $\mathbb{B}_{\rho}$, see \eqref{nb-def}, is 
\begin{equation*}
\lambda_3\bigl(\tb,\xb,\Wb,\nb(\xb)\bigr)= \frac{1}{\tb^\ep \rho(\tb)}\bigl(\rho_1\rho_0 - \rho_0|\alphab \eb \nb(\xb)|\bigr) \geq \frac{1}{\tb^\ep \rho(\tb)}\Bigl(\rho_1\rho_0  - \frac{\rho_0\rho_1}{9}\Bigr) = \frac{8\rho_0\rho_1}{\tb^\ep \rho(\tb)}>0.
\end{equation*}
This shows, in particular, that all the eigenvalues of the boundary matrix, ie. \eqref{Pfb-def} evaluated at $\xib=\nb(\xb)$, $\Wb \in \Vcb$ and $\xb\in \del{}\mathbb{B}_{\rho_0}$, are positive, which in turn, implies that the boundary $\del{}\mathbb{B}_{\rho_0}$ is non-characteristic, all the characteristic are outgoing, and the uniform Kreiss-Lopatiniski\u{i} condition is trivially satisfied. By \eqref{loc-sol-cont-B} and \eqref{Wb-in-Ucb}, $\Wb$ takes values in $\Vc$ for $(\tb,\xb)\in [t_*,t_0]\times \mathbb{B}_{\rho_0}$, and so, we can conclude from \cite[Thm.~9.15] {BenzoniSerre:2007}\footnote{See also the remark in the proof of \cite[Thm.~9.15] {BenzoniSerre:2007} regarding rough coefficients along with \cite[Thm.~9.18]{BenzoniSerre:2007} and the remark preceding it.} that the only solution to \eqref{cprop-sys-A.1}-\eqref{cprop-sys-A.2} is $\Ufb=0$. Since $t_*\in (t_1,t_0]$ was chosen arbitrarily, all the constraint vanish in  $(t_1,t_0]\times \mathbb{B}_{\rho_0}$, which completes the proof of statement (c).
\end{proof}

\section{Fuchsian formulation\label{sec:Fuch-global-existence-theory}}
The tetrad formulation of the conformal Einstein-scalar field equations derived in the previous section is not yet quite well-suited for analysing the behaviour of nonlinear perturbations of the Kasner-scalar field family of solutions near $ t=0 $. To address this, we renormalise the tetrad variables  
$\{\alphat, \Hct, \et^\Omega_A, \At_A, \Ut_A, \Sigmat_{AB}, \Nt_{AB}\}$  
by introducing weights of the form $t^\ep \alphat$ and off-set a subset of the renormalised components so that they vanish on the Kasner-scalar field solutions. The purpose of this transformation, along with the other manipulations in this section, is to reformulate the tetrad equations into a Fuchsian form suitable for deriving global-in-time estimates using the Fuchsian global existence theory from \cite{BeyerOliynyk:2024a, BOOS:2021}.

\subsection{Renormalised tetrad variables}
The renormalised tetrad variables are defined as follows:
\begin{align}
    \Hc =& t\alphat\Hct - \frac{r_0}{6}, \label{Hc-def}\\
    \Sigma_{AB} =& t\alphat\Sigmat_{AB} - \biggl(\frac{1}{2}r_{AB}-\frac{r_0}{6}\delta_{AB}\biggr),\label{Sigma-def} \\
    \alpha =& t^{\ep_1}\alphat,\label{alpha-def} \\
    U_{A} =& t \alphat\Ut_A, \label{U-def} \\
    e_A^\Omega =& t^{\ep_2}\alphat\et_A^\Omega, \label{e-def}\\
    A_A =& t \alphat\At_A, \label{N-def} \\
    N^{AB} =& t \alphat\Nt^{AB}, \label{N-def}
\end{align}
where $\ep_1,\ep_2\in \Rbb$ are any constants satisfying 
\begin{gather}
\ep_1 + \frac{r_0}{2}> 0 \label{eigenval-posA}
\intertext{and}
0<\ep_2 < 1,\quad \ep_2+\frac{r_0}{2}-\frac{r_A}{2}>0. \label{eigenval-posB}
\end{gather}
Throughout this section, we assume that the Kasner exponents $q_I$ lie in the subcritical range \eqref{SCR}.

As the following lemma shows, if the Kasner exponents $q_I$ lie in the subcritical range \eqref{SCR}, then it is always possible to find a $\ep_2\in \Rbb$ that satisfies the inequalities \eqref{eigenval-posB}. 
\begin{lem} \label{lem:subcritical}
Suppose $\nu >0$, $P\in (0,\sqrt{1/3}]$ and the $q_{\Omega}\in\Rbb$, $\Omega=1,2,3$, satisfy the Kasner relations \eqref{Kasner-rels-A} and the subcritical condition 
\begin{equation}\label{subcritical}
    \max\limits_{\substack{\Omega,\Lambda,\Gamma \in \{1,2,3\} \\ \Omega < \Lambda}}\{q_\Omega+q_\Lambda-q_\Gamma\} < 1-\frac{2P\nu}{\sqrt{3}} < 1.
\end{equation}
Then the constants $r_0$ and $r_\Omega$, $\Omega=1,2,3$, defined by \eqref{Kasner-exps-A} satisfy 
\begin{equation*}
    \min_{\Omega,\Lambda \in \{1,2,3\}}\biggl\{r_\Omega + 1, 
    1 + \frac{r_\Omega}{2} + \frac{r_\Lambda}{2},1 + \frac{r_0}{2} - \frac{r_\Omega}{2}\biggr\} > \nu >0.
\end{equation*}
\end{lem}
\begin{proof}
In terms of the conformal Kasner exponents $r_\Omega$, the subcritical condition \eqref{subcritical} can be expressed as
\begin{equation}\label{Kasner-ineqs-B}
    \max\limits_{\Omega < \Lambda}\{r_\Omega+r_\Lambda-r_\Gamma\} < r_0 + 2-2\nu.
\end{equation}
Setting $\Omega=\Gamma$ and $\Lambda=\Gamma$ in \eqref{Kasner-ineqs-B} shows that $r_\Omega < r_0 + 2-2\nu$, or equivalently $1 + \frac{r_0}{2} - \frac{r_\Omega}{2} > \nu > 0$.
Since we are in four spacetime dimensions, $r_1+r_2-r_3$, $r_1+r_3-r_2$ and $r_2+r_3-r_1$ are the only combinations of $r_\Omega+r_\Lambda-r_\Gamma$ with distinct indices. Thus we deduce from Kasner relation \eqref{Kasner-rels-B} and \eqref{Kasner-ineqs-B} that $r_\Omega + 1 > \nu>0$,
which, in turn, implies that
$1 + \frac{r_\Omega}{2} + \frac{r_\Lambda}{2} > \nu > 0$.
\end{proof}

Using the renormalised variables \eqref{Hc-def}-\eqref{N-def} and the constraint equation \eqref{D-cnstr-2}, which we assume holds, we find, after a straightforward calculation, that the evolution equations \eqref{mEEb.1}-\eqref{mEEb.7} can be expressed as
\begin{align}
    \del{t}e_A^\Omega &= t^{-1}\biggl(\Bigl(\ep_2+\frac{r_0}{2}\Bigr)\delta_A^B-\frac{1}{2}r_A^B\biggr)e_B^\Omega + t^{-1}(2\Hc e_A^\Omega - \Sigma_A{}^B e_B^\Omega), \label{EEc.1} \\
    \del{t}A_A &= -(1+2\rho)t^{-\ep_2}e_A(\Hc) + \Bigl(\frac{1}{2}+\rho\Bigr)t^{-\ep_2}e_B(\Sigma_A^B) + t^{-1}\biggl(\Bigl(1+\frac{\rho+1}{2}r_0\Bigr)\delta_A^B-\frac{3\rho+1}{2}r_A^B\biggr)A_B \notag \\
        &\quad - \frac{1}{2}\rho t^{-1}\ep_{ABC}r_D^C N^{BD} + \rho t^{-1}\biggl(\Bigl(1+\frac{r_0}{2}\Bigr)\delta^B_A-\frac{1}{2}r_A^B\biggr)U_B \notag\\
        &\quad + t^{-1}[(U+A)*(\Hc+\Sigma)]_A + t^{-1}[N*\Sigma]_A, \label{EEc.2}\\
    \del{t}N^{AB} &= -t^{-\ep_2}\ep^{CD(A}e_C(\Sigma_D^{B)}) + t^{-1}\biggl(\delta^A_C\delta^B_D+\frac{1}{2}\delta^A_C\delta^B_D+\frac{1}{2}\delta^B_C r^A_D\biggr)N^{CD} + t^{-1}[(\Hc+\Sigma)*N]^{AB}, \label{EEc.3}\\
    \del{t}\Hc &= -\frac{4}{3}t^{-\ep_2}e_A(A^A) + \frac{1}{3}t^{-\ep_2}e_A(U^A)
        + t^{-1}[A*(A+U)] + t^{-1}[N*N],\label{EEc.4} \\
    \del{t}\Sigma_{AB} &= t^{-\ep_2}e_{\langle A}(U_{B\rangle }) - t^{-\ep_2}e_{\langle A}(A_{B\rangle }) + t^{-\ep_2}\ep_{CD\langle A}e^C(N_{B\rangle}^D) + t^{-1}[A*(U+N)]_{AB} + t^{-1}[N*N]_{AB}, \label{EEc.5}\\
    \del{t}\alpha &= \frac{1}{t}\biggl(\ep_1+\frac{r_0}{2}\biggr)\alpha + \frac{3}{t}\alpha\Hc, \label{EEc.6}\\
    \del{t}U_A &= (3-2\gamma)t^{-\ep_2}e_A(\Hc) + \gamma t^{-\ep_2}e_B(\Sigma_A{}^B) +t^{-1}(1+\gamma)\biggl(\Bigl(1+\frac{r_0}{2}\Bigr)\delta^B_A-\frac{1}{2}r_A^B\biggr)U_B \notag \\
        &\quad - \frac{1}{2}\gamma t^{-1}\ep_{ABC}r_D^C N^{BD} + \gamma t^{-1}\biggl(\frac{1}{2}r_0\delta_A^B-\frac{3}{2}r_A^B\biggr)A_B+ t^{-1}[(\Hc+\Sigma)*U]_A + t^{-1}[(A+N)*\Sigma]_A,\label{EEc.7}
\end{align}
where the $*$-notation is as defined above in Section \ref{Index-op}.

\begin{rem} \label{rem:Kasner-perturb}
The Kasner-scalar field family of solutions \eqref{Kasner-solns-B}, when expressed in terms of the renormalised variables \eqref{Hc-def}-\eqref{N-def}, become
\begin{equation} \label{Kasner-solns-C}
\bigl(\mathring{e}_A^\Omega, \mathring{\alpha}, \mathring{A}_A, \mathring{U}_A, \mathring{\Hc}, \mathring{\Sigma}_{AB}, \mathring{N}_{AB}\bigr) =  \bigl(t^{\ep_2+\frac{r_0}{2}-\frac{r_A}{2}}\delta^\Omega_A,t^{\ep_1+\frac{r_0}{2}},0,0,0,0,0\bigr). 
\end{equation}
Since $\ep_1$ and $\ep_2$ satisfy \eqref{eigenval-posA} and \eqref{eigenval-posB}, respectively, it is clear that
\begin{equation} \label{Kasner-solns-limit}
\lim_{t\searrow 0}\bigl(\mathring{e}_A^\Omega, \mathring{\alpha}, \mathring{A}_A, \mathring{U}_A, \mathring{\Hc}, \mathring{\Sigma}_{AB}, \mathring{N}_{AB}\bigr) =  (0,0,0,0,0,0).
\end{equation}
The significance of this property of the Kasner-scalar field solutions is that we can make these solutions as small as we like for $t\in (0,t_0]$ by choosing the initial time $t_0 > 0$ sufficiently small. This allows us to regard these solutions, along with small deviation from them, as small perturbations of the trivial solution. It is important to emphasise that no generality is lost by choosing $t_0 > 0$ small. This is because \eqref{Kasner-solns-C} defines an exact solution of the Einstein-scalar field equations. By Cauchy stability, we can evolve a perturbation at any finite time $T_0 > 0$ of a reference Kasner-scalar field solution to obtain a perturbation of the reference solution at any smaller time $t_0 \in (0, T_0)$. Furthermore, the perturbation at $t_0$ can be made arbitrarily small by ensuring the initial perturbation at $T_0$ is sufficiently small.
\end{rem}

\subsection{Renormalised  constraint equations}
When expressed in terms of the renormalised variables \eqref{Hc-def}-\eqref{N-def}, the constraint equations \eqref{A-cnstr-2}-\eqref{H-cnstr-2} become
\begin{align}
    \Af_{AB}^\Omega &:= t^{1+\ep_2}\alphat^2\tilde{\Af}_{AB}^\Omega = 2t^{1-\ep_2}e_{[A}(e_{B]}^\Omega) - 2U_{[A}e_{B]}^\Omega - 2A_{[A}e_{B]}^\Omega - \ep_{ABD}N^{DC}e_C^\Omega=0, \label{A-cnstr-3} \\
    \Bf_{AB} &:= t^2\alphat^2\tilde{\Bf}_{AB} = -2t^{1-\ep_2}e_{[A}(U_{B]}) + 2U_{[A}U_{B]} + 2A_{[A}U_{B]} + \ep_{ABC}N^{CD}U_D=0, \label{B-cnstr-3} \\
    \Cf^A &:= t^2\alphat^2\tilde{\Cf}^A = t^{-\ep_2}\bigl(e_B(N^{BA}) - \ep^{BAC}e_B(A_C)\bigr) - (U_B N^{BA}-\ep^{BAC}U_B A_C+2A_B N^{BA})=0, \label{C-cnstr-3} \\
    \Df_A &:= t^{1+\ep_1}\alphat^2\tilde{\Df}_A = \alpha U_A - t^{1-\ep_2}e_A(\alpha)=0, \label{D-cnstr-3} \\
    \Mf_A &:= t^2\alphat^2\Mft_A = t^{1-\ep_2}e_B(\Sigma_A^B) - 2t^{1-\ep_2}e_A(\Hc) + \Bigl(\frac{r_0}{2}+1-\frac{r_A}{2}\Bigr)U_A - \Bigl(\frac{3}{2}r_A-\frac{r_0}{2}\Bigr)A_A \notag \\
            &\quad - U_B\Sigma_A^B + 2U_A\Hc - 3A_B\Sigma_A^B - \ep_{ABC}N^{BD}\Sigma_D^C - \frac{1}{2}\ep_{ABC}r_D^C N^{BD}=0, \label{M-cnstr-3} \\
    \Hf &:= t^2\alphat^2\Hft = 4t^{1-\ep_2}e_A(A^A) + 6\Hc^2 + (6+2r_0)\Hc - \Sigma_{AB}\Sigma^{AB} - r_{AB}\Sigma^{AB} \notag \\
            &\quad - \Bigl(4U_A A^A+6A^A A_A+N^{AB}N_{AB}-\frac{1}{2}(N_A{}^A)^2\Bigr)=0. \label{H-cnstr-3}
\end{align}

\subsection{Low order equations\label{sec:low-order}}
The derivation of our Fuchsian formulation of the conformal tetrad Einstein-scalar field equations involves two different versions of the evolution equations \eqref{EEc.1}-\eqref{EEc.7}. Acknowledging that this violates the symmetry and positivity criteria of Section~\ref{sec:sym-hyp}, the first version is obtained by setting $\gamma=\rho=0$ in \eqref{EEc.1}-\eqref{EEc.7}, which we can write in matrix form as
\begin{equation} \label{FuchA}
    \del{t}W =\frac{1}{t^{\ep_2}} e_C^\Lambda B^C_{(0,0)} \del{\Lambda}W + \frac{1}{t}\Bc\Pbb W + \frac{1}{t}F
\end{equation}
where
\begin{gather} \label{W-def}
W = (e_P^\Sigma, \alpha, A_P, U_P, \Hc, \Sigma_{PQ}, N_{PQ})^{\tr}, \\
    \Pbb = \diag\Bigl(
    \delta_A^P\delta_\Sigma^\Omega, 1, \delta_A^P , \delta_A^P , 0 , 0 , \delta_A^P\delta_B^Q  
   \Bigr), \label{Pbb-def}\\
    F 
       = \begin{bmatrix}
            [(\Hc+\Sigma)*e] \\
            3\alpha\Hc \\
            [(U+A)*(\Hc+\Sigma)] \\
            [\Hc*U] \\
            [A*(A+U)] + [N*N] \\
            [A*(U+N)] + [N*N] \\
            [(\Hc+\Sigma)*N]
    \end{bmatrix} \label{F-def}
\intertext{and}
\Bc =\diag \biggl( \Bigl(\beta_2\delta^P_A -\frac{1}{2}r_A^P\Bigr)\delta^\Omega_\Sigma,  \beta_1,\beta_0\delta^P_A -\frac{1}{2}r_A^P,\beta_0\delta^P_A-\frac{1}{2}r_A^P,\beta_1, \beta_1\delta_A^P\delta_B^Q, \delta_A^P\delta_B^Q +\frac{1}{2}(r_A^P\delta^Q_B+\delta_A^P r_B^Q) 
    \biggr)  \label{Bc-def}\\
\intertext{with}
\beta_0 = 1+\frac{r_0}{2}, \quad 
\beta_1 = \ep_1+\frac{r_0}{2} \AND\beta_2 = \ep_2+\frac{r_0}{2}. \label{beta-defs}
\end{gather}
For use below, we also define the complementary projection operator to \eqref{Pbb-def} by
\begin{equation} \label{Pbb-perp-def}
    \Pbb^\perp = \id-\Pbb.
\end{equation}

By promoting the spatial derivatives to dynamical variables below, we can treat equation \eqref{FuchA} along with  spatial derivatives of it, see \eqref{FuchB} below, as ODEs for the \textit{low order} variables. Control on the low order variables is obtained from the singular term $\frac{1}{t}\Bc \Pbb W$. This term will turn out to be \emph{good} (in the sense of the Fuchsian theory \cite{BOOS:2021})  due to the matrix being $\Bc \Pbb$ diagonal and positive semi-definite. Indeed, from \eqref{rAB-def} and \eqref{Bc-def}, we observe that the matrix $\Bc$ is diagonal and has eigenvalues
$\beta_2-\frac{r_A}{2}$, $\beta_1$, $\beta_0-\frac{r_A}{2}$, $\beta_1, 1+\frac{r_A}{2}+\frac{r_B}{2}$. Since we are assuming that the Kasner exponents lie in the subcritical range \eqref{SCR}, these eigenvalues are easily seen to be positive due to \eqref{eigenval-posA}-\eqref{eigenval-posB}, \eqref{beta-defs} and Lemma \ref{lem:subcritical}.
Then by \eqref{Pbb-def}, it is clear that the range of $\Pbb$ corresponds to eigenspaces of the matrix $\Pbb \Bc$ with positive eigenvalues. It is worth noting here that the positivity of the eigenvalues of $\Pbb \Bc$ identify the \emph{decaying} variables 
\begin{equation} \label{Pbb-W}
\Pbb W =  (e_P^\Sigma, \alpha, A_P, U_P, 0, 0, N_{PQ})^{\tr}
\end{equation}
that will vanish in the limit $t\searrow 0$ (as we will justify by means of the Fuchsian theory in \cite{BOOS:2021}). On the other hand, the image of $\Pbb^\perp$ coincides with the kernel of $\Pbb \Bc$, and identifies the \emph{non-decaying} variables
\begin{equation} \label{Pbb-perp-W}
\Pbb^\perp W =  (0,0,0,0,\Hc,\Sigma_{PQ},0)^{\tr}
\end{equation}
that converge as $t\searrow 0$  to limits that do not, in general, vanish pointwise everywhere at $t=0$. 

Before proceeding, we note some properties of the nonlinear
source term \eqref{F-def} that will play an important role in our subsequent arguments. First, it is clear from \eqref{F-def} that $F$ is quadratic in $W$. Moreover, 
with the help of \eqref{Pbb-def}, \eqref{Pbb-perp-def}, \eqref{Pbb-W} and \eqref{Pbb-perp-W}, we observe that $F$ satisfies
\begin{equation} \label{F-structure-A}
\Pbb F = W *\Pbb W \AND
\Pbb^\perp F = \Pbb W * \Pbb W.
\end{equation}

Next, given that $\ep_2<1$, cf. \eqref{eigenval-posB}, we let $\nu>0$ be any positive number satisfying 
\begin{equation} \label{eigenval-posC}
   \ep_2+\nu<1,
\end{equation} 
and we introduce the \textit{low-order renormalised spatial derivatives of $W$} via
\begin{equation} \label{low-order-vars}
    W_\bc := t^{|\bc|\nu}\del{}^\bc W =   t^{|\bc|\nu}\bigl(\del{}^\bc e_P^\Sigma, \del{}^\bc \alpha, \del{}^\bc A_P,\del{}^\bc U_P,\del{}^\bc \Hc, \del{}^\bc \Sigma_{PQ}, \del{}^\bc N_{PQ}\bigr)^{\tr}, \quad |\bc|< k,
\end{equation}
where $\bc = (\bc_1,\bc_2,\bc_3)$ is a multiindex, $k\in \Nbb$ is a positive integer to be determined, and we set $W_{(0,0,0)}=W$. Collectively, we will refer to these derivatives as the \textit{low order variables}. 

Applying the operator $t^{|\bc|\nu}\del{}^{\bc}$ to \eqref{FuchA} yields
\begin{equation*}
    \del{t}W_\bc + t^{|\bc|\nu-\ep_2}B_{(0,0)}^C \biggl(\sum_{\bc'\leq\bc} \binom{\bc}{\bc'} \del{}^{\bc-\bc'}(e_C^\Lambda) \del{}^{\bc'}(\del{\Lambda}W)\biggr) = \frac{1}{t}(|\bc|\nu\id+\Bc\Pbb) W_\bc + \frac{1}{t}t^{|\bc|\nu}\del{}^{\bc}F,\quad |\bc|<k,
\end{equation*}
where in deriving this we made use of the fact, see
\eqref{B-rho-gamma} with $\upsilon=(0,0)$,
\eqref{Pbb-def} and \eqref{Bc-def}, that $B_{(0,0)}^C$, $\Bc$ and $\Pbb$ are constant in space. 
Since $F$ is quadratic in $W$ and satisfies \eqref{F-structure-A}, it is not difficult to verify using the product rule that we can express the above equation as
\begin{equation}  \label{FuchB}
    \del{t}W_\bc  = \frac{1}{t}(|\bc|\nu\id+\Pbb\Bc)W_\bc + \frac{1}{t}\sum_{\bc'\leq\bc} F_{\bc'} + \frac{1}{t^{\ep_2+\nu}}\sum_{\substack{\bc'\leq\bc\\|\ac|=1}}\Pbb W_{\bc-\bc'}*W_{\bc'+\ac},\quad |\bc|<k,
\end{equation}
where 
\begin{equation} \label{F-structure-B}
    F_{\bc'} =W_{\bc'} *W_{\bc-\bc'} 
\end{equation} and 
\begin{equation} \label{F-structure-C}
\Pbb F_{\bc'} = W_{\bc'}*\Pbb W_{\bc-\bc'} + \Pbb W_{\bc'} * W_{\bc-\bc'} \AND
\Pbb^\perp F_{\bc'} = \Pbb W_{\bc'}*\Pbb W_{\bc-\bc'}.
\end{equation}

\subsection{High order equations\label{sec:high-order}}
To obtain a closed system that yields energy estimates, we need a symmetric hyperbolic formulation of the equations for the \textit{high order variables}
\begin{equation} \label{high-order-vars}
    W_{\bc} :=t^{k\nu}\del{}^\bc W= t^{k\nu}\bigl(\del{}^\bc e_P^\Sigma, \del{}^\bc \alpha, \del{}^\bc A_P,\del{}^\bc U_P,\del{}^\bc \Hc, \del{}^\bc \Sigma_{PQ}, \del{}^\bc N_{PQ}\bigr)^{\tr}, \quad |\bc|=k,
\end{equation} 
which we obtain using the same method as in Section \ref{Sec:loc-exist}. Thus, we begin by setting  $\rho=0$ and $\gamma=2$ in \eqref{EEc.1}-\eqref{EEc.7}, which, according to Section~\ref{sec:sym-hyp} implies symmetrisability and positivity of the evolution system, and which gives
\begin{equation} \label{FuchC}
    \del{t}W - \frac{1}{t^{\ep_2}}e_C^\Lambda B_{(0,2)}^C \del{\Lambda}W = \frac{1}{t}\check{\Bc}W + \frac{1}{t}\Fch,
\end{equation}
where
\begin{equation} \label{Bc-check-def}
\check{\Bc} = \begin{bmatrix}
        \Bigl(\beta_2\delta^P_A -\frac{1}{2}r_A^P\Bigr)\delta^\Omega_\Sigma &  0&  0&  0&  0&  0&  0 \\
        0&  \beta_1&  0&  0&  0&  0&  0 \\
        0&  0&  \beta_0\delta^P_A -\frac{1}{2}r_A^P&  0&  0&  0&  0 \\
        0&  0&  r_0\delta^P_A-3r_A^P&  3\Bigl(\beta_0\delta^P_A-\frac{1}{2}r_A^P\Bigr)&  0&  0& \ep_A{}^{PB}r_B^Q \\
        0&  0&  0&  0&  0&  0&  0 \\
        0&  0&  0&  0&  0&  0&  0 \\
        0&  0&  0&  0&  0&  0&  \delta_A^P\delta_B^Q +\frac{1}{2}(r_A^P\delta^Q_B+\delta_A^P r_B^Q) 
    \end{bmatrix},
\end{equation}
and
\begin{equation} \label{F-check-def}
\Fch = \begin{bmatrix}
        [(\Hc+\Sigma)*e] \\
        3\alpha\Hc \\
        [(U+A)*(\Hc+\Sigma)] \\
        [(\Hc+\Sigma)*U] + [\Sigma*(A+N)] \\
        [A*(A+U)] + [N*N] \\
        [A*(U+N)] + [N*N] \\
        [(\Hc+\Sigma)*N]
    \end{bmatrix}.
\end{equation}
Spatially differentiating \eqref{FuchC} gives
\begin{equation} \label{FuchD}
    \del{t}W_\bc - \frac{1}{t^{\ep_2}}e_C^\Lambda B_{(0,2)}^C \del{\Lambda}W_\bc = \frac{1}{t}(k\nu\id+\check{\Bc}) W_\bc + \frac{1}{t}\sum_{\bc'\leq\bc} \Fch_{\bc'} + \frac{1}{t^{\ep_2+\nu}}\sum_{\substack{\bc'\leq\bc, \bc'\neq \bc\\|\ac|=1}}\Pbb W_{\bc-\bc'}*W_{\bc'+\ac},\quad |\bc|=k,
\end{equation}
where 
\begin{equation} \label{F-structure-D}
    \Fch_{\bc'} =W_{\bc'} *W_{\bc-\bc'} 
\end{equation} and 
\begin{equation} \label{F-structure-E}
\Pbb \Fch_{\bc'} = W_{\bc'}*\Pbb W_{\bc-\bc'} + \Pbb W_{\bc'} * W_{\bc-\bc'} \AND
\Pbb^\perp \Fch_{\bc'} = \Pbb W_{\bc'}*\Pbb W_{\bc-\bc'}.
\end{equation}

We now symmetrise the system \eqref{FuchD} in the same way as in Section~\ref{Sec:loc-exist}, namely, we set
\begin{equation} \label{Wch-bc-def}
    \Wch_\bc := V^{-1}W_\bc= t^{k\nu}\Bigl(\del{}^\bc e_P^\Sigma, \del{}^\bc \alpha, \frac{1}{3}\del{}^\bc A_P,-2\del{}^\bc A_P+\del{}^\bc U_P,\del{}^\bc \Hc, \del{}^\bc \Sigma_{PQ}, \del{}^\bc N_{PQ}\Bigr)^{\tr}, \quad |\bc|=k,
\end{equation}
where $V$ is defined above by \eqref{V-fix}, and multiply 
\eqref{FuchD} on the left by the matrix $\Sc$ defined by \eqref{Sc-fix} to get
\begin{equation} \label{FuchE}
    \Bh^0\del{t}\Wch_\bc + \frac{1}{t^{\ep_2}}e_C^\Lambda \Bh^C \del{\Lambda}\Wch_\bc =  \frac{1}{t}\Bsc\Wch_\bc + \frac{1}{t}\sum_{\bc'\leq\bc} \Sc\Fch_{\bc'} + \frac{1}{t^{\ep_2+\nu}}\sum_{\substack{\bc'\leq\bc, \bc'\neq \bc\\|\ac|=1}}\Pbb W_{\bc-\bc'}*W_{\bc'+\ac},\quad |\bc|=k,
\end{equation}
where $\Bh^0$ and $\Bh^C$  are defined above by \eqref{Bh0-def}  and \eqref{BhC-def}, respectively, and   
\begin{equation} \label{Bsc-fix}
    \Bsc = \text{\scalebox{0.63}{$ \begin{bmatrix}
    \biggl(\beta_2\delta^P_A -\frac{1}{2}r_A^P+k\nu\delta_A^P\biggr)\delta^\Omega_\Sigma &  0&  0&  0&  0&  0&  0 \\
        0&  \beta_1+k\nu&  0&  0&  0&  0&  0 \\
        0&  0&  \Bigl(\bigl(6+3r_0\bigr)\delta^P_A -3r_A^P\Bigr)+6k\nu\delta_A^P&  0&  0&  0&  0 \\
        0&  0&  (4+3r_0)\delta^P_A-5r_A^P&  \Bigl(\beta_0\delta^P_A-\frac{1}{2}r_A^P\Bigr)+\frac{1}{3}k\nu\delta_A^P&  0&  0&  \frac{1}{3}\ep_A{}^{PB}r_B^Q \\
        0&  0&  0&  0&  k\nu&  0&  0 \\
        0&  0&  0&  0&  0&  \frac{1}{3}k\nu\delta_A^P\delta_B^Q&  0 \\
        0&  0&  0&  0&  0&  0&  \frac{1}{3}(1+k\nu)\delta_A^P\delta_B^Q +\frac{1}{6}(r_A^P\delta^Q_B+\delta_A^P r_B^Q) 
    \end{bmatrix}$ }}.
\end{equation}

\begin{rem} The need for the low and higher order rescaled spatial derivatives can be understood in terms of the matrix $\Bsc$, which appears from symmetrising the tetrad equations. When $k=0$, the eigenvalues of $\frac{1}{2}(\Bsc^{\tr}+\Bsc)$ are negative for an open subset of the Kasner exponents $q_I$ that lie in the subcritical range \eqref{SCR}. Because of this, it is not possible to obtain energy estimates that are useful for controlling the behaviour of solutions near $t=0$ directly in terms of the undifferentiated fields for these Kasner exponents. On the other hand, by not symmetrising and viewing the lowest order equation \eqref{FuchA} as an ODE, the singular matrix $\Bc\Pbb$ does not have any negative eigenvalues. Moreover, as is clear from \eqref{FuchB}, the introduction of the low order renormalised variables \eqref{low-order-vars} leads to equations \eqref{FuchB} where the singular matrix $(|\bc|\nu\id +\Pbb \Bc)$ is positive definite with eigenvalues that grow as the number of derivatives is increased. For the symmetrised high order equation \eqref{FuchE}, the renormalised high order variables \eqref{Wch-bc-def} introduce a positive term into the matrix $\frac{1}{2}(\Bsc^{\tr}+\Bsc)$, as we show below, that makes it positive definite for any choice of subcritical Kasner exponents provided $k$ is chosen sufficiently large. The positive definiteness of  $\frac{1}{2}(\Bsc^{\tr}+\Bsc)$ is key to obtaining energy estimates that control the behaviour of solutions near $t=0$. 
\end{rem}
\begin{rem}
The need to use progressively higher derivatives as the Kasner exponents approach the boundary of the subcritical regime was also observed in the stability proofs of \cite{Fournodavlos_et_al:2023, Groeniger_et_al:2023}. However, in those works, the control of higher derivatives was achieved through a technical interpolation argument, rather than the simpler renormalization argument employed here. It is worth noting that optimal scattering results for the linearized Einstein-scalar field equations were obtained in \cite{Li:2024}. These results provide strong evidence that the need to control higher derivatives as the Kasner exponents approach the boundary of the subcritical region is unavoidable and not a shortcoming of the stability proof presented in this article or those from \cite{Fournodavlos_et_al:2023, Groeniger_et_al:2023}. 
\end{rem}

Since $\ep_1$ and $\ep_2$ satisfy the inequalities \eqref{eigenval-posA} and \eqref{eigenval-posB}, respectively, and $k\geq 1$, it is clear from \eqref{rAB-def}, \eqref{beta-defs}  and \eqref{Bsc-fix} that the symmetrised matrix $\frac{1}{2}(\Bsc^{\tr}+\Bsc)$ will be positive definite provided $\frac{1}{2}(\Bsc_*^{\tr}+\Bsc_*)$ is positive definite where $\Bsc_*$ is the subblock
\begin{equation*}
    \Bsc_* = \text{\scalebox{0.8}{$\begin{bmatrix}
        \Bigl(\bigl(6+3r_0\bigr)\delta^P_A -3r_A^P\Bigr)+6k\nu\delta_A^P&  0&  0&  0&  0 \\
        (4+3r_0)\delta^P_A-5r_A^P&  \Bigl(\beta_0\delta^P_A-\frac{1}{2}r_A^P\Bigr)+\frac{1}{3}k\nu\delta_A^P&  0&  0&  \frac{1}{3}\ep_A{}^{PB}r_B^Q \\
        0&  0&  k\nu&  0&  0 \\
        0&  0&  0&  \frac{1}{3}k\nu\delta_A^P\delta_B^Q&  0 \\
        0&  0&  0&  0&  \frac{1}{3}(1+k\nu)\delta_A^P\delta_B^Q +\frac{1}{6}(r_A^P\delta^Q_B+\delta_A^P r_B^Q)
    \end{bmatrix}$}}.
\end{equation*}
For any non-zero vector
\begin{equation*}
    \mathscr{V} = (A_P, U_P, \Hc, \Sigma_{PQ}, N_{PQ})^{\tr},
\end{equation*}
we have
\begin{align} \label{VBBV}
    \frac{1}{2}\Vsc^{\tr}(\Bsc_*^{\tr}+\Bsc_*)\Vsc &= \Isc_1 + \Isc_2 + k\nu \Hc^2 + \frac{1}{3}k\nu \Sigma_{PQ}\Sigma^{PQ},
\end{align}
where
\begin{align}
    \Isc_1 &= \sum_{P=1}^3 (6w_P+6k\nu) (A_P)^2 + \sum_{P=1}^3 (6w_P-2v_P) A_P U_P + \sum_{P=1}^3 w_P(U_P)^2, \label{Isc1-def} \\
    \Isc_2 &= \sum_{P=1}^3 \sum_{Q=1}^3 \frac{1}{3}(v_P+k\nu) (N_{PQ})^2 + \frac{1}{3}\ep^{APB}r_B^Q U_A N_{PQ} + \sum_{P=1}^3 \frac{1}{3}k\nu (U_P)^2, \label{Isc2-def}
\end{align}
and
\begin{equation*}
    w_P = 1+\frac{r_0}{2}-\frac{r_P}{2}, \quad v_P = 1+r_P.
\end{equation*}
We also note that 
\begin{equation} \label{simplify1}
    \frac{1}{3}\ep^{APB}r_B^Q U_A N_{PQ} = \frac{1}{3}\Bigl((r_3-r_2)U_1 N_{23} + (r_1-r_3)U_2 N_{13} + (r_2-r_1)U_3 N_{12}\Bigr),
\end{equation}
where we recall that $N_{AB}=N_{BA}$.
By \eqref{SCR} and Lemma \ref{lem:subcritical}, it follows that $w_P,v_P>0$, and consequently, the coefficients of $(A_P)^2$, $(U_P)^2$ and $(N_{PQ})^2$ in \eqref{Isc1-def} and \eqref{Isc2-def} are positive. Plugging the expansion \eqref{simplify1} into \eqref{Isc2-def} and using Young's inequality, $ab\leq a^2/2+b^2/2$, $a,b\geq 0$, we observe that $\Isc_1>0$ and $\Isc_2>0$ provided the inequalities 
\begin{align*}
    2\sqrt{(6w_P+6k\nu)w_P} &> |6w_P-2v_P|, \quad P=1,2,3,
    \intertext{and}
    2\sqrt{\frac{1}{3}(v_P+k\nu) \cdot \frac{1}{3}k\nu} &> \frac{1}{3}|v_P-v_Q|, \quad P,Q=1,2,3,\; P<Q,
\end{align*}
are satisfied. Thus, by choosing  $k\in\Nbb$ large enough so that
\begin{align} \label{k-nu-fix}
    k\nu &> \max \biggl\{0, \frac{(3w_P-v_P)^2}{6w_P}-w_P, \frac{1}{2}\Bigl(-v_P+\sqrt{v_P^2+(v_P-v_Q)^2}\Bigr)\biggr\} \notag \\
    &= \max \biggl\{0, \frac{(4+3r_0-5r_P)^2}{4(6+3r_0-3r_P)}-1-\frac{r_0}{2}+\frac{r_P}{2}, \frac{1}{2}\Bigl(-r_P-1+\sqrt{2r_P^2+r_Q^2-2r_Pr_Q+2r_P+1}\Bigr)\biggr\},
\end{align}
we can ensure that the matrix $\frac{1}{2}(\Bsc_*^{\tr}+\Bsc_*)$ is positive definite. 

\begin{rem}
The choice $k=0$ is also allowed provided the right hand side of \eqref{k-nu-fix} is $0$, and hence that  $\frac{1}{2}(\Bsc_*^{\tr}+\Bsc_*)\bigl|_{k=0} \geq 0$. Since the conformal Kasner exponents all vanish, ie. $r_i=0$, $i=0,1,2,3$, for the FLRW solution, it is clear from \eqref{k-nu-fix} that $k=0$ is a valid choice in this case.  
\end{rem}

\subsection{The Fuchsian system\label{sec:total-Fuch-sys}}
Gathering together \eqref{FuchB} and \eqref{FuchE}, we arrive at the desired Fuchsian formulation of the tetrad equations:
\begin{equation} \label{FuchFinal-A}
    \Bv^0 \del{t}\Wv + \frac{1}{t^{\ep_2}}\Bv^\Lambda(\Wv) \del{\Lambda}\Wv = \frac{1}{t}\Bcv\Pv\Wv + \frac{1}{t^{\ep_2+\nu}}\Gv(\Wv) + \frac{1}{t}\Fv(\Wv)
\end{equation}
where
\begin{align} 
    \Wv &= \bigl( (W_{\bc})_{|\bc|= 0},(W_{\bc})_{|\bc|=1},(W_{\bc})_{|\bc|=2}, \ldots,(W_{\bc})_{|\bc|=k-1},(\Wch_\bc)_{|\bc|=k}\bigr)^{\tr}, \label{Wv-def}\\
    \Bv^0 &= \diag\bigl(\id,\id,\id,\ldots,\id,\Bh^0\bigr), \label{Bv0-def}\\
    \Bv^\Lambda(\Wv) &= \diag\bigl(0,0,0,\ldots,0,e_C^\Lambda\Bh^C\bigr), \label{Bv-Lambda-def}\\
    \Bcv &= \diag\bigl(\Bc,\nu\id+\Bc\Pbb,2\nu\id+\Bc\Pbb,...,(k-1)\nu\id+\Bc\Pbb,\Bsc\bigr),
    \label{Bcv-def} \\
    \Pv &= \diag\bigl(\Pbb,\id,\id,\ldots,\id,\id), \label{Pv-def}
\end{align}
and $\Gv(\Wv)$ and $\Fv(\Wv)$ are quadratic in $\Wv$.

\subsection{Coefficient properties\label{coeff-props}} In order to be able to apply the Fuchsian global existence theory from \cite{BOOS:2021}, we need to verify the coefficients of the Fuchsian equation \eqref{FuchFinal-A} satisfies the conditions set out in \cite[\S 3.4]{BOOS:2021}.  
To start, we observe from \eqref{Pbb-def} that $\Pv$ is a constant, symmetric projection matrix, ie.
\begin{equation}\label{Pv-props}
\Pv^2 = \Pv, \quad \Pv^{\tr}=\Pv \AND \del{t}\Pv = \del{\Lambda}\Pv = 0. 
\end{equation}
Letting
\begin{equation}\label{Pv-perp-def}
\Pv^\perp := \id -\Pv = \diag\bigl(\Pbb^\perp,0,0,\ldots,0,0),
\end{equation}
we see from \eqref{Bv0-def} that the matrix  $\Bv^0$ satisfies 
\begin{equation}\label{Bv0-prop}
\Pv \Bv^0 \Pv^\perp = \Pv^\perp \Bv^0 \Pv = 0.
\end{equation}
Moreover, it is clear from \eqref{Bh0-def} and \eqref{Bv0-def} that $\Bv^0$ is symmetric and satisfies
\begin{equation}\label{Bv0-lbnd}
\frac{1}{3}\id \leq \Bv^0\leq 6\id.
\end{equation}
It is also clear from \eqref{W-def}-\eqref{Pbb-def} and \eqref{Pv-def} that 
the matrices $\Bv^\Lambda(\Wv)$ satisfy
\begin{equation}\label{BvLambda-prop}
\Bv^\Lambda(\Wv) = \Bv^\Lambda(\Pv \Wv) = \Ord(|\Pv \Wv|).
\end{equation}

Next, by \eqref{Sc-fix} and \eqref{Pbb-def}, we observe that $[\Pbb,\Sc]=0$.
Using this, it is straightforward to verify from  \eqref{FuchB}-\eqref{F-structure-C}, \eqref{F-structure-D}-\eqref{FuchE} and \eqref{Pv-def}-\eqref{Pv-perp-def} that the quadratic nonlinear term $\Gv(\Wv)$ and $\Fv(\Wv)$ satisfies
\begin{gather}
\Gv(\Wv) = \Ord\bigl(|\Wv||\Pv\Wv|\bigr) \label{Gv-prop},\\
\Pv \Fv(\Wv)=\Ord\bigl(|\Wv||\Pv\Wv|\bigr) \label{Fv-prop-A}
\intertext{and}
\Pv^\perp \Fv(\Wv) = \Ord\bigl(|\Pv\Wv|^2\bigr). \label{Fv-prop-B}
\end{gather}
Using \eqref{Fv-prop-A}, we deduce that 
\begin{equation*}
\Pv\Fv(\Wv) = \Pv \Lv(\Wv)\Pv\Wv
\end{equation*}
for some matrix-valued map $\Lv(\Wv)$ that depends linearly on $\Wv$, ie.
\begin{equation} \label{L-bnd}
\Lv(\Wv)= \Ord(|\Wv|).
\end{equation}
Using this allows us to express \eqref{FuchFinal-A} as
\begin{equation} \label{FuchFinal-B}
    \Bv^0 \del{t}\Wv + \frac{1}{t^{\ep_2}}\Bv^\Lambda \del{\Lambda}\Wv = \frac{1}{t}\bigl(\Bcv+\Pv \Lv(\Wv)\Pv\bigr)\Pv\Wv + \frac{1}{t^{\ep_2+\nu}}\Gv(\Wv) + \frac{1}{t}\Pv^\perp\Fv(\Wv).
\end{equation}

From \eqref{Pbb-def} and \eqref{Bc-def}, we have $[\Bc,\Pbb]=0$. Using this, it follows immediately from the definition \eqref{Bcv-def} and \eqref{Pv-def} that
\begin{equation} \label{Bcv-Pv-com-A}
[\Bcv,\Pv]=0,
\end{equation}
which, in turn, implies that
\begin{equation} \label{Bcv-Pv-com-B}
[\Pv,\Bcv+\Pv \Lv(\Wv)\Pv]=0.
\end{equation}
Moreover, for any choice of subcritical Kasner exponents $q_I$, see \eqref{SCR}, $\ep_1$ and $\ep_2$ satisfying \eqref{eigenval-posA} and \eqref{eigenval-posB}, respectively, $\nu>0$ satisfying \eqref{eigenval-posC}, and $k\in \Nbb$ satisfying \eqref{k-nu-fix}, we know from the arguments above in Section \ref{sec:high-order} that there exist $\tilde{\kappa},\tilde{\gamma} > 0$ such that $\tilde{\kappa}^{-1}\id \leq \Bc \leq \tilde{\gamma}$, which, in turn, implies by \eqref{Bcv-def} and \eqref{Bv0-lbnd} that
\begin{equation} \label{Bv0-Bcv-bnds}
\frac{1}{3}\id \leq \Bv^0\leq  \frac{1}{\kappa} \Bcv \leq \gamma
\end{equation}
for some $\gamma,\kappa>0$.

Finally, defining the \textit{matrix divergence} via
\begin{equation*}
\Div \!\Bv = \frac{1}{t^{\ep_2}}\partial_\Lambda (\Bv^\Lambda(\Wv)),
\end{equation*}
it follows directly from \eqref{Bv-Lambda-def} that 
\begin{equation}\label{DivBv-A}
\Div\! \Bv = \Ord(t^{-\ep_2}|\del{}\Wv|). 
\end{equation}
It is worth noting the relation
\begin{equation*}
\Div\! \Bv= \partial_{t}\Bv^0 + \partial_\Lambda \biggl(\frac{1}{t^{\ep_2}}\Bv^\Lambda(\Wv)\biggr)
\end{equation*}
holds
for solutions of \eqref{FuchFinal-B}, which justifies the term \textit{matrix divergence}.

\subsection{Past Fuchsian stability}
We now turn to establishing the past, ie. on time intervals of the form $(0,t_0]$, nonlinear stability of the trivial solution $\Wv\equiv 0$ to the Fuchsian equation \eqref{FuchFinal-B}. This stability result is accompanied by energy and decay estimates that provide the main technical tools needed to prove the past nonlinear stability of the Kasner-scalar field solutions with subcritical exponents. 

\subsubsection{Past global-in-space Fuchsian stability}
The global-in-space past stability of the trivial solution $\Wv\equiv 0$ to \eqref{FuchFinal-B} means establishing the existence and uniqueness of solutions to the Fuchsian global initial value problem (GIVP)  
\begin{align}
   \Bv^0 \del{t}\Wv + \frac{1}{t^{\ep_2}}\Bv^\Lambda \del{\Lambda}\Wv &= \frac{1}{t}\bigl(\Bcv+\Pv \Lv(\Wv)\Pv\bigr)\Pv\Wv + \frac{1}{t^{\ep_2+\nu}}\Gv(\Wv) + \frac{1}{t}\Pv^\perp\Fv(\Wv) \hspace{0.5cm} \text{in $M_{0,t_0}=(0,t_0]\times \Tbb^3$,} \label{Fuch-past-A.1}\\
  \Wv &= \Wv_0 \hspace{7.85cm} \text{in $\Sigma_{t_0}=\{t_0\}\times \Tbb^3$,}
  \label{Fuch-past-A.2}
\end{align}
under a suitable smallness assumption on the initial data $\Wv_0$, which we do in the following proposition. 

\begin{prop}
  \label{prop:Fuch-global}
Suppose $T_0>0$, $k_0\in \Zbb_{>\frac{3}{2}+1}$, the Kasner exponents $q_I\in \Rbb$, $I=1,2,3$, and  $P\in (0,\sqrt{1/3})$ satisfy the Kasner relations \eqref{Kasner-rels-A} and the subcritical condition \eqref{SCR},  $\ep_1\in (0,1)$ and $\ep_2\in \Rbb$ satisfy \eqref{eigenval-posA}  and \eqref{eigenval-posB}, respectively, where $r_0$ and $r_I$, $I=1,2,3$, are defined by \eqref{Kasner-rels-B}, $\nu>0$ satisfies \eqref{eigenval-posC}, and $k\in \Nbb_0$ is chosen large enough so that \eqref{k-nu-fix} holds.\footnote{If the right hand side of \eqref{k-nu-fix} is $0$, then $k=0$ is allowed.}     
Then there exists
a $\delta_0 > 0$ such that for every $t_0\in (0,T_0]$ and $\Wv_0\in H^{k_0}(\Tbb^3)$ satisfying
$\norm{\Wv_0}_{H^{k_0}(\mathbb T^{3})}< \delta_0$
the GIVP \eqref{Fuch-past-A.1}-\eqref{Fuch-past-A.2} admits a unique solution 
\begin{equation*}
\Wv \in C^0_b\bigl((0,t_0],H^{k_0}(\mathbb T^{3})\bigr)\cap C^1\bigl((0,t_0],H^{k_0-1}(\mathbb T^{3})\bigr)
\end{equation*}
such that the limit $\lim_{t\searrow 0} \Pv^\perp \Wv(t)$, denoted $\Pv^\perp \Wv(0)$, exists in $H^{k_0-1}(\mathbb{T}^3)$.

\noindent Moreover, the following hold:
\begin{enumerate}[(a)]
    \item For any $t_1\in (0,t_0)$, $\Wv$ is the unique $C^1$-solution of \eqref{Fuch-past-A.1} on $M_{t_1,t_0}$ satisfying the initial condition $\Wv|_{t=t_0}=\Wv_0$.  
    \item There exists a $\zeta>0$ such that $\Wv$ satisfies the energy estimate
\begin{equation*}
\norm{\Wv(t)}_{H^{k_0}(\mathbb T^{3})}^2 + \int^{t_0}_t \frac{1}{s} \norm{\Pv \Wv(s)}_{H^{k_0}(\mathbb T^{3})}^2\, ds  \lesssim \norm{\Wv_0}^2_{H^{k_0}(\mathbb T^{3})}
\end{equation*}
and decay estimate
\begin{equation*}
  \norm{\Pv \Wv(t)}_{H^{k_0-1}(\mathbb T^{3})}+\norm{\Pv^\perp \Wv(t) - \Pv^\perp \Wv(0)}_{H^{k_0-1}(\mathbb T^{3})} \lesssim t^\zeta
\end{equation*}
for all $t\in(0,t_0]$.
\end{enumerate}
\medskip

\noindent The implicit constants and $\zeta$ in the energy and decay estimates are independent of the choice of $t_0\in (0,T_0]$ and $\Wv_0$ satisfying $\norm{\Wv_0}_{H^{k_0}(\Tbb^3)}<\delta_0$.
\end{prop}
\begin{proof}
Due to \eqref{Pv-props}, \eqref{Bv0-prop}, \eqref{BvLambda-prop}, \eqref{Gv-prop}, \eqref{Fv-prop-B}, \eqref{L-bnd}, \eqref{Bcv-Pv-com-A}, \eqref{Bcv-Pv-com-B}, \eqref{Bv0-Bcv-bnds} and \eqref{DivBv-A}, there exists, for any $T_0>0$, $R>0$, constants $\lambda_1,\lambda_2,\lambda_3>0$, $\theta>0$ such that the coefficients of the Fuchsian equation \eqref{Fuch-past-A.1} satisfy all the conditions from \cite[\S 3.4]{BOOS:2021} where the remaining constants from \cite[\S 3.4]{BOOS:2021} are chosen as follows: $\kappa$ and $\gamma$ are as given above in \eqref{Bv0-Bcv-bnds}, $\gamma_1=\frac{1}{3}$, $\gamma_2=\gamma$, $\alpha=0$, $\beta_\ell=0$, $\ell=0,1,\ldots,7$, and $p=\min\bigl\{1-\ep_2,1-(\ep_2+\nu)/2\bigr\}>0$. Consequently, for any $t_0\in (0,T_0]$, we obtain
from\footnote{Note that the regularity condition $k_0>\frac{3}{2}+3$ required to apply \cite[Thm.~3.8]{BOOS:2021} can be lowered to $k_0>\frac{3}{2}+1$ in our application because no singular terms of the form $t^{-1}$ appear in the spatial derivative matrix $t^{-\ep_2}\Bv^\Lambda$ since $\ep_2<1$ by assumption. This allows us to avoid the use of \cite[Lemma~3.5]{BOOS:2021} in the proof of \cite[Thm.~3.8]{BOOS:2021}, which yields the reduction in the required regularity from $k_0>\frac{3}{2}+3$ to $k_0>\frac{3}{2}+1$.} \cite[Thm.~3.8]{BOOS:2021} and the discussion in \cite[\S 3.4]{BOOS:2021} the existence of a constant
$\delta_0 > 0$ such that if $\norm{\Wv_0}_{H^{k_0}(\mathbb \Tbb^{3})}< \delta_0$,
then there exists a unique solution 
\begin{equation*}
\Wv \in C^0_b\bigl((0,t_0],H^{k_0}(\mathbb \Tbb^{3})\bigr)\cap C^1\bigl((0,t_0],H^{k_0-1}(\mathbb \Tbb^{3})\bigr)
\end{equation*}
of the GIVP \eqref{Fuch-past-A.1}-\eqref{Fuch-past-A.2}. Moreover, the limit $\lim_{\searrow 0} \Pv^\perp \Wv(t)$, denoted $\Pv^\perp \Wv(0)$, exists in $H^{k_0-1}(\mathbb \Tbb^{3})$, and there exists a $\zeta>0$ and such that the energy estimate \begin{equation*}
\norm{\Wv(t)}_{H^{k_0}(\mathbb T^{3})}^2 + \int^{t_0}_t \frac{1}{s} \norm{\Pv \Wv(s)}_{H^{k_0}(\mathbb T^{3})}^2\, ds  \lesssim \norm{\Wv_0}^2_{H^{k_0}(\mathbb T^{3})}
\end{equation*}
and decay estimate
\begin{equation*}
  \norm{\Pv \Wv(t)}_{H^{k_0-1}(\mathbb T^{3})}+\norm{\Pv^\perp \Wv(t) - \Pv^\perp \Wv(0)}_{H^{k_0-1}(\mathbb T^{3})} \lesssim t^\zeta
\end{equation*}
hold for all $t\in(0,t_0]$.

We further observe that the constant \(\delta_0 > 0\) does not depend on the choice of \(t_0 \in (0, T_0]\). This follows because, for any \(t_0 \in (0, T_0]\), the coefficient assumptions remain satisfied on \((0, t_0]\) with the \emph{same constants} as on the larger interval \((0, T_0]\). Consequently, it is straightforward to verify that all explicit and implicit constants in the estimates used in the proof of \cite[Thm.~3.8]{BOOS:2021} are independent of \(t_0\in (0, T_0] \). Since \(\delta_0\) is determined by these constants, it follows that the same \(\delta_0 > 0\) ensures existence on any subinterval \((0, t_0]\) as it does on \((0, T_0]\). 

To complete the proof, we recall, see \cite[Thm.~2.1]{Majda:1984}, that $C^1$-solutions of symmetric hyperbolic equations are unique. Thus, restricting the solution $\Wv$ of the GIVP \eqref{Fuch-past-A.1}-\eqref{Fuch-past-A.2} to $M_{t_1,t_0}$, where $t_1\in (0,t_1)$, yields the unique $C^1$-solution of \eqref{Fuch-past-A.1} on $M_{t_1,t_0}$ satisfying the initial condition $\Wv|_{t=t_0}=\Wv_0$. 
\end{proof}

\bigskip
\cnote{red}
\bigskip

\begin{rem} \label{rem:lower-regularity}
In Proposition \ref{prop:Fuch-global}, the regularity requirement on the initial data is certainly not optimal. Indeed, we have no reason to believe that the condition \eqref{k-nu-fix} is optimal in the sense of being necessary to establish the existence of solutions globally to the past. 
Furthermore, in the proof of Proposition \ref{prop:Fuch-global}, we treat $W$ and its derivatives as independent variables. However, for our application, we are only interested in solutions to \eqref{Fuch-past-A.1} that take the form \eqref{Wv-def}. This implies that we do not need to estimate all the components of \eqref{Wv-def} in the same function spaces as we do in the proof of Proposition \ref{prop:Fuch-global}. For instance, only the components corresponding to the top order would need to be estimated in $L^2$. While this approach would lower the differentiability requirements on the initial data, we do not adopt it here, as it would significantly complicate the proof of Proposition \ref{prop:Fuch-global} and would still not yield an optimal result. 
In any case, we do not pursue the question of optimal differentiability further and instead leave it for future work.
\end{rem}

\subsubsection{Past local-in-space Fuchsian stability}
With the past global-in-space stability problem resolved, we now turn to proving the past local-in-space stability of the trivial solution $\Wv \equiv 0$ to the Fuchsian system \eqref{FuchFinal-B}. By this, we mean establishing the existence and uniqueness of solutions, under a suitable smallness assumption on the initial data $\Wv_0$, to the following IVP on a truncated cone domain that terminates at $t = 0$:
\begin{align}
   \Bv^0 \del{t}\Wv + \frac{1}{t^{\ep_2}}\Bv^\Lambda \del{\Lambda}\Wv &= \frac{1}{t}\bigl(\Bcv+\Pv \Lv(\Wv)\Pv\bigr)\Pv\Wv + \frac{1}{t^{\ep_2+\nu}}\Gv(\Wv) + \frac{1}{t}\Pv^\perp\Fv(\Wv) \hspace{0.5cm} \text{in $\Omega_{\qv}$,} \label{Fuch-past-B.1}\\
  \Wv &= \Wv_0 \hspace{7.85cm} \text{in $\{t_0\}\times \mathbb{B}_{\rho_0}$,}
  \label{Fuch-past-B.2}
\end{align}
where $\qv=(0,t_0,\rho_0,\rho_1,\ep_2)$, $\ep_2$ satisfies \eqref{eigenval-posB}, 
$t_0>0$, $\rho_0\in (0,L)$, and
$\rho_1 = t_0^{\ep_2-1}(1-\ep_2)(1-\vartheta)\rho_0$ with $\vartheta\in (0,1)$. We note the truncated cone $\Omega_{\qv}$, see \eqref{Omega-def}, in this IVP is capped above by $\{t_0\}\times \mathbb{B}_{\rho_0}$ and below by $\{0\}\times \mathbb{B}_{\vartheta \rho_0}$.

\begin{prop}
  \label{prop:Fuch-global-loc}
Suppose $T_0>0$, $k_0\in \Zbb_{>\frac{3}{2}+1}$ the Kasner exponents $q_I\in \Rbb$, $I=1,2,3$, and  $P\in (0,\sqrt{1/3})$ satisfy the Kasner relations \eqref{Kasner-rels-A} and the subcritical condition \eqref{SCR},  $\ep_1\in (0,1)$ and $\ep_2\in \Rbb$ satisfy \eqref{eigenval-posA}  and \eqref{eigenval-posB}, respectively, where $r_0$ and $r_I$, $I=1,2,3$, are defined by \eqref{Kasner-rels-B}, $\nu>0$ satisfies \eqref{eigenval-posC}, $k\in \Nbb_0$ is chosen large enough so that \eqref{k-nu-fix} holds\footnote{If the right hand side of \eqref{k-nu-fix} is $0$, then $k=0$ is allowed.},  
$0<\rho_0<L$ and $\vartheta \in (0,1)$. Then there exists
a $\delta_0 > 0$ such that for every $t_0\in (0,T_0]$ and $\Wv_0\in H^{k_0}\bigl(\mathbb{B}_{\rho_0}\bigr)$ satisfying
$\norm{\Wv_0}_{H^{k_0}(\mathbb{B}_{\rho_0})}< \delta_0$
the GIVP \eqref{Fuch-past-B.1}-\eqref{Fuch-past-B.2} admits a unique classical solution $\Wv\in C^1(\grave{\Omega}_{\qv})$ where $\qv=(0,t_0,\rho_0,\rho_1,\ep_2)$ and
$\rho_1 = t_0^{\ep_2-1}(1-\ep_2)(1-\vartheta)\rho_0$. Moreover, 
the following hold:
\begin{enumerate}[(a)]
\item Letting $\rho(t) = \rho_0 + (1-\vartheta) \rho_0 \bigl( \bigl(\frac{t}{t_0}\bigr)^{1-\ep_2}-1\bigr)$, then
$\Wv(t) \in H^{k_0}(\mathbb{B}_{\rho(t)})$ and $\del{t}\Wv(t)\in H^{k_0-1}(\mathbb{B}_{\rho(t)})$ for each $t\in (0,t_0]$.
\item The limit  $\lim_{t\searrow 0} \Pv^\perp \Wv|_{\{t\}\times \mathbb{B}_{\vartheta\rho_0}}$, denoted $\Pv^\perp \Wv(0)$, exists in $H^{k_0-1}(\mathbb{B}_{\vartheta \rho_0})$.
\item There exists a $\zeta>0$ such that $\Wv$ satisfies the energy estimate
\begin{equation*}
\norm{\Wv(t)}_{H^{k_0}(\mathbb{B}_{\rho(t)})}^2 + \int^{t_0}_t \frac{1}{s} \norm{\Pv \Wv(s)}_{H^{k_0}(\mathbb{B}_{\rho(s)})}^2\, ds  \lesssim \norm{\Wv_0}^2_{H^{k_0}(\mathbb{B}_{\rho_0})}
\end{equation*}
and decay estimate
\begin{equation*}
  \norm{\Pv \Wv(t)}_{H^{k_0-1}(\mathbb{B}_{\vartheta\rho_0})}+\norm{\Pv^\perp \Wv(t) - \Pv^\perp \Wv(0)}_{H^{k_0-1}(\mathbb{B}_{\vartheta\rho_0})} \lesssim t^\zeta
\end{equation*}
for all $t\in(0,t_0]$.
\item Letting $e=(e^\Sigma_P)$
be variable from the vector component
\begin{equation*}
(W_{\bc})_{|\bc|= 0} = W \overset{\eqref{W-def}}{=}  (e_P^\Sigma, \alpha, A_P, U_P, \Hc, \Sigma_{PQ}, N_{PQ})^{\tr}
\end{equation*}
of $\Wv$, see \eqref{Wv-def}, then
\begin{equation*}
\sup_{(t,x)\in M_{0,t_0}}|e(t,x)| \leq \frac{(1-\ep_2)(1-\theta)}{10 T_0^{1-\ep_2}}< \frac{\rho_1}{9},
\end{equation*}
and for any $t_1\in (0,t_0)$, $\Wv$ is the unique $C^1(\grave{\Omega}_{\qv_1})$-solution of \eqref{Fuch-past-B.1} on $\Omega_{\qv_1}$, $\qv_1=(t_1,t_0,\rho_0,\rho_1,\ep_2)$, satisfying the initial condition $\Wv|_{\{t_0\}\times\mathbb{B}_{\rho_0}}=\Wv_0$. 
\end{enumerate}

\noindent The implicit constants and $\zeta$ in the energy and decay estimates are independent of the choice of $t_0\in (0,T_0]$ and $\Wv_0$ satisfying $\norm{\Wv_0}_{H^{k_0}(\Tbb^3)}<\delta_0$.
\end{prop}
\begin{proof}
$\;$
\subsubsection*{Existence}
Given initial data $\Wv_0\in H^{k_0}(\mathbb{B}_{\rho_0})$, we can use the extension operator $\Ebb_{\rho_0}$, see \eqref{Ebb-def}-\eqref{Ebb-prop}, to define 
\begin{equation}\label{hatWv0-def}
    \hat{\Wv}_0=\Ebb_{\rho_0}\Wv_0 \in H^{k_0}(\Tbb^3),
\end{equation} 
which satisfies
\begin{equation}\label{hatWv0-bnd}
\norm{\hat{\Wv}_0}_{H^{k_0}(\Tbb^3)}\leq C_0 \norm{\Wv_0}_{H^{k_0}(\mathbb{B}_{\rho_0})}
\end{equation}
for some constant $C_0>0$ independent of the choice of initial data $\Wv_0$. Then by Proposition \ref{prop:Fuch-global}, we know that there exists a $\hat{\delta}_0>0$ such that for any $t_0\in (0,T_0]$, if 
\begin{equation} \label{Wv0-smallness}
\norm{\Wv_0}_{H^{k_0}(\mathbb{B}_{\rho_0})} <  \frac{\hat{\delta}_0}{C_0},
\end{equation}
then there exists a unique solution
\begin{equation}\label{hatWv-reg}
\hat{\Wv} \in C^0_b\bigl((0,t_0],H^{k_0}(\mathbb T^{3})\bigr)\cap C^1\bigl((0,t_0],H^{k_0-1}(\mathbb T^{3})\bigr) \subset C^1(M_{0,t_0})
\end{equation}
of the the Fuchsian equation \eqref{Fuch-past-A.1} on $M_{0,t_0}$ satisfying 
$\hat{\Wv}|_{t=t_0}=\hat{\Wv}_0$. Moreover, we know from Proposition \ref{prop:Fuch-global} that the limit 
\begin{equation}\label{hatWv-limit}
   \lim_{t\searrow 0} \Pv^\perp \hat{\Wv}(t)= \Pv^\perp \hat{\Wv}(0)
\end{equation} exists in $H^{k_0-1}(\mathbb{T}^3)$, and there exists a $\zeta>0$ such that the energy estimate
\begin{equation} \label{hatWv-energy}
\norm{\hat{\Wv}(t)}_{H^{k_0}(\mathbb T^{3})}^2 + \int^{t_0}_t \frac{1}{s} \norm{\Pv \hat{\Wv}(s)}_{H^{k_0}(\mathbb T^{3})}^2\, ds  \lesssim \norm{\hat{\Wv}_0}^2_{H^{k_0}(\mathbb T^{3})} 
\end{equation}
and decay estimate
\begin{equation} \label{hatWv-decay}
  \norm{\Pv \hat{\Wv}(t)}_{H^{k_0-1}(\mathbb T^{3})}+\norm{\Pv^\perp \hat{\Wv}(t) - \Pv^\perp \hat{\Wv}(0)}_{H^{k_0-1}(\mathbb T^{3})} \lesssim t^\zeta
\end{equation}
hold for all $t\in(0,t_0]$. Noting that $\hat{\Wv}|_{\mathbb{B}_{\rho_0}}=\Wv_0$ by \eqref{Ebb-prop} and \eqref{hatWv0-def}, the  restriction 
\begin{equation}\label{hatWv-restrict}
\Wv:= \hat{\Wv}|_{\Omega_{\qv}}
\end{equation}
yields a $C^1(\grave{\Omega}_{\qv})$-solution of the GIVP \eqref{Fuch-past-B.1}-\eqref{Fuch-past-B.2}. This establishes the existence of a classical solution. Uniqueness will follow from the proof of statement (d) below.

\subsubsection*{Proof of statements (a) and (b)} These statements follow directly from the definition \eqref{hatWv-restrict}, the regularity \eqref{hatWv-reg} of the solution $\hat{\Wv}$, the limit \eqref{hatWv-limit}, and the definitions \eqref{Omega-def} and \eqref{Omega-grave-def} of the truncated cone domains $\Omega_{\qv}$ and $\grave{\Omega}_{\qv}$.

\subsubsection*{Proof of statement (c)} From \eqref{hatWv0-bnd}, \eqref{hatWv-energy}, and \eqref{hatWv-restrict}, the fact that the matrix $\Pv$ is constant, see \eqref{Pbb-def} and \eqref{Pv-def}, and the definition \eqref{Omega-def} of the truncated cone domains $\Omega_{\qv}$, we immediately deduce that the inequalities
\begin{align}
\norm{\Wv(t)}_{H^{k_0}(\mathbb{B}_{\rho(t)})}^2 + \int^{t_0}_t \frac{1}{s} \norm{\Pv \Wv(s)}_{H^{k_0}(\mathbb{B}_{\rho(s)})}^2\, ds &\leq 
\norm{\hat{\Wv}(t)}_{H^{k_0}(\mathbb T^{3})}^2 + \int^{t_0}_t \frac{1}{s} \norm{\Pv \hat{\Wv}(s)}_{H^{k_0}(\mathbb T^{3})}^2\, ds \notag \\ 
&\lesssim \norm{\hat{\Wv}_0}^2_{H^{k_0}(\mathbb T^{3})} 
\lesssim \norm{\Wv_0}^2_{H^{k_0}(\mathbb{B}_{\rho_0})
}\label{Wv-energy}
\end{align}
and
\begin{align*}
 \norm{\Pv \Wv(t)}_{H^{k_0-1}(\mathbb{B}_{\vartheta\rho_0})}+&\norm{\Pv^\perp \Wv(t) - \Pv^\perp \Wv(0)}_{H^{k_0-1}(\mathbb{B}_{\vartheta\rho_0})} \\
 &\leq  
 \norm{\Pv \hat{\Wv}(t)}_{H^{k_0-1}(\mathbb T^{3})}+\norm{\Pv^\perp \hat{\Wv}(t) - \Pv^\perp \hat{\Wv}(0)}_{H^{k_0-1}(\mathbb T^{3})}
 \lesssim t^\zeta
\end{align*}
hold for all $t\in (0,t_0]$ where we have set
$\Pv^\perp\Wv(0) = \Pv^\perp\hat{\Wv}(0)|_{\mathbb{B}_{\vartheta\rho_0}}$ and the implicit constants in above estimates are independent of the choice of $t_0\in (0,T_0]$ and $\Wv_0$ satisfying \eqref{Wv0-smallness}.

\subsubsection*{Proof of statement (d)}
Suppose $t_1\in (0,t_0)$, and $\tilde{\Wv}\in C^1(\grave{\Omega}_{\qv_1})$ is a solution of
\eqref{Fuch-past-B.1} on  $\Omega_{\qv_1}$, $\qv_1=(t_1,t_0,\rho_0,\rho_1,\ep_2)$, satisfying the initial condition $\tilde{\Wv}|_{\{t_0\}\times\mathbb{B}_{\rho_0}}=\Wv_0$.
Then the difference $\Zv=\Wv-\tilde{\Wv}$ satisfies the IVP 
\begin{align}
\Bv^0\del{t}\Zv + \frac{1}{t^{\ep_2}}\Bv^\Lambda\del{\Lambda}\Zv &= \Hv(t,\Wv,\tilde{\Wv},\del{}\tilde{\Wv})\Zv\hspace{0.5cm} \text{in $\Omega_{\qv_1}$,} \label{Fuch-past-unique.1}\\
\Zv &= 0 \hspace{3.40cm}  \text{in $\{t_0\}\times \mathbb{B}_{\rho_0}$,}
\label{Fuch-past-unique.2}
\end{align}
where $\Hv$ is a matrix valued map that is smooth in its arguments.

Next, we set
\begin{equation*}
\Bh_n = n_0 \Bh^0 +\frac{1}{t^{\ep_2}} n_\Lambda e^\Lambda_C\Bh^C \overset{\eqref{n-def}}{=} \frac{1}{t^{\ep_2}} \biggl(\rho_1 \Bh^0 + \frac{1}{|x|}x_\Lambda e^\Lambda_C \Bh^ C\biggr) 
\end{equation*}
where $n = n_\nu dx^\nu$ is the outward pointing normal to the ``side'' boundary component $\Gamma_{\qv_1}$ of $\Omega_{\qv_1}$ (see \eqref{Omega-bndry}), $\Bh^0$ and $\Bh^C$ are as defined previously (see \eqref{Bh0-def}-\eqref{BhC-def}), and $e=(e^\Sigma_P)$
is the variable that appears in the first vector component
\begin{equation*}
(W_{\bc})_{|\bc|= 0} = W \overset{\eqref{W-def}}{=}  (e_P^\Sigma, \alpha, A_P, U_P, \Hc, \Sigma_{PQ}, N_{PQ})^{\tr}
\end{equation*}
from $\Wv$ (see \eqref{Wv-def}). Then the same calculation, cf. \eqref{Bn-ubnd-B} and \eqref{Bn-ubnd-A}, from the proof of Proposition \ref{prop:locA} that led to the inequality \eqref{spacelike} shows that \begin{equation} \label{Bhn-past-Fuch}
\Bh_n|_{\Gamma_{\qv_1}} \leq 0
\end{equation}
provided $e^\Lambda_C$ satisfies
\begin{equation}\label{e-past-Fuch-bnd}
|e| \leq \frac{(1-\ep_2)(1-\vartheta)}{10T_0^{1-\ep_2}},
\end{equation}
where we note that
$\frac{(1-\ep_2)(1-\vartheta)}{10T_0^{1-\ep_2}} < \frac{(1-\ep_2)(1-\vartheta)}{9t_0^{1-\ep_2}}=\frac{\rho_1}{9}$
holds since $\vartheta,\ep_2\in (0,1)$ and $t_0\in (0,T_0]$ by assumption. Next, setting
\begin{equation*}
\Bv_n =  n_0 \Bv^0 +\frac{1}{t^{\ep_2}} n_\Lambda \Bv^\Lambda,
\end{equation*}
it is clear from \eqref{Bv0-def}- \eqref{Bv-Lambda-def} and \eqref{Bhn-past-Fuch} that
\begin{equation} \label{Bvn-past-Fuch}
\Bv_n|_{\Gamma_{\qv_1}} \leq 0
\end{equation}
provided \eqref{e-past-Fuch-bnd} is satisfied. 

Since $e=(e^\Lambda_C)$ is one of the components of $\Wv$, there exists, by Sobolev's inequality (Theorem \ref{Sobolev}) a constant $C_{\text{Sob}}>0$ such that $\sup_{(t,x)\in\Omega_{\qv}}{|e(t,x)|}\leq C_{\text{Sob}} \sup_{0<t<t_0} \norm{\Wv(t)}_{H^{k_0}(\mathbb{B}_{\rho(t)})}$. This inequality together with \eqref{Wv-energy} implies that
$\sup_{(t,x)\in\Omega_{\qv}}{|e(t,x)|}\leq C_1 \norm{\Wv_0}_{H^{k_0}(\mathbb{B}_{\rho_0})}$
for some constant $C_1>0$ that is independent of $t\in (0,T_0]$ and $\Wh_0$ satisfying \eqref{Wv0-smallness}.
Thus assuming that the initial data satisfies
\begin{equation*}
\norm{\Wv_0}_{H^{k_0}(\mathbb{B}_{\rho_0})} < \delta_0:=\min\left\{\frac{(1-\ep_2)(1-\vartheta)}{10C_1 T_0^{1-\ep_2}},\frac{\hat{\delta}_0}{C_0}\right\},
\end{equation*}
then \eqref{e-past-Fuch-bnd} as well as \eqref{Bvn-past-Fuch} hold.

Using the evolution equation \eqref{Fuch-past-A.1} and the symmetry of the matrices $\Bv^0$ and $\Bv^\Lambda$, we find after a short calculation that 
\begin{equation*}
\del{\mu}\Bigl(\delta^\mu_0 \Zv^T \Bv^0\Zv
+ \delta^\mu_\Lambda\frac{1}{t^{\ep_2}} \Zv^T \Bv^\Lambda \Zv\Bigr) = 2 \Zv^T \Hv + \frac{1}{t^{\ep_2}}\Zv^T\del{\Lambda}\Bv^\Lambda\Zv. \end{equation*}
Fixing $t_*\in (t_1,t_0]$ and integrating the above expression over 
over $\Omega_{\qv(t)}$, where $\qv(t)=(t,t_0,\rho_0,\rho_1,\ep_2)$ and $t\in [t_*,t_0]$ , we find with the help of Stokes' theorem and the initial condition \eqref{Fuch-past-unique.2} that
\begin{equation*}
-\int_{\mathbb{B}_{\rho(t)}} \Zv^T \Bv^0 \Zv \, d^3 x + \int_{\Gamma_{\qv(t)}}\Zv^T \Bv_n \Zv \, \sigma_{\Gamma_{q(t)}} = \int^{t_0}_t \int_{\mathbb{B}_{\rho(s)}} \Zv^T\bigl(\Hv +  \del{\Lambda}\Bv^\Lambda\bigr) \Zv \, d^3 x\, ds .
\end{equation*}
Multiplying the above expression by $-1$, we deduce, from \eqref{Bvn-past-Fuch}, the invertibility of the symmetric matrix $\Bv^0$, and the boundedness of the matrix $\Hv + \del{\Lambda}\Bv^\Lambda$, the existence of a constant $C_2>0$ such that
\begin{equation*}
\int_{\mathbb{B}_{\rho(t)}} \Zv^T \Bv^0 \Zv \, d^3 x \leq C_2 \int^{t_0}_t \int_{\mathbb{B}_{\rho(s)}} \Zv^T \Bv^0 \Zv \, d^3 x\, ds 
\end{equation*}
holds for all $t\in [t_*,t_0]$. An application of Gr\"{o}nwall's inequality implies $\Zv|_{\Omega_{\qv(t)}}=0$ for $t\in [t_*,t_0]$, and  since $t_*\in (t_1,t_0]$ was chosen arbitrarily and $\Zv=\Wv-\tilde{\Wv}$, we conclude that
$\Wv=\tilde{\Wv}$ on $\Omega_{\qv_1}$. This shows that the solution $\Wv$ is unique and completes the proof. 
\end{proof}

\section{Stable big bang formation\label{sec:stable-big-bang}}
By combining the local and global existence theories developed in Sections \ref{sec:local-existence-theory} and \ref{sec:Fuch-global-existence-theory}, we can now prove the past stability of the subcritical Kasner-scalar field spacetimes. Specifically, we establish two distinct stability results. The first, stated in Theorem \ref{glob-stab-thm}, demonstrates global-in-space past stability, while the second, stated in Theorem \ref{loc-stab-thm}, demonstrates local-in-space past stability.

\subsection{Physical curvatures in terms of the renormalised conformal variables}
Before proving stability, we first derive formulas that express the second fundamental forms of the $t = \text{constant}$ hypersurfaces associated with the physical and conformal metrics in terms of the renormalised variables, as well as the curvature invariants $\Rb$, $\Rb_{ab}\Rb^{ab}$, and $\Cb_{abcd}\Cb^{abcd}$ of the physical metric in terms of the renormalised variables. These formulas are essential for the proofs of Theorems \ref{glob-stab-thm} and \ref{loc-stab-thm}.

To start, we observe by \eqref{ESF.3}, \eqref{conf-metricA} and \eqref{f-metric} that the frame components of the physical Ricci curvature tensor can be expressed as
\begin{equation} \label{phy-Ricci}
    \Rb_{ab} = \Tb_{ab} - \frac{1}{2}\gb^{cd}\Tb_{cd} \gb_{ab} = \Tb_{ab} - \frac{1}{2}\eta^{cd}\Tb_{cd} \eta_{ab} = \Tb_{ab} - \frac{1}{2}\Tb_c{}^c \eta_{ab},
\end{equation}
where 
\begin{equation*}
\Tb_c{}^c=\eta^{cd}\Tb_{cd}=-\Tb_{00}+\Tb_A{}^A.
\end{equation*}
Using \eqref{phy-Ricci} along with \eqref{conf-metricA}, \eqref{Phi-fix}, \eqref{t-fix} and \eqref{Tb-Upsilon-cmpts}, we find following a straightforward calculation that the curvature invariants $\Rb = \gb^{ab}\Rb_{ab}$ and $\Rb_{ab}\Rb^{ab} = \gb^{ac}\gb^{bd}\Rb_{ab}\Rb_{cd}$
are given by
\begin{align} \label{phy-scalar}
    \Rb = \gb^{ab}\Rb_{ab} 
&= -t^{-1}\Tb_c{}^c = -\frac{3}{2 \alphat^2 t^{3}}
\end{align}
and 
\begin{align} \label{phy-RabRab}
    \Rb_{ab}\Rb^{ab} 
        &= t^{-2}(\Tb_{00}\Tb_{00} - 2\delta^{AB}\Tb_{0A}\Tb_{0B} + \delta^{AC}\delta^{BD}\Tb_{AB}\Tb_{CD}) = \frac{9}{4\alphat^4 t^{6}},
\end{align}
respectively.

Letting $\thetat^a$ denote the dual basis to $\et_a$, ie. $\thetat^a(\et_b)=\delta_b^a$, then
\begin{equation*}
\thetat^0 = \alphat dt \AND \theta^A = \thetat^A_\Lambda dx^\Lambda, \quad (\thetat^A_\Lambda)=(\et^\Lambda_A)^{-1},
\end{equation*}
and the physical metric, by \eqref{conf-metricA}, \eqref{f-metric}, \eqref{Phi-fix} and \eqref{t-fix}, can be written as
\begin{equation*} \label{phys-metric}
    \gb = -t\thetat^0\otimes\thetat^0 + \gttb_{AB} \thetat^A\otimes\thetat^B,
\end{equation*}
where 
\begin{equation} \label{phys-spatial-metric}
\gttb_{AB}=t\delta_{AB}. 
\end{equation}
From this, we observe that the lapse and normal vector determined by the $t=\text{const}$ foliation and the physical metric are $t^{\frac{1}{2}}\alphat$ and $(t^{\frac{1}{2}}\alphat)^{-1}\del{t}$, respectively. Furthermore, using \eqref{comm-decomp-0}, \eqref{Hc-def} and \eqref{Sigma-def}, we see that the frame components of the second fundamental form of the $t=\text{const}$ foliation are determined by
\begin{align}
    \Kttb_{AB} &= \frac{1}{2t^{\frac{1}{2}}\alphat}(\Lsc_{\del{t}}\gttb)(\et_A,\et_B) = \frac{1}{2t^{\frac{1}{2}}\alphat}(\Lsc_{(\alphat \et_0)}\gttb)(\et_A,\et_B) \notag \\
        &= \frac{1}{2t^{\frac{1}{2}}\alphat}\Bigl(\Lsc_{(\alphat\et_0)}\bigl(\gttb(\et_A,\et_B)\bigr) - \gttb\bigl([\alphat\et_0,\et_A],\et_B\bigr) - \gttb\bigl(\et_A,[\alphat\et_0,\et_B]\bigr)\Bigr) \notag \\
        &= \frac{1}{2t^{\frac{1}{2}}\alphat}(\delta_{AB} + r_{AB} + 2\Hc\delta_{AB} + 2\Sigma_{AB}). \label{phys-2nd-FF}
\end{align}

Writing the conformal metric as
\begin{equation*}
    \gt = -\thetat^0\otimes\thetat^0 + \gttt_{AB} \thetat^A\otimes\thetat^B,
\end{equation*}
where $\gttt_{AB}=\delta_{AB}$, we note that the lapse and normal vector determined by the $t=\text{const}$ foliation and the conformal metric  are $\alpha$ and $\et_0$, respectively. Then, using a similar calculation as above, we find the following formula for the second fundamental form of the $t=\text{const}$ foliation:
\begin{align}
    \Kttt_{AB} &= \frac{1}{2\alphat}(\Lsc_{(\alphat \et_0)}\gttt)(\et_A,\et_B) \notag \\
        &= \frac{1}{2\alphat} \Bigl(\Lsc_{(\alphat\et_0)}\bigl(\gttt(\et_A,\et_B)\bigr) - \gttt\bigl([\alphat\et_0,\et_A],\et_B\bigr) - \gttt\bigl(\et_A,[\alphat\et_0,\et_B]\bigr)\Bigr) \notag \\
        &= \frac{1}{2t\alphat}\bigl(r_{AB}+2\Hc\delta_{AB}+2\Sigma_{AB}\Bigr). \label{conf-2nd-FF}
\end{align}

Next, letting $\Ct_{abcd}$ denote the Weyl curvature of the conformal metric $\gt$, we observe from \cite[Eq.~(1.34) and (1.35)]{Wainwright-Ellis:1997} that the Weyl curvature invariant $\Ct_{abcd}\Ct^{abcd}$ can be expressed as
\begin{equation}
    \Ct_{abcd}\Ct^{abcd} = 8(\Et_{AB}\Et^{AB}-\Ht_{AB}\Ht^{AB}) 
    \label{Weyl-invariant}
\end{equation}
where $\Et_{AB}$ and $\Ht_{AB}$ are, respectively, the electric and magnetic parts of the Weyl curvature.
According to \cite[Eq.~(1.80) and (1.81)]{Wainwright-Ellis:1997}, the electric and magnetic parts are given by 
\begin{align*}
    \Et_{AB} &= \Hct\Sigmat_{AB} - \Sigmat_{AC}\Sigmat_B{}^C + \frac{1}{3}\delta_{AB}\Sigmat_{CD}\Sigmat^{CD} + {}^{*}S_{AB}, 
    \\
    \Ht_{AB} &= \frac{1}{2}\Nt_C{}^C\Sigmat_{AB} - 3\Nt_{(A}^C\Sigmat_{B)C} + \Nt_{CD}\Sigmat^{CD}\delta_{AB} + (\et_C-\At_C)\big(\Sigmat_{D(A}\ep_{B)}{}^{CD}\big), 
\end{align*}
where
\begin{equation*}
    {}^{*}S_{AB} = \et_{(A}(\At_{B)}) + \Bt_{AB} - \frac{1}{3}\delta_{AB}\bigl(\et_C(\At^C)+\Bt_C{}^C\bigr) - (\et_C-2\At_C)\big(\Nt_{D(A}\ep_{B)}{}^{CD}\big) \label{*S-def}
\end{equation*}
with
\begin{equation*}
    \Bt_{AB} = 2\Nt_{AC}\Nt_B{}^C - \Nt_C{}^C\Nt_{AB}.
\end{equation*}
A short calculation then shows that $\Et_{AB}$ and $\Ht_{AB}$ can be expressed in terms of the renormalised variables \eqref{Hc-def}-\eqref{N-def} as follows: 
\begin{align}
    t^2\alphat^2 \Et_{AB} &= \Hc\Sigma_{AB} + \Bigl(\frac{1}{2}r_{AB}-\frac{r_0}{6}\delta_{AB}\Bigr)\Hc + \frac{r_0}{3}\Sigma_{AB} + \frac{r_0}{4}r_{AB} \notag \\
        &\quad - \Sigma_{AC}\Sigma_B{}^C - \frac{1}{2}r_B^C\Sigma_{AC} - \frac{1}{2}r_{AC}\Sigma_B{}^C - \frac{1}{4}r_{AC}r_B^C + \frac{r_0}{6}\Sigma_{AB} \notag \\
        &\quad + \frac{1}{3}\delta_{AB}\Sigma_{CD}\Sigma^{CD} + \frac{1}{3}\delta_{AB}r_{CD}\Sigma^{CD} + \frac{r_0}{3}\delta_{AB} \notag \\
        &\quad + t^{1-\ep_2}e_{(A}(A_{B)}) - U_{(A}A_{B)} + 2N_{AC}N_B{}^C - N_C{}^CN_{AB} - \frac{1}{3}t^{1-\ep_2}\delta_{AB} + \frac{1}{3}\delta_{AB}U_C A^C \notag \\
        &\quad - \frac{2}{3}\delta_{AB}N_{AC}N^{AC} + \frac{1}{3}\delta_{AB}(N_C{}^C)^2 - t^{1-\ep_2}e_C\bigl(N_{D(A}\ep_{B)}{}^{CD}\bigr) + (U_C+2A_C)\bigl(N_{D(A}\ep_{B)}{}^{CD}\bigr), \label{Et-form} \\
    t^2\alphat^2 \Ht_{AB} &= \frac{1}{2}N_C{}^C\Sigma_{AB} + \frac{1}{4}r_{AB}N_C{}^C - \frac{1}{4}\delta_{AB}N_C{}^C - 3N_{(A}{}^C\Sigma_{B)C} - \frac{3}{2}N_{(A}{}^C r_{B)C} \notag \\
        &\quad + \frac{r_0}{2}N_{(A}{}^C\delta_{B)C} + \delta_{AB}N_{CD}\Sigma^{CD} + \frac{1}{2}\delta_{AB}r_{CD}N^{CD} \notag \\
        &\quad + t^{1-\ep_2}e_C(\Sigma_{D(A})\ep_{B)}{}^{CD} - (A_C+U_C)\Bigl(\Sigma_{D(A}+\frac{1}{2}r_{D(A}-\frac{r_0}{6}\delta_{D(A}\Bigr)\ep_{B)}{}^{CD}. \label{Ht-form}
\end{align}
From \cite[Thm.~7.30]{JohnMLee-Rie:2018}, we know that  under the conformal transformation \eqref{conf-metricA} the Weyl curvature $ \Cb_{abcd}$ of the physical metric $\gb_{ab}$ is related to that of the conformal metric $\gt_{ab}$ via
\begin{equation*}
    \Cb_{abcd} = e^{2\Phi}\Ct_{abcd}.
\end{equation*}
Letting $\Cb^{abcd}=\gb^{ap}\gb^{bq}\gb^{cr}\gb^{ds}\Cb_{pqrs}$, it follows that the physical Weyl curvature invariant $\Cb^{abcd}\Cb_{abcd}$ can be expressed as
\begin{equation}
\Cb^{abcd}\Cb_{abcd} = e^{-6\Phi}\Ct^{abcd}\Ct_{abcd}=\frac{8}{t^{7}\alphat^4}\Bigl(t^2\alphat^2 \Et_{AB} t^2\alphat^2 \Et^{AB}-t^2\alphat^2 \Ht_{AB}t^2\alphat^2\Ht^{AB}\Bigr), \label{phys-Weyl-invariant}
\end{equation}
where in deriving the second equality we used \eqref{Phi-fix}, \eqref{t-fix}, and \eqref{Weyl-invariant}. 

\subsection{Global in space past stability of the subcritical Kasner-scalar field spacetimes}

\begin{thm}\label{glob-stab-thm}
Suppose $T_0>0$, $k_0\in \Zbb_{>\frac{3}{2}+1}$, $\delta_1\in \bigl(0,\frac{1}{2}+\frac{r_0}{6}\bigr)$, the Kasner exponents $q_I\in \Rbb$, $I=1,2,3$, and $P\in (0,\sqrt{1/3}]$ satisfy the Kasner relations \eqref{Kasner-rels-A} and the subcritical condition \eqref{SCR},  $\ep_1\in (0,1)$ and $\ep_2\in \Rbb$ satisfy \eqref{eigenval-posA}  and \eqref{eigenval-posB}, respectively, where $r_0$ and $r_I$, $I=1,2,3$, are defined by \eqref{Kasner-rels-B}, $\nu>0$ satisfies \eqref{eigenval-posC}, $k_1\in \Nbb_0$ is chosen large enough so that \eqref{k-nu-fix} holds\footnote{If the right hand side of \eqref{k-nu-fix} is $0$, then $k_1=0$ is allowed.}. Then there exists a $\delta_0 >0$ such that for every $t_0\in (0,T_0]$, $\delta \in (0,\delta_0]$, and\footnote{$\Wbb$ is the vector space defined above by \eqref{Wbb-def}.} 
\begin{equation*}
W_0 = ( e_{0P}^\Sigma, \alpha_0, A_{0P}, U_{0P}, \Hc_0, \Sigma_{0PQ}, N_{0PQ})^{\tr} \in H^{k}(\Tbb^3,\Wbb), \quad k=k_0+k_1,
\end{equation*}
satisfying $\inf_{x\in\Tbb^3}\alpha_0(x)>0$, $\inf_{x\in\Tbb^3}\det\bigl(e_{0P}^\Sigma(x)\bigr)>0$, the constraint equations \eqref{A-cnstr-3}-\eqref{H-cnstr-3} at  $t=t_0$, and
$\norm{W_0}_{H^{k}(\Tbb^3)} < \delta$,
there exists a 
\begin{equation*} 
W= \bigl( e_{P}^\Sigma, \alpha, A_{P}, U_{P}, \Hc, \Sigma_{PQ}, N_{PQ}\bigr)^{\tr} \in \bigcap_{j=0}^{k} C^j\bigl((0,t_0],H^{k-j}(\Tbb^3,\Wbb)\bigr) \subset C^1(M_{0,t_0},\Wbb)
\end{equation*}
that solves, for all choices of the parameters $\gamma,\rho \in\Rbb$, the evolution equations \eqref{EEc.1}-\eqref{EEc.7} on $M_{0,t_0}=(0,t_0]\times \Tbb^3$, and satisfies the initial condition $W|_{t=t_0}=W_0$, the constraint equations \eqref{A-cnstr-3}-\eqref{H-cnstr-3} on $M_{0,t_0}$, and  $\alpha>0$ and $\det(e_{P}^\Sigma)>0$ on $M_{0,t_0}$. Moreover the following hold:
\begin{enumerate}[(a)]
\item The solution $W$ is uniformly bounded by
\begin{equation*}
\norm{W(t)}_{H^{k}(\Tbb^3)} \lesssim \delta
\end{equation*}
for all $t\in (0,t_0]$, 
and there exist a constant $\zeta>0$ and functions 
$\alpha_*, \Hc_*, \Sigma_{AB}^*\in H^{k-1}(\Tbb^3)\subset C^{0}(\Tbb^3)$
satisfying
\begin{gather*}
0<\inf_{x\in\Tbb^3}\alpha_*(x) \leq \sup_{x\in\Tbb^3}\alpha_*(x)\lesssim 1, \quad \norm{\Hc_*}_{L^\infty(\Tbb^3)}\leq \delta_1,\\
\Sigma^*_{AB}=\Sigma^*_{BA},  \quad \delta^{AB}\Sigma^*_{AB}=0 \AND \max\big\{\norm{\Hc_*}_{H^{k-1}(\Tbb^3)}, \norm{\Sigma_{AB}^*}_{H^{k-1}(\Tbb^3)} \bigr\}\lesssim \delta
\end{gather*}
such that the components of $W$ decay according to
\begin{gather*}
e_P^\Omega = \Ord_{H^{k-1}(\Tbb^3)}(t^\zeta),\quad 
\alpha = t^{\ep_1+\frac{r_0}{2}+3\Hc_*}\alpha_*\Bigl(1+\Ord_{H^{k-1}(\Tbb^3)}(t^\zeta)\Bigr), \\
A_P = \Ord_{H^{k-1}(\Tbb^3)}(t^\zeta), \quad
U_P = \Ord_{H^{k-1}(\Tbb^3)}(t^\zeta), \quad
N_{PQ} = \Ord_{H^{k-1}(\Tbb^3)}(t^\zeta), \\
\Hc = \Hc_* + \Ord_{H^{k-1}(\Tbb^3)}(t^\zeta) \AND
\Sigma_{PQ} = \Sigma_{PQ}^* + \Ord_{H^{k-1}(\Tbb^3)}(t^\zeta)
\end{gather*}
for $t\in (0,t_0]$. Furthermore, the explicit and implicit constants in these estimates are independent to the choice of $\delta\in (0,\delta_0]$ and $t\in (0,t_0]$.
\item The pair
\begin{equation*}
\biggl\{\gb = t\bigl(-\alphat^2 dt\otimes dt + \gt_{\Sigma\Omega}dx^\Sigma \otimes dx^\Omega\bigr),\; \phi = \frac{\sqrt{3}}{2}\ln(t)\Biggr\},
\end{equation*}
where  
\begin{equation*}
\alphat = t^{-\ep_1}\alpha,\quad 
(\gt_{\Sigma\Omega}) =(\delta^{AB}\et_A^{\Sigma}\et_B^{\Omega})^{-1} \AND \et_A^\Omega = t^{\ep_1-\ep_2}\alpha^{-1}e_A^\Omega,
\end{equation*}
determines a classical solution of the  
Einstein-scalar-field equations \eqref{ESF.1}-\eqref{ESF.2} on $M_{0,t_0}$. Furthermore, the trace of the physical second fundamental form $\Kttb_{AB}$ determined by the physical metric $\gb$ and the $t=\textit{const}$ hypersurfaces satisfies
\begin{equation*}
\Kttb_A{}^A=\gttb^{AB}\Kttb_{AB} = \frac{1}{t^{\frac{3}{2}+\frac{r_0}{2}+3\Hc_*}} \frac{1}{\alpha_*}\biggl(\frac{3}{2}+\frac{r_0}{2}+3\Hc_* + \Ord_{H^{k-1}(\Tbb^3)}(t^\zeta)\biggr)
\end{equation*}
for $t\in (0,t_0]$, and is bounded below by
\begin{equation*}
\inf_{x\in\Tbb^3}\Kttb_A{}^A(t,x) \gtrsim \frac{1}{t^{\frac{3}{2}+\frac{r_0}{2}-3\delta_1}}, \quad 0<t\leq t_0.
\end{equation*}
Since $\frac{3}{2}+\frac{r_0}{2}-3\delta_1>0$, $\Kttb_A{}^A$ blows up uniformly as $t\searrow 0$ and the hypersurface $t=0$ is a \textit{crushing singularity}, cf.~\cite{Eardley:1979}. 
\item The pair 
\begin{equation*}
\Bigl\{\gt = -\alphat^2 dt\otimes dt + \gt_{\Sigma\Omega}dx^\Sigma \otimes dx^\Omega,\; \tau = t\Bigr\}
\end{equation*}
determines a classical solution of the  
conformal Einstein-scalar-field equations \eqref{cESF.1}-\eqref{cESF.2} on $M_{0,t_0}$. Furthermore, the second fundamental form $\Kttt_{AB}$ determined by the conformal metric $\gt$ and the $t=\textit{const}$ hypersurfaces, when expressed relative to the orthonormal frame $\{\et_0= \alphat^{-1}\del{t},\et_A=\et_A^\Lambda\del{\Lambda}\}$,
decays according to
\begin{equation*}
2t\alphat \Kttt_{AB}  = \Kf_{AB}+\Ord_{H^{k-1}(\Tbb^3)}(t^\zeta)
\end{equation*}
where 
\begin{equation*} 
\Kf_{AB}= r_{AB}+2\Hc_*\delta_{AB}+2\Sigma^*_{AB}\in H^{k-1}(\Tbb^3,\Sbb{3})\subset C^{0}(\Tbb^3,\Sbb{3})
\end{equation*}
satisfies 
\begin{equation*} 
 (\Kf_A{}^A)^2 - \Kf_A{}^B\Kf_B{}^A + 4\Kf_A{}^A = 0
\AND
\Kf_A{}^A=\sqrt{4+\Kf_{A}{}^B \Kf_{B}{}^A}-2\geq 0.
\end{equation*}
In particular, this implies that the solution $\{\gt,\tau\}$ is \textit{asymptotically pointwise Kasner} on $\Tbb^3$ relative to the frame $\{\et_0= \alphat^{-1}\del{t},\et_A=\et_A^\Lambda\del{\Lambda}\}$, cf.~\cite[Def.~1.1]{BeyerOliynyk:2024b}.
\item The physical solution $\{\gb,\phi\}$ of the Einstein-scalar field equations on $M_{0,t_1}$ is past $C^2$ inextendible at $t=0$ and past timelike geodesically incomplete. The curvature invariants   $\Rb=\gb^{ab}\Rb_{ab}$ and $\Rb_{ab}\Rb^{ab}$ of the physical metric $\gb$ 
satisfy
\begin{equation*}
\Rb = \frac{3}{2\alpha_*^2 t^{2(\frac{3}{2}+\frac{r_0}{2}+3\Hc_*)}}\bigl(-1+\Ord_{H^{k-1}(\Tbb^3)}(t^\zeta)\bigr) 
\AND
\Rb_{ab}\Rb^{ab} = \frac{9}{4\alpha_*^4t^{4(\frac{3}{2}+\frac{r_0}{2}+3\Hc_*)}}\bigl(1+\Ord_{H^{k-1}(\Tbb^3)}(t^\zeta)\bigr)
\end{equation*}
for $t\in (0,t_0]$, respectively, and are bounded above and below by
\begin{equation*}
\sup_{x\in\Tbb^3}\Rb(t,x) \lesssim -\frac{1}{t^{2(\frac{3}{2}+\frac{r_0}{2}-3\delta_1)}} \AND
\inf_{x\in\Tbb^3}\Rb_{ab}(t,x)\Rb^{ab}(t,x) \gtrsim \frac{1}{t^{4(\frac{3}{2}+\frac{r_0}{2}-3\delta_1)}}
\end{equation*}
for $t\in (0,t_0]$, respectively. Since $\frac{3}{2}+\frac{r_0}{2}-3\delta_1>0$, $\Rb$ and $\Rb_{ab}\Rb^{ab}$ blow up uniformly as $t\searrow 0$.
\item 
The Weyl curvature invariant $\Cb^{abcd}\Cb_{abcd}$ of the physical metric $\gb$ decays according to
\begin{equation} \label{Wey-invariant-exp}
\Cb^{abcd}\Cb_{abcd} =\frac{8}{t^{7+2r_0+12\Hc_*}\alpha_*^4}\Bigl(\Ec_{AB}\Ec^{AB}+\Ord_{H^{k-1}(\Tbb^3)}(t^\zeta)\Bigr)
\end{equation}
for $t\in (0,t_0]$ where
\begin{equation}
    \label{eq:ecABdef}
\ec_{AB} = \frac{r_0}{3}\delta_{AB} +  \frac{r_0}{4}r_{AB}- \frac{1}{4}r_{AC}r_B^C
\end{equation}
and 
\begin{align}
\Ec_{AB} &=\ec_{AB}+ \Hc_*\Sigma^*_{AB} + \Bigl(\frac{1}{2}r_{AB}-\frac{r_0}{6}\delta_{AB}\Bigr)\Hc_* + \frac{r_0}{2}\Sigma^*_{AB}  - \Sigma^*_{AC}\Sigma^*_B{}^C - \frac{1}{2}r_B^C\Sigma^*_{AC} \label{eq:EcABdef}\\
&\quad - \frac{1}{2}r_{AC}\Sigma^*_B{}^C + \frac{1}{3}\delta_{AB}\Sigma^*_{CD}\Sigma^{*CD} + \frac{1}{3}\delta_{AB}r^{CD}\Sigma^*_{CD}\in H^{k-1}(\Tbb^3,\Sbb{3}) \subset C^0(\Tbb^3,\Sbb{3}).\notag
\end{align}
Moreover, if the Kasner exponents $q_I$ are not all equal, or equivalently, $(r_1,r_2,r_3)\neq (0,0,0)$, then 
\begin{equation*}
    \inf_{x\in\Tbb^3}\Ec_{AB}(x) \Ec^{AB}(x) >0
\end{equation*} and there exists a $t_1\in (0,t_0]$ such that  
$\Cb^{abcd}\Cb_{abcd}$ is bounded below by
\begin{equation}
    \label{eq:LowerWeylbound}
\inf_{x\in\Tbb^3}\Cb^{abcd}(t,x)\Cb_{abcd}(t,x) \gtrsim \frac{1}{t^{1+4(\frac{3}{2}+\frac{r_0}{2}-3\delta_1)}}, \quad 0<t\leq t_1.
\end{equation}
Since  $\frac{3}{2}+\frac{r_0}{2}-3\delta_1>0$, $\Cb^{abcd}\Cb_{abcd}$ blows up uniformly as $t\searrow 0$ provided $(r_1,r_2,r_3)\neq (0,0,0)$.
\end{enumerate}
\end{thm}

\begin{rem}\label{rem:Kasner-idata-stability}
The Kasner-scalar field solutions, when evaluated at $t = t_0$, determine initial data $W_0^{\textrm{Kas}}$ that, by \eqref{Kasner-solns-C}-\eqref{Kasner-solns-limit}, satisfies  
\begin{equation}\label{Kasner:t->0}  
\lim_{t_0 \searrow 0} \|W_0^{\textrm{Kas}}\|_{H^k(\Tbb^3)} = 0,  
\end{equation}  
and where $\alpha_0 > 0$ and $\det(\varepsilon^\Omega_{0P}) > 0$ on $\Tbb^3$ for any $k \in \Nbb_0$.  
Consequently, by choosing $t_0 > 0$ sufficiently small, we can always ensure that the Kasner-scalar field initial data satisfy the conditions of Theorem \ref{glob-stab-thm}. This guarantees the existence of an open neighbourhood of initial data around the Kasner-scalar field initial data that also satisfy the conditions of Theorem \ref{glob-stab-thm}, thereby establishing the past nonlinear stability of the Kasner-scalar field solutions and their big bang singularities. 
\end{rem}

\begin{rem}
    \label{rem:AVTD}
The solutions determined by this theorem also exhibit \emph{asymptotically velocity term dominated (AVTD) behaviour} \cite{Eardley:1972,isenberg1990}, meaning that near \( t = 0 \), spatial derivative terms in the evolution equations decay more rapidly than other terms. As a result, these solutions can be viewed as asymptotic solutions of the \emph{velocity term dominated} equations, which are obtained by neglecting the spatial derivatives in the full evolution equations. Given that the evolution equations are in Fuchsian form, see \eqref{Fuch-past-A.1}, this AVTD behaviour follows from the coefficient properties established in the proof of Proposition~\ref{prop:Fuch-global} and Section~\ref{coeff-props}, particularly from the restriction \eqref{eigenval-posC}.
\end{rem}

\begin{rem}
    \label{rem:Weyl}
It is expected that the Weyl curvature generically blows up at big bang singularities except in the special case of exact FLRW solutions where the Weyl tensor is zero and only the Ricci part of the curvature tensor diverges. Theorem~\ref{glob-stab-thm}.(e) confirms this expectation for 
all sufficiently small perturbations subcritical Kasner-scalar field solutions where $(r_1,r_2,r_3) \neq (0,0,0)$. However, our results do not fully determine the asymptotic behaviour of the Weyl curvature for perturbations of the isotropic Kasner solution, ie. FLRW, where $(r_1,r_2,r_3) = (0,0,0)$. This is because $\ec_{AB}$ vanishes in this case as a consequence of \eqref{eq:ecABdef}, and therefore, we cannot guarantee that $\Ec_{AB}$, defined by \eqref{eq:EcABdef}, is everywhere non-vanishing. Consequently, the expansion for the Weyl curvature in \eqref{Wey-invariant-exp} does not provide a divergent lower bound in case. It is possible that this issue could be resolved by deriving additional terms in the expansion of the Weyl tensor beyond those given in \eqref{Wey-invariant-exp}, but we leave this for future work.

We also emphasise that the blow-up of the Weyl curvature established in Theorem~\ref{glob-stab-thm} is \emph{purely electric}, in the sense that the magnetic part of the Weyl tensor is of higher order near $t=0$ than the electric part; see \eqref{Et-Ht-decay}. Future work could explore interpreting the perturbed solutions of Theorem~\ref{glob-stab-thm} within the framework of \emph{silent universes}; see, for example, \cite{bruni1995} (despite the absence of irrotational dust matter in our treatment here).

\end{rem}

\subsection{Proof of Theorem \ref{glob-stab-thm}}

\subsubsection*{Choice of constants}
Fix $T_0\in (0,1]$, $k_0\in \Zbb_{>\frac{3}{2}+1}$, and fix Kasner exponents $q_I\in \Rbb$, $I=1,2,3$, and a constant $P\in (0,\sqrt{1/3}]$ that together satisfy the Kasner relations \eqref{Kasner-rels-A}. Further, let $r_0$ and $r_I$, $I=1,2,3$, be as defined by \eqref{Kasner-rels-B}, and assume that the Kasner exponents satisfy the subcritical condition \eqref{SCR}. Finally, fix  $\ep_1\in (0,1)$ and $\ep_2\in \Rbb$ satisfying \eqref{eigenval-posA}  and \eqref{eigenval-posB}, respectively, fix $\nu>0$ satisfying \eqref{eigenval-posC}, choose $k_1\in \Nbb_0$ large enough so that \eqref{k-nu-fix} holds, and set $k=k_0+k_1$. 

\subsubsection*{Initial data}
Suppose $t_0\in (0,T_0]$ and let $\Wbb$ be the vector space defined above by \eqref{Wbb-def}. Assume that the initial data 
\begin{equation*}
W_0 = ( e_{0P}^\Sigma, \alpha_0, A_{0P}, U_{0P}, \Hc_0, \Sigma_{0PQ}, N_{0PQ})^{\tr} \in H^{k}(\Tbb^3,\Wbb),
\end{equation*}
for the Einstein-scalar field equations satisfies $\alpha_0(x)>0$ and $\det\bigl(e_{0P}^\Sigma(x)\bigr)>0$ for all $x\in \Tbb^3$, and the constraint equations \eqref{A-cnstr-3}-\eqref{H-cnstr-3} at  $t=t_0$. Further, assume that
\begin{equation} \label{stability-idata-A}
\norm{W_0}_{H^{k}(\Tbb^3)} < \delta. 
\end{equation}
For now, $\delta>0$ is arbitrary, but later, it will be required to be chosen sufficiently small.  Importantly, this smallness condition on $\delta$ is \textit{independent} of the choice of $t_0\in (0,T_0]$.

Spatially differentiating $W_0$ yields 
\begin{equation}\label{stability-Wv0-def}
\Wv_0 = \bigl( (W_{0\bc})_{|\bc|= 0},(W_{0\bc})_{|\bc|=1},(W_{0\bc})_{|\bc|=2}, \ldots,(W_{0\bc})_{|\bc|=k_1-1},(\Wch_{0\bc})_{|\bc|=k_1}\bigr)^{\tr},
\end{equation}
where $W_{0\bc} = t_0^{|\bc|\nu}\del{}^\bc W_0$, $1\leq|\bc|<k_1$, and $\Wch_{0\bc}=V^{-1} t_0^{|\bc|\nu}\del{}^\bc W_0$, $|\bc|=k_1$, with $V$ defined by \eqref{V-fix}. In the following, we will use this as initial data for the Fuchsian equation \eqref{Fuch-past-A.1}. Noting that $0<t_0\leq T_0\leq 1$, it is clear from \eqref{stability-idata-A} that
\begin{equation} \label{stability-idata-B}
\norm{\Wv_0}_{H^{k_0}(\Tbb^3)} < C_0\delta
\end{equation}
for some constant $C_0>0$ that is independent of the choice of $\delta>0$ and $t_0\in (0,T_0]$.

\subsubsection*{Global existence to the past} 
By Proposition \ref{prop:Fuch-global} and \eqref{stability-idata-B}, there exists a constant $\delta_0>0$, independent of $t_0\in (0,T_0]$, such that if $\delta \in (0,\delta_0]$, then there exists a solution 
\begin{equation} \label{stability-Fuch-sol}
\Wv \in C^0_b\bigl((0,t_0],H^{k_0}(\mathbb T^{3})\bigr)\cap C^1\bigl((0,t_0],H^{k_0-1}(\mathbb T^{3})\bigr)\subset C^1(M_{0,t_0})
\end{equation}
of the Fuchsian GIVP \eqref{Fuch-past-A.1}-\eqref{Fuch-past-A.2} such that the limit $\lim_{t\searrow 0} \Pv^\perp \Wv(t)$, denoted $\Pv^\perp \Wv(0)$, exists in $H^{k_0-1}(\mathbb{T}^3)$. Moreover, there exist constants $C_1,\zeta>0$, both independent of $t_0\in (0,T_0]$ and $\delta \in (0,\delta_0]$, such that the  solution $\Wv(t)$ satisfies the energy estimate
\begin{equation}\label{stability-energy} 
\norm{\Wv(t)}_{H^{k_0}(\mathbb T^{3})}^2 + \int^{t_0}_t \frac{1}{s} \norm{\Pv \Wv(s)}_{H^{k_0}(\mathbb T^{3})}^2\, ds \leq C_1 \delta^2, \quad \forall\, t \in (0,t_0],
\end{equation}
and the decay estimates
\begin{equation}\label{stability-decay}
  \norm{\Pv \Wv(t)}_{H^{k_0-1}(\mathbb T^{3})}+\norm{\Pv^\perp \Wv(t) - \Pv^\perp \Wv(0)}_{H^{k_0-1}(\mathbb T^{3})} \leq C_1 t^\zeta, \quad \forall t\in (0,t_0].
\end{equation}
It is worth noting here that while $\delta_0$ depends on the constants, eg. $T_0$, $q_I$, $\ep_1$, $\ep_2$, $\nu$, $k_0$, $k_1$, that were fixed at the beginning of the proof, the parameter $\delta$, other than satisfying $0<\delta \leq \delta_0$, does not depend on these constants. 

Since the initial data $W_0$ satisfies the constraint equations, and $\alpha_0(x)>0$ and $\det(e_{0P}^\Sigma(x))>0$ for all $x\in \Tbb^3$, we also know from Proposition \ref{prop:locA} that there exist a $t_1\in (0,t_0]$ and a
\begin{equation} \label{stability-W-sol}
W= \bigl( e_{P}^\Sigma, \alpha, A_{P}, U_{P}, \Hc, \Sigma_{PQ}, N_{PQ}\bigr)^{\tr} \in \bigcap_{j=0}^{k} C^j\bigl((t_1,t_0],H^{k-j}(\Tbb^3,\Wbb)\bigr)
\end{equation}
that solves the evolution equations \eqref{EEc.1}-\eqref{EEc.7}, for all choices of the parameters $\gamma,\rho \in\Rbb$, on $M_{t_1,t_0}=(t_1,t_0]\times \Tbb^3$, and that satisfies the initial condition $W|_{t=t_0}=W_0$, the constraint equations \eqref{A-cnstr-3}-\eqref{H-cnstr-3} on $M_{t_1,t_0}$, and  $\alpha>0$ and $\det(e_{P}^\Sigma)>0$ on $M_{t_1,t_0}$. Moreover, from the continuation principle from Proposition \ref{prop:locA}, this solution can be continued past $t_1$ provided
$\sup_{t_1<t<t_0}\norm{W(t)}_{W^{1,\infty}(\Tbb^3)} < \infty$.

Next, we define 
\begin{equation} \label{hatWv-def}
\hat{\Wv} = \bigl( W,(W_{\bc})_{|\bc|=1},(W_{\bc})_{|\bc|=2}, \ldots,(W_{\bc})_{|\bc|=k_1-1},(W_{\bc})_{|\bc|=k_1}\bigr)^{\tr},
\end{equation}
where $W_{\bc} = t^{|\bc|\nu}\del{}^\bc W$, $1\leq |\bc|<k_1$, and $W_{\bc}=V^{-1} t^{|\bc|\nu}\del{}^\bc W$, $|\bc|=k_1$, with $V$ defined by \eqref{V-fix}. Then
\begin{equation*}
\hat{\Wv}\in \bigcap_{j=0}^{k_0} C^j\bigl((t_1,t_0],H^{k_0-j}(\Tbb^3)\bigr) \subset C^1(M_{t_1,t_0}),
\end{equation*}
and it follows from the calculation carried out in Sections \ref{sec:low-order} and \ref{sec:high-order} that $\hat{\Wv}$ satisfies the Fuchsian equation \eqref{Fuch-past-A.1} on $M_{t_1,t_0}$. Moreover, it is clear that $\hat{\Wv}|_{t=t_0}=\Wv_0$ due to \eqref{stability-Wv0-def}, \eqref{hatWv-def} and the fact that $W|_{t=t_0}=W_0$. By Proposition \ref{prop:Fuch-global}.(a), we conclude that $\hat{\Wv}=\Wv$ on $M_{t_1,t_0}$. In particular, this implies, by \eqref{stability-energy} and Sobolev's inequality (Theorem \ref{Sobolev}), that $\sup_{t_1<t<t_0}\norm{W(t)}_{W^{1,\infty}(\Tbb^3)} \leq C_2\delta$
for some constant $C_2>0$ independent of $t_0\in (0,T_0]$, $\delta\in (0,\delta_0]$ and $t_1 \in (0,t_0]$. Thus, the solution \eqref{stability-W-sol} can be continued past $t_1$ and in fact must exist the whole interval $(0,t_0]$, ie.~$t_1=0$. 

Since $\alpha>0$ and $\det(e_A^\Omega)>0$ on $M_{0,t_0}$, the solution \eqref{stability-W-sol}
determines, cf.~\eqref{conformal-metric}, \eqref{spatial-conformal-metric}, \eqref{tau-def},  \eqref{conf-metricB}, \eqref{t-fix}, and  \eqref{alpha-def}-\eqref{e-def}, a classical solution of the Einstein-scalar-field equations \eqref{ESF.1}-\eqref{ESF.2} on $M_{0,t_0}$ via  
\begin{equation*}
\biggl\{\gb = t\bigl(-\alphat^2 dt\otimes dt + \gt_{\Sigma\Omega}dx^\Sigma \otimes dx^\Omega\bigr),\; \phi = \frac{\sqrt{3}}{2}\ln(t)\Biggr\}
\end{equation*}
where  
\begin{equation} \label{alphat-gt-et}
    \alphat = t^{-\ep_1}\alpha, \quad
(\gt_{\Sigma\Omega}) =(\delta^{AB}\et_A^{\Sigma}\et_B^{\Omega})^{-1} \AND \et_A^\Omega = t^{\ep_1-\ep_2}\alpha^{-1}e_A^\Omega.
\end{equation}
By uniqueness, we also have that
\begin{equation*}
\Wv = \hat{\Wv} =  \bigl( W,(W_{\bc})_{|\bc|=1},(W_{\bc})_{|\bc|=2}, \ldots,(W_{\bc})_{|\bc|=k_1-1},(W_{\bc})_{|\bc|=k_1}\bigr)^{\tr} \quad \text{in $M_{0,t_0}$.}
\end{equation*}
This, in turn, implies by \eqref{Pbb-W}, \eqref{Pbb-perp-W}, \eqref{Pv-def}, \eqref{Pv-perp-def}, and the energy estimate \eqref{stability-energy} that
\begin{equation} \label{stability-W-bnd}
\norm{W(t)}_{H^{k}(\Tbb^3)} \leq C_1\delta.
\end{equation}
From the above estimate and the decay estimates \eqref{stability-decay}, we get
\begin{gather}
e_P^\Omega = \Ord_{H^{k-1}(\Tbb^3)}(t^\zeta),\quad 
\alpha =\Ord_{H^{k-1}(\Tbb^3)}(t^\zeta), \quad
A_P = \Ord_{H^{k-1}(\Tbb^3)}(t^\zeta), \label{stability-decay-A.1} \\
U_P = \Ord_{H^{k-1}(\Tbb^3)}(t^\zeta), \quad
N_{PQ} = \Ord_{H^{k-1}(\Tbb^3)}(t^\zeta), \label{stability-decay-A.2}  \\
\Hc = \Hc_* + \Ord_{H^{k-1}(\Tbb^3)}(t^\zeta) \AND
\Sigma_{PQ} = \Sigma_{PQ}^* + \Ord_{H^{k-1}(\Tbb^3)}(t^\zeta), \label{stability-decay-A.3} 
\end{gather}
where
$\Hc_*,\Sigma_{AB}^*\in H^{k-1}(\Tbb^3)$ satisfy
\begin{equation} \label{asymp-data-bnd-A}
\max\Bigl\{\norm{\Hc_*}_{H^{k-1}(\Tbb^3)}, \norm{\Sigma_{PQ}^*}_{H^{k-1}(\Tbb^3)}\Bigr\} \leq C_1 \delta, \quad 0<\delta \leq \delta_0,
\end{equation}
and the implicit constants in \eqref{stability-decay-A.1}-\eqref{stability-decay-A.3} are independent of the choice of $\delta\in (0,\delta_1]$ and
$t_0\in (0,T_0]$.
The bounds \eqref{stability-W-bnd} and \eqref{asymp-data-bnd-A}, in turn, imply by Sobolev's inequality (Theorem \ref{Sobolev}) that
\begin{equation*}
\max\biggl\{\norm{\Hc_*}_{L^{\infty}(\Tbb^3)}, \norm{\Sigma_{PQ}^*}_{L^{\infty}(\Tbb^3)}, \sup_{0<t\leq t_0}\norm{\Hc(t)}_{L^{\infty}(\Tbb^3)},\sup_{0<t\leq t_0}\norm{\Sigma_{PQ}(t)}_{L^{\infty}(\Tbb^3)}\biggr\} \leq C_{3} \delta, \quad 0<\delta \leq \delta_0,
\end{equation*}
for some constant $C_{3}>0$. 

For the remainder of the proof, we will further restrict $\delta$ to lie in the interval 
$0<\delta \leq \deltat_0:=\frac{\delta_1}{C_3}$
where, for now,
  $\delta_1 \in (0, C_3\delta_0]$
can be freely chosen.
This will ensure that 
\begin{equation} \label{asympt-data-bnd}
\max\biggl\{\norm{\Hc_*}_{L^{\infty}(\Tbb^3)}, \norm{\Sigma_{PQ}^*}_{L^{\infty}(\Tbb^3)}, \sup_{0<t\leq t_0}\norm{\Hc(t)}_{L^{\infty}(\Tbb^3)},\sup_{0<t\leq t_0}\norm{\Sigma_{PQ}(t)}_{L^{\infty}(\Tbb^3)}\biggr\} \leq \delta_1.
\end{equation}

\subsubsection*{Improved asymptotics for the lapse}
Since $\alpha$ satisfies \eqref{EEc.6} and $\alpha>0$, we see, after a short calculation, that 
\begin{equation*}
\del{t}\ln\bigl(t^{-\ep_1-\frac{r_0}{2}}\alpha\bigr) = \frac{3}{t}\Hc.
\end{equation*}
By introducing
\begin{equation*}
\alphabr = \frac{\alpha}{t^{\ep_1+\frac{r_0}{2}+3\Hc_*}},
\end{equation*}
we can use the decay estimate from \eqref{stability-decay-A.3} for $\Hc$ to express the above differential equation as
\begin{equation*}
\del{t}\ln(\alphabr) = \Ord_{H^{k-1}(\Tbb^3)}(t^{-1+\zeta}).
\end{equation*}
Then, by Lemma \ref{lem:asymptotic}, there exists a $\alphabr_*\in H^{k-1}(\Tbb^3)$ such that 
\begin{equation*}
\ln(\alphabr) = \alphabr_* +  \Ord_{H^{k-1}(\Tbb^3)}(t^\zeta). 
\end{equation*}
Setting $\alpha_*=\exp(\alphabr_*)$ and exponentiating the above expression, it follows from the Sobolev and Moser inequalities (Theorems \ref{Sobolev} and \ref{Moser}) that
\begin{equation} \label{alpha-asymp-data}
\alpha_* \in H^{k-1}(\Tbb^3) \AND  0<\inf_{x\in\Tbb^3}\alpha_*(x) \leq \sup_{x\in\Tbb^3}\alpha_*(x)\lesssim 1,
\end{equation}
and that 
\begin{equation} \label{alpha-decay}
\alphat \overset{\eqref{alpha-def}}{=} t^{-\ep_1}\alpha = t^{\frac{r_0}{2}+3\Hc_*}\alpha_*\att
\end{equation}
where 
\begin{gather}
\att =1+ \Ord_{H^{k-1}(\Tbb^3)}(t^\zeta), \label{att-bnds.1}\\
1 \lesssim \inf_{(t,x)\in M_{0,t_0}}\att(t,x) \leq \sup_{(t,x)\in M_{0,t_0}}\att(t,x) \lesssim 1,\label{att-bnds.2}
\end{gather}
and the implicit constants in the above expressions are independent of the choice of $\delta \in (0,\delta_1]$ and $t_0\in (0,T_0]$. 

We now further restrict $\delta_1$ by demanding that
\begin{equation} \label{delta-restrict}
0<\delta_1 < \min\biggl\{ C_3\delta_0,
\frac{1}{2}+\frac{r_0}{6}\biggr\} 
\end{equation}
in order to ensure that
\begin{equation} \label{Hc-lbnds}
\frac{3}{2}+\frac{r_0}{2}+3\Hc_*(x)\geq \frac{3}{2}+\frac{r_0}{2}-3\delta_1 >0 \AND
\frac{3}{2}+\frac{r_0}{2}+3\Hc(t,x) \geq \frac{3}{2}+\frac{r_0}{2}-3\delta_1>0
\end{equation}
for all $(t,x)\in (0,t_0]\times \Tbb^3$.

\subsubsection*{Crushing singularity}
Noting from \eqref{Wbb-def} and \eqref{stability-W-sol} that $\Sigma_{AB}$ is symmetric and trace-free, ie. \begin{equation} \label{Sigma-TF}
\delta^{AB}\Sigma_{AB} =0\AND \Sigma_{AB}=\Sigma_{BA},
\end{equation}
we see from \eqref{phys-spatial-metric}, \eqref{phys-2nd-FF}, \eqref{stability-decay-A.3}, \eqref{alpha-decay} and \eqref{att-bnds.1} that the trace of the physical second fundamental form $\Kttb_{AB}$ determined by that physical metric $\gb$ and the $t=\textit{const}$ hypersurfaces is given by
\begin{equation*}
\Kttb_A{}^A=\gttb^{AB}\Kttb_{AB} = \frac{1}{t^{\frac{3}{2}+\frac{r_0}{2}+3\Hc_*}} \frac{1}{\alpha_* \att}\biggl(\frac{3}{2}+\frac{r_0}{2}+3\Hc\biggr)=\frac{1}{t^{\frac{3}{2}+\frac{r_0}{2}+3\Hc_*}} \frac{1}{\alpha_*}\biggl(\frac{3}{2}+\frac{r_0}{2}+3\Hc_* + \Ord_{H^{k-1}(\Tbb^3)}(t^\zeta)\biggr),
\end{equation*}
where in deriving this we have used the fact that $\delta^{AB}r_{AB} = r_0$ due to \eqref{Kasner-rels-B} and \eqref{rAB-def}. Because of \eqref{alpha-asymp-data}, \eqref{att-bnds.2} and \eqref{Hc-lbnds}, we can bound $\Kttb_A{}^A$ below by
\begin{equation} \label{phs-2nd-FF-lbnd}
\inf_{x\in\Tbb^3}\Kttb_A{}^A(t,x) \gtrsim \frac{1}{t^{\frac{3}{2}+\frac{r_0}{2}-3\delta_1}}, \quad 0<t\leq t_0.
\end{equation}
Since $\frac{3}{2}+\frac{r_0}{2}-3\delta_1>0$, this implies that $\Kttb_A{}^A$ blows up uniformly as $t\searrow 0$, and hence, by definition \cite{Eardley:1979}, the $t=0$  hypersurface is a \textit{crushing singularity}. 

\subsubsection*{Past geodesic incompleteness}
Due to \eqref{phs-2nd-FF-lbnd}, $\Kttb_A{}^A\geq \Ct$ on $\Sigma_{t_0}$ for some $\Ct>0$. Furthermore, because the spacetime $(\gb,M_{0,t_0})$ is obtained from solving an initial value problem for the Einstein-scalar field equations and the scalar field stress energy tensor satisfies the strong energy condition, we know that $\Sigma_{t_0}$ is a  spacelike Cauchy hypersurface and that $\Rb_{ab}v^a v^b\geq 0$ for all timelike vector fields $v^a$.  Past timelike geodesic incompleteness is then a direct consequence of Hawking's singularity theorem \cite[Ch.~14, Thm.~55A]{ONeill:1983}.

\subsubsection*{Ricci curvature blow up and $C^2$-inextendibility of the physical metric}
Using \eqref{phy-scalar}, \eqref{phy-RabRab} and \eqref{alpha-decay}, we can express the physical curvature invariants $\Rb$ and $\Rb_{ab}\Rb^{ab}$ 
as
\begin{align*}
\Rb &= -\frac{3}{2\alpha_*^2\att^2t^{2(\frac{3}{2}+\frac{r_0}{2}+3\Hc_*)}} = \frac{3}{2\alpha_*^2 t^{2(\frac{3}{2}+\frac{r_0}{2}+3\Hc_*)}}\bigl(-1+\Ord_{H^{k-1}(\Tbb^3)}(t^\zeta)\bigr) 
\intertext{and}
\Rb_{ab}\Rb^{ab} &= \frac{9}{4\alpha_*^4\att^4t^{4(\frac{3}{2}+\frac{r_0}{2}+3\Hc_*)}}
=\frac{9}{4\alpha_*^4t^{4(\frac{3}{2}+\frac{r_0}{2}+3\Hc_*)}}\bigl(1+\Ord_{H^{k-1}(\Tbb^3)}(t^\zeta)\bigr) ,
\end{align*}
respectively. By \eqref{alpha-asymp-data} and \eqref{att-bnds.2}, we conclude that
\begin{equation*}
\sup_{x\in\Tbb^3}\Rb(t,x) \lesssim -\frac{1}{t^{2(\frac{3}{2}+\frac{r_0}{2}-3\delta_1)}} \AND
\inf_{x\in\Tbb^3}\Rb_{ab}(t,x)\Rb^{ab}(t,x) \gtrsim \frac{1}{t^{4(\frac{3}{2}+\frac{r_0}{2}-3\delta_1)}}
\end{equation*}
for all $t\in (0,t_0]$. Since $\frac{3}{2}+\frac{r_0}{2}-3\delta_1>0$, these formulas show that the physical curvature invariants $\Rb$ and $\Rb_{ab}\Rb^{ab}$ blow up uniformly as $t\searrow 0$, which establishes the $C^2$-inextendibility  of the physical metric $\gb$ across the $t=0$ hypersurface. 

\subsubsection*{Asymptotically pointwise Kasner behaviour}
From \eqref{alpha-def}, \eqref{conf-2nd-FF}, and \eqref{stability-decay-A.3}, we see that the second fundamental form $\Kttt_{AB}$ determined by the conformal metric $\gt$ and the $t=\textit{const}$ hypersurfaces, when expressed relative to the orthonormal frame $\{\et_0= \alphat^{-1}\del{t},\et_A=\et_A^\Lambda\del{\Lambda}\}$, is given by
\begin{equation} \label{Kt-decay}
2t\alphat \Kttt_{AB} = r_{AB}+2\Hc\delta_{AB}+2\Sigma_{AB} = \Kf_{AB}+\Ord_{H^{k-1}(\Tbb^3)}(t^\zeta)
\end{equation}
where 
\begin{equation} \label{Kf-def}
\Kf_{AB}= r_{AB}+2\Hc_*\delta_{AB}+2\Sigma^*_{AB}\in H^{k-1}(\Tbb^3,\Sbb{3})\subset C^{0}(\Tbb^3,\Sbb{3}).
\end{equation}
From Sobolev's inequality, we then deduce the pointwise limit
\begin{equation} \label{Asymp-Kasner-D}
\lim_{t\searrow 0}\bigl| 2t\alphat(t,x)\Kt_A{}^A(t,x) - \Kf_A{}^A(x)\bigr|, \quad \forall \, x\in \Tbb^3.
\end{equation}

Noting that
\begin{equation} \label{Sigma*-TF}
\delta^{AB}\Sigma^*_{AB} =0\AND \Sigma^*_{AB}=\Sigma^*_{BA}
\end{equation}
due to \eqref{stability-decay-A.3} and \eqref{Sigma-TF}, and that
\begin{equation}\label{rAB-trace}
\delta^{AB}r_{AB}=r_0 \AND r_A^B r_B^A = \sum_{A=1}^3r_A^2 = r_0^2+4 r_0
\end{equation}
by \eqref{Kasner-rels-B} and \eqref{rAB-def},
we find, following a short calculation, that
\begin{equation}\label{Asymp-Kasner-A}
 (\Kf_A{}^A)^2 - \Kf_A{}^B\Kf_B{}^A + 4\Kf_A{}^A = 4\bigl(6\Hc_*^2 + (6+2r_0)\Hc_*-\Sigma^*_A{}^B \Sigma^*_B{}^A - r^{AB}\Sigma^*_{AB}\bigr).
\end{equation}
Since the solution \eqref{stability-W-sol} satisfies the constraint equations, we observe from \eqref{H-cnstr-3} that the renormalised Hamiltionian constraint 
\begin{equation*}
6\Hc^2 + (6+2r_0)\Hc - \Sigma_{AB}\Sigma^{AB} - r_{AB}\Sigma^{AB} 
=-4t^{1-\ep_2}e_A(A^A)
             +4U_A A^A+6A^A A_A+N^{AB}N_{AB}-\frac{1}{2}(N_A{}^A)^2
\end{equation*}
is satisfied.
With the help of the calculus inequalities from Appendix \ref{calc}, the bound \eqref{stability-W-bnd}, the decay estimates \eqref{stability-decay-A.1}-\eqref{stability-decay-A.3}, and the fact that $1-\ep_2>0$,  we further observe that the renormalised Hamiltonian constraint can be expressed as 
\begin{equation*}
 (\Kf_A{}^A)^2 - \Kf_A{}^B\Kf_B{}^A + 4\Kf_A{}^A = \Ord_{H^{k-1}(\Tbb^3)}(t^{\zetat})
\end{equation*}
where $\zetat=\min\{1-\ep_2,\zeta\}>0$. Then letting $t\searrow 0$ shows that
\begin{equation} \label{Asymp-Kasner-B}
 (\Kf_A{}^A)^2 - \Kf_A{}^B\Kf_B{}^A + 4\Kf_A{}^A = 0.
\end{equation}
Furthermore, from \eqref{Kf-def}, \eqref{Sigma*-TF} and \eqref{rAB-trace}, we have that $\Kf_A{}^A = r_0 + 6\Hc_*$, which, in turn, implies by \eqref{Hc-lbnds} that $\Kf_A{}^A\geq -3$.
On the other hand, solving \eqref{Asymp-Kasner-B} for $\Kf_A{}^A$, we find that 
$\Kf_A{}^A=\pm\sqrt{4+\Kf_{A}{}^B \Kf_{B}{}^A}-2$. Only the positive solution is compatible with the lower bound $\Kf_A{}^A\geq -3$, and hence, we have
\begin{equation}\label{Asymp-Kasner-C}
\Kf_A{}^A=\sqrt{4+\Kf_{A}{}^B \Kf_{B}{}^A}-2\geq 0.
\end{equation}
Since 
\begin{equation} \label{cESF-sol}
\Bigl\{\gt = -\alphat^2 dt\otimes dt + \gt_{\Sigma\Omega}dx^\Sigma \otimes dx^\Omega,\; \tau = t\Bigr\},
\end{equation}
where  $\alphat$ and $\gt$ are determined by  \eqref{alphat-gt-et}, defines a solution of the conformal Einstein-scalar-field equations \eqref{cESF.1}-\eqref{cESF.2}, we conclude from \eqref{Asymp-Kasner-D}, \eqref{Asymp-Kasner-B} and \eqref{Asymp-Kasner-C} that the solution \eqref{cESF-sol} is \textit{asymptotically pointwise Kasner on $\Tbb^3$}, cf.~\cite[Def.~1.1]{BeyerOliynyk:2024b}. 

\subsubsection*{Weyl curvature blow up}
From \eqref{Et-form} and \eqref{Ht-form}, it is straightforward, using the decay estimates \eqref{stability-decay-A.1}-\eqref{stability-decay-A.3}, the bounds \eqref{stability-W-bnd} and \eqref{asymp-data-bnd-A}, and the calculus inequalities from Appendix \ref{calc}, to verify that 
\begin{equation} \label{Et-Ht-decay}
t^2 \alphat^2 \Et_{AB}= \Ec_{AB} + \Ord_{H^{k-1}(\Tbb^3)}(t^\zeta) \AND
t^2 \alphat^2 \Ht_{AB}  = \Ord_{H^{k-1}(\Tbb^3)}(t^\zeta),
\end{equation}
where 
\begin{align}
\Ec_{AB} &=\ec_{AB}+ \Hc_*\Sigma^*_{AB} + \Bigl(\frac{1}{2}r_{AB}-\frac{r_0}{6}\delta_{AB}\Bigr)\Hc_* + \frac{r_0}{3}\Sigma^*_{AB}  - \Sigma^*_{AC}\Sigma^*_B{}^C - \frac{1}{2}r_B^C\Sigma^*_{AC} \notag \\ 
&\quad- \frac{1}{2}r_{AC}\Sigma^*_B{}^C  + \frac{r_0}{6}\Sigma^*_{AB} + \frac{1}{3}\delta_{AB}\Sigma^*_{CD}\Sigma^{*CD} + \frac{1}{3}\delta_{AB}r^{CD}\Sigma^*_{CD}\label{Fc-def}, \\ 
\ec_{AB} &= \frac{r_0}{3}\delta_{AB} +  \frac{r_0}{4}r_{AB}- \frac{1}{4}r_{AC}r_B^C,   \label{ec-def} 
\end{align}
and again the implicit constants in the above decay estimates are independent of the choice of $\delta \in (0,\delta_1]$ and $t_0\in (0,T_0]$. By \eqref{phys-Weyl-invariant}, \eqref{alpha-decay} and \eqref{Et-Ht-decay}, it then follows that the physical Weyl curvature invariant $\Cb^{abcd}\Cb_{abcd}$ satisfies
\begin{equation} \label{Wey-invariant-exp}
\Cb^{abcd}\Cb_{abcd} =\frac{8}{t^{7+2r_0+12\Hc_*}\alpha_*^4}\Bigl(\Ec_{ab}\Ec^{ab}+\Ord_{H^{k-1}(\Tbb^3)}(t^\zeta)\Bigr).
\end{equation}

Now, by \eqref{rAB-trace} and \eqref{ec-def}, we have
\begin{equation*}
    \ec_{AB}\ec^{AB} = \sum_{A=1}^3 \Bigl(\frac{r_0}{3}+\frac{r_0}{4}r_A-\frac{1}{4}r_A^2\Bigr)^2.
\end{equation*}
Clearly 
    $\ec_{AB}\ec^{AB}\geq 0$
and we claim that
$\ec_{AB}\ec^{AB}=0$ 
if and only if $r_1=r_2=r_3=0$.
To see why this holds, we observe from \eqref{Kasner-rels-B} that there can be 3 zeros, 1 zero or no zeros in the set $\{r_1,r_2,r_3\}$. If there is only one zero, say $r_1=0$ and $r_2,r_3\neq 0$, then $r_0>0$ and $\frac{r_0}{3}+\frac{r_0}{4}r_1-\frac{1}{4}r_1^2 = \frac{r_0}{3} > 0$, from which it follows that $\ec_{AB}\ec^{AB}>0$. On  the other hand, if $r_1,r_2,r_3\neq 0$, ie. there are no zeros, then $\ec_{AB}\ec^{AB}=0$ implies that $\frac{r_0}{3}+\frac{r_0}{4}r_A-\frac{1}{4}r_A^2 = 0$,
for $A=1,2,3$. Solving for $r_A$ gives
$r_A = \frac{1}{6}\Bigl(3r_0\pm\sqrt{57}r_0\Bigr)$. But this would imply that $r_1+r_2+r_3\neq r_0$, a contradiction, and so $\ec_{AB}\ec^{AB}>0$ must hold in this case. From these arguments, we conclude that $\ec_{AB}\ec^{AB}>0$ if and only if there is at least one nonzero conformal Kasner exponent $r_A$.

Now from  \eqref{stability-decay-A.3}, \eqref{asymp-data-bnd-A} and  \eqref{Fc-def}, we see, with the help of the calculus inequalities from Appendix \ref{calc} that 
\begin{equation*}
    \norm{\Ec_{AB} \Ec^{AB}-\ec_{AB}\ec^{AB}}_{H^{k-1}(\Tbb^3)} \leq C_4(1+\delta)\delta, \quad 0<\delta\leq \delta_1.
\end{equation*}
Assuming now that $(r_1,r_2,r_3)\neq 0$, then $\ec_{AB}\ec^{AB}>0$, it then follows from 
the above bound and Sobolev's inequality (Theorem \ref{Sobolev}) that there exists a $\delta_2\in (0,\delta_1]$ such that   
\begin{equation} \label{Ec-lbnd}
\inf_{x\in\Tbb^3}\Ec_{AB}(x) \Ec^{AB}(x) >0
\end{equation}
for any $\delta \in (0,\delta_2]$. From \eqref{Hc-lbnds}, \eqref{Wey-invariant-exp}, \eqref{Ec-lbnd} and Sobolev's inequality, we then deduce the existence of a $t_1\in (0,t_0]$ such that 
\begin{equation*}
\inf_{x\in\Tbb^3}\Cb^{abcd}(t,x)\Cb_{abcd}(t,x) \gtrsim \frac{1}{t^{1+4(\frac{3}{2}+\frac{r_0}{2}-3\delta_1)}}
\end{equation*}
for all $t\in (0,t_1]$. But $\frac{3}{2}+\frac{r_0}{2}-3\delta_1>0$, and consequently, the physical Weyl curvature invariant $\Cb^{abcd}\Cb_{abcd}$ blows up uniformly as $t\searrow 0$ provided that the Kasner exponents $q_I$ are not all equal.

\subsection{Localised past stability of the subcritical Kasner-scalar field spacetimes}
  A localised version of Theorem \ref{glob-stab-thm} is readily obtained by modifying its proof as follows: replace Proposition \ref{prop:locA} with Proposition \ref{prop:locC} to establish the local-in-time existence of solutions, and use Proposition \ref{prop:Fuch-global-loc} in place of Proposition \ref{prop:Fuch-global} to derive global bounds. Otherwise, the proof remains essentially the same. As in Proposition \ref{prop:Fuch-global-loc}, existence is established on truncated cone domains $\Omega_{\qv}$ where $\qv = (0, t_0, \rho_0, \rho_1, \epsilon_2)$, $\epsilon_2$ satisfies \eqref{eigenval-posB}, $t_0 > 0$, $\rho_0 \in (0, L)$, and $\rho_1 = t_0^{\epsilon_2 - 1}(1 - \epsilon_2)(1 - \vartheta)\rho_0$ with $\vartheta \in (0, 1)$. Note that these domains are capped above by $\{t_0\} \times \mathbb{B}_{\rho_0}$ and below by $\{0\} \times \mathbb{B}_{\vartheta \rho_0}$. 
Because the modifications needed to prove a localised version of Theorem \ref{glob-stab-thm} are straightforward, we omit the details and only state the result in the following theorem.

\begin{thm}\label{loc-stab-thm}
Suppose $T_0>0$, $k_0\in \Zbb_{>\frac{3}{2}+1}$, the Kasner exponents $q_I\in \Rbb$, $I=1,2,3$, and  $P\in (0,\sqrt{1/3}]$ satisfy the Kasner relations \eqref{Kasner-rels-A} and the subcritical condition \eqref{SCR},  $\ep_1\in (0,1)$ and $\ep_2\in \Rbb$ satisfy \eqref{eigenval-posA}  and \eqref{eigenval-posB}, respectively, where $r_0$ and $r_I$, $I=1,2,3$, are defined by \eqref{Kasner-rels-B}, $\nu>0$ satisfies \eqref{eigenval-posC}, $k_1\in \Nbb_0$ is chosen large enough so that \eqref{k-nu-fix} holds\footnote{If the right hand side of \eqref{k-nu-fix} is $0$, then $k_1=0$ is allowed.},  
$0<\rho_0<L$ and $\vartheta \in (0,1)$. 
Then there exists a $\delta_0 >0$ such that for every $t_0\in (0,T_0]$, $\delta \in (0,\delta_0]$, and\footnote{$\Wbb$ is the vector space defined above by \eqref{Wbb-def}.} 
\begin{equation*}
W_0 = ( e_{0P}^\Sigma, \alpha_0, A_{0P}, U_{0P}, \Hc_0, \Sigma_{0PQ}, N_{0PQ})^{\tr} \in H^{k}(\mathbb{B}_{\rho_0},\Wbb), \quad k=k_0+k_1,
\end{equation*}
satisfying $\inf_{x\in\mathbb{B}_{\rho_0}}\alpha_0(x)>0$, $\inf_{x\in\mathbb{B}_{\rho_0}}\det\bigl(e_{0P}^\Sigma(x)\bigr)>0$, the constraint equations \eqref{A-cnstr-3}-\eqref{H-cnstr-3} at  $t=t_0$, and
$\norm{W_0}_{H^{k}(\mathbb{B}_{\rho_0})} < \delta$,
there exists a 
\begin{equation*} 
W= \bigl( e_{P}^\Sigma, \alpha, A_{P}, U_{P}, \Hc, \Sigma_{PQ}, N_{PQ}\bigr)^{\tr} \in C^1(\grave{\Omega}_{\qv}), 
\end{equation*}
where $\qv=(0,t_0,\rho_0,\rho_1,\ep_2)$ and
$\rho_1 = t_0^{\ep_2-1}(1-\ep_2)(1-\vartheta)\rho_0$,
that solves, for all choices of the parameters $\gamma,\rho \in\Rbb$, the evolution equations \eqref{EEc.1}-\eqref{EEc.7} on $\Omega_{\qv}$, and satisfies the initial condition $W|_{t=t_0}=W_0$, the constraint equations \eqref{A-cnstr-3}-\eqref{H-cnstr-3} on $\Omega_{\qv}$, and  $\alpha>0$ and $\det(e_{P}^\Sigma)>0$ on $\grave{\Omega}_{\qv}$. Moreover the following hold:
\begin{enumerate}[(a)]
\item Letting $\rho(t) = \rho_0 + (1-\vartheta) \rho_0 \bigl( \bigl(\frac{t}{t_0}\bigr)^{1-\ep_2}-1\bigr)$,
$W(t) \in H^{k}(\mathbb{B}_{\rho(t)})$ and $\del{t}W(t)\in H^{k-1}(\mathbb{B}_{\rho(t)})$ for each $t\in (0,t_0]$.
\item The solution $W$ is uniformly bounded by
\begin{equation*}
\norm{W(t)}_{H^{k}(\mathbb{B}_{\rho(t)})} \lesssim \delta
\end{equation*}
for all $t\in (0,t_0]$, 
and there exist a constant $\zeta>0$ and functions 
$\alpha_*, \Hc_*, \Sigma_{AB}^*\in H^{k-1}(\mathbb{B}_{\vartheta\rho_0})\subset C^{0}_b(\mathbb{B}_{\vartheta\rho_0})$
satisfying
\begin{gather*}
0<\inf_{x\in\mathbb{B}_{\vartheta\rho_0}}\alpha_*(x) \leq \sup_{x\in\mathbb{B}_{\vartheta\rho_0}}\alpha_*(x)\lesssim 1, \quad \norm{\Hc_*}_{L^\infty(\mathbb{B}_{\vartheta\rho_0})}\leq \delta_1,\\
\Sigma^*_{AB}=\Sigma^*_{BA},  \quad \delta^{AB}\Sigma^*_{AB}=0 \AND \max\big\{\norm{\Hc_*}_{H^{k-1}(\mathbb{B}_{\vartheta\rho_0})}, \norm{\Sigma_{AB}^*}_{H^{k-1}(\mathbb{B}_{\vartheta\rho_0})} \bigr\}\lesssim \delta
\end{gather*}
such the components of $W$ decay according to
\begin{gather*}
e_P^\Omega = \Ord_{H^{k-1}(\mathbb{B}_{\vartheta\rho_0})}(t^\zeta),\quad 
\alpha = t^{\ep_1+\frac{r_0}{2}+3\Hc_*}\alpha_*\Bigl(1+\Ord_{H^{k-1}(\mathbb{B}_{\vartheta\rho_0})}(t^\zeta)\Bigr), \\
A_P = \Ord_{H^{k-1}(\mathbb{B}_{\vartheta\rho_0})}(t^\zeta), \quad
U_P = \Ord_{H^{k-1}(\mathbb{B}_{\vartheta\rho_0})}(t^\zeta), \quad
N_{PQ} = \Ord_{H^{k-1}(\mathbb{B}_{\vartheta\rho_0})}(t^\zeta), \\
\Hc = \Hc_* + \Ord_{H^{k-1}(\mathbb{B}_{\vartheta\rho_0})}(t^\zeta) \AND
\Sigma_{PQ} = \Sigma_{PQ}^* + \Ord_{H^{k-1}(\mathbb{B}_{\vartheta\rho_0})}(t^\zeta)
\end{gather*}
for $t\in (0,t_0]$. Furthermore, the explicit and implicit constants in these estimates are independent to the choice of $\delta\in (0,\delta_0]$ and $t\in (0,t_0]$.
\item The pair
\begin{equation*}
\biggl\{\gb = t\bigl(-\alphat^2 dt\otimes dt + \gt_{\Sigma\Omega}dx^\Sigma \otimes dx^\Omega\bigr),\; \phi = \frac{\sqrt{3}}{2}\ln(t)\Biggr\}
\end{equation*}
where  
\begin{equation*}
\alphat = t^{-\ep_1}\alpha,\quad 
(\gt_{\Sigma\Omega}) =(\delta^{AB}\et_A^{\Sigma}\et_B^{\Omega})^{-1} \AND \et_A^\Omega = t^{\ep_1-\ep_2}\alpha^{-1}e_A^\Omega
\end{equation*}
determines a classical solution of the  
Einstein-scalar-field equations \eqref{ESF.1}-\eqref{ESF.2} on $\Omega_{\qv}$. Furthermore, the trace of the physical second fundamental form $\Kttb_{AB}$ determined by that physical metric $\gb$ and the $t=\textit{const}$ hypersurfaces satisfies
\begin{equation*}
\Kttb_A{}^A=\gttb^{AB}\Kttb_{AB} = \frac{1}{t^{\frac{3}{2}+\frac{r_0}{2}+3\Hc_*}} \frac{1}{\alpha_*}\biggl(\frac{3}{2}+\frac{r_0}{2}+3\Hc_* + \Ord_{H^{k-1}(\mathbb{B}_{\vartheta\rho_0})}(t^\zeta)\biggr)
\end{equation*}
for $t\in (0,t_0]$, and is bounded below by
\begin{equation*}
\inf_{x\in\mathbb{B}_{\vartheta\rho_0}}\Kttb_A{}^A(t,x) \gtrsim \frac{1}{t^{\frac{3}{2}+\frac{r_0}{2}-3\delta_1}}, \quad 0<t\leq t_0.
\end{equation*}
Since $\frac{3}{2}+\frac{r_0}{2}-3\delta_1>0$, $\Kttb_A{}^A$ blows up uniformly as $t\searrow 0$ and the hypersurface $t=0$ is a \textit{crushing singularity}, cf.~\cite{Eardley:1979}. 
\item The pair 
\begin{equation*}
\Bigl\{\gt = -\alphat^2 dt\otimes dt + \gt_{\Sigma\Omega}dx^\Sigma \otimes dx^\Omega,\; \tau = t\Bigr\}
\end{equation*}
determines a classical solution of the  
conformal Einstein-scalar-field equations \eqref{cESF.1}-\eqref{cESF.2} on $\Omega_{\qv}$. Furthermore, the second fundamental form $\Kttt_{AB}$ determined by the conformal metric $\gt$ and the $t=\textit{const}$ hypersurfaces, when expressed relative to the orthonormal frame $\{\et_0= \alphat^{-1}\del{t},\et_A=\et_A^\Lambda\del{\Lambda}\}$,
decays according to
\begin{equation*}
2t\alphat \Kttt_{AB}  = \Kf_{AB}+\Ord_{H^{k-1}(\mathbb{B}_{\vartheta\rho_0})}(t^\zeta)
\end{equation*}
where 
\begin{equation*} 
\Kf_{AB}= r_{AB}+2\Hc_*\delta_{AB}+2\Sigma^*_{AB}\in H^{k-1}(\mathbb{B}_{\vartheta\rho_0},\Sbb{3})\subset C^{0}_b(\mathbb{B}_{\vartheta\rho_0},\Sbb{3})
\end{equation*}
satisfies 
\begin{equation*} 
 (\Kf_A{}^A)^2 - \Kf_A{}^B\Kf_B{}^A + 4\Kf_A{}^A = 0
\AND
\Kf_A{}^A=\sqrt{4+\Kf_{A}{}^B \Kf_{B}{}^A}-2\geq 0.
\end{equation*}
In particular, this implies that the conformal solution  $\{\gt,\tau\}$ is \textit{asymptotically pointwise Kasner} on $\mathbb{B}_{\vartheta\rho_0}$ relative to the frame $\{\et_0= \alphat^{-1}\del{t},\et_A=\et_A^\Lambda\del{\Lambda}\}$, cf.~\cite[Def.~1.1]{BeyerOliynyk:2024b}.
\item The physical solution $\{\gb,\phi\}$ of the Einstein-scalar field equations on $\Omega_{\qv}$ is past $C^2$ inextendible at $t=0$ and past timelike geodesically incomplete. The curvature invariants   $\Rb=\gb^{ab}\Rb_{ab}$ and $\Rb_{ab}\Rb^{ab}$ of the physical metric $\gb$ 
satisfy
\begin{equation*}
\Rb = \frac{3}{2\alpha_*^2 t^{2(\frac{3}{2}+\frac{r_0}{2}+3\Hc_*)}}\bigl(-1+\Ord_{H^{k-1}(\mathbb{B}_{\vartheta\rho_0})}(t^\zeta)\bigr) 
\AND
\Rb_{ab}\Rb^{ab} = \frac{9}{4\alpha_*^4t^{4(\frac{3}{2}+\frac{r_0}{2}+3\Hc_*)}}\bigl(1+\Ord_{H^{k-1}(\mathbb{B}_{\vartheta\rho_0})}(t^\zeta)\bigr)
\end{equation*}
for $t\in (0,t_0]$, respectively, and are bounded above and below by
\begin{equation*}
\sup_{x\in\mathbb{B}_{\vartheta\rho_0}}\Rb(t,x) \lesssim -\frac{1}{t^{2(\frac{3}{2}+\frac{r_0}{2}-3\delta_1)}} \AND
\inf_{x\in\mathbb{B}_{\vartheta\rho_0}}\Rb_{ab}(t,x)\Rb^{ab}(t,x) \gtrsim \frac{1}{t^{4(\frac{3}{2}+\frac{r_0}{2}-3\delta_1)}}
\end{equation*}
for $t\in (0,t_0]$, respectively. Since $\frac{3}{2}+\frac{r_0}{2}-3\delta_1>0$, $\Rb$ and $\Rb_{ab}\Rb^{ab}$ blow up uniformly as $t\searrow 0$.
\item 
The Weyl curvature invariant $\Cb^{abcd}\Cb_{abcd}$ of the physical metric $\gb$ decays according to
\begin{equation*}
\Cb^{abcd}\Cb_{abcd} =\frac{8}{t^{7+2r_0+12\Hc_*}\alpha_*^4}\Bigl(\Ec_{AB}\Ec^{AB}+\Ord_{H^{k-1}(\mathbb{B}_{\vartheta\rho_0})}(t^\zeta)\Bigr)
\end{equation*}
for $t\in (0,t_0]$ where
\begin{equation*}
\ec_{AB} = \frac{r_0}{3}\delta_{AB} +  \frac{r_0}{4}r_{AB}- \frac{1}{4}r_{AC}r_B^C
\end{equation*}
and 
\begin{align*}
\Ec_{AB} &=\ec_{AB}+ \Hc_*\Sigma^*_{AB} + \Bigl(\frac{1}{2}r_{AB}-\frac{r_0}{6}\delta_{AB}\Bigr)\Hc_* + \frac{r_0}{2}\Sigma^*_{AB}  - \Sigma^*_{AC}\Sigma^*_B{}^C - \frac{1}{2}r_B^C\Sigma^*_{AC} \notag\\
&\quad - \frac{1}{2}r_{AC}\Sigma^*_B{}^C + \frac{1}{3}\delta_{AB}\Sigma^*_{CD}\Sigma^{*CD} + \frac{1}{3}\delta_{AB}r^{CD}\Sigma^*_{CD}\in H^{k-1}(\Tbb^3,\Sbb{3}) \subset C^0(\Tbb^3,\Sbb{3}).
\end{align*}
Moreover if the Kasner exponents $q_I$ are not all equal, or equivalently, $(r_1,r_2,r_3)\neq (0,0,0)$, then 
\begin{equation*}
    \inf_{x\in\mathbb{B}_{\vartheta\rho_0}}\Ec_{AB}(x) \Ec^{AB}(x) >0
\end{equation*} and there exists a $t_1\in (0,t_0]$ such that  
$\Cb^{abcd}\Cb_{abcd}$ is bounded below by
\begin{equation*}
\inf_{x\in\mathbb{B}_{\vartheta\rho_0}}\Cb^{abcd}(t,x)\Cb_{abcd}(t,x) \gtrsim \frac{1}{t^{1+4(\frac{3}{2}+\frac{r_0}{2}-3\delta_1)}}, \quad 0<t\leq t_1.
\end{equation*}
Since  $\frac{3}{2}+\frac{r_0}{2}-3\delta_1>0$, $\Cb^{abcd}\Cb_{abcd}$ blows up uniformly as $t\searrow 0$ provided $(r_1,r_2,r_3)\neq (0,0,0)$.
\end{enumerate}
\end{thm}

We emphasise that Remarks~\ref{rem:AVTD} and ~\ref{rem:Weyl} both apply to Theorem~\ref{loc-stab-thm}.

\begin{rem}\label{rem:Kasner-idata-stability-B}
Due to \eqref{Kasner:t->0}, it is clear that the Kasner-scalar field initial data satisfies 
\begin{equation*}
\lim_{t_0\searrow 0}\norm{W^{\textrm{Kas}}_0}_{H^{k}(\mathbb{B}_{\rho_0})} =0,
\end{equation*}
and $\alpha>0$ and $\det(\ep^\Omega_P)>0$ on $\overline{\mathbb{B}}_{\rho_0}$ for any $k\in \Nbb_0$ and any $\rho_0\in (0,L)$.
As in Remark \ref{rem:Kasner-idata-stability}, we can, by taking $t_0>0$ sufficiently small, guarantee that the Kasner-scalar field initial data satisfy the conditions of Theorem \ref{loc-stab-thm}. This implies the existence of a open neighbourhood of initial data around the Kasner-scalar field initial data that also satisfy the conditions of Theorem \ref{glob-stab-thm}, which, in turn, implies the localised past nonlinear stability of the Kasner-scalar field solutions and their big bang singularities. 
\end{rem}

\appendix

\section{\label{mat-eqn}Matrix inequalities}
In this appendix, we collect together a few matrix inequalities that will be employed in Section \ref{loc-exist-Omega}. Suppose $U=(U_A), V=(V_B)\in\Rbb^3$ and $X=(X_{AB})\in \Mbb{3}$. 
Then using the Euclidean metric $\delta_{AB}$ to raise and lower indices, we have
\begin{equation*}
U_A V_B X^{AB}=\ipe{U V^T}{X}
\end{equation*}
where the inner-product $\ipe{\cdot}{\cdot}$ is the Frobenius inner-product \eqref{Frobenius}. With the help of the Cauchy-Schwartz inequality, we see that $U_A V_B X^{AB}$ can be bounded above by
\begin{equation*}
U_A V_B X^{AB}\leq |U V^T| |X|.
\end{equation*}
But 
\begin{equation*}
|U V^T|^2 = \ipe{U V^T}{U V^T}= \text{Tr}(V^T U^T U V^T) = |U|^2 |V|^2,  
\end{equation*}
and so we have
\begin{equation}\label{mat-inq-A}
U_A V_B X^{AB}\leq |U| |V| |X|.
\end{equation}
From this inequality and \eqref{Frobenius-inq}, we deduce that
\begin{equation*}
|XU|^2 = U_B U_C X^{AB}X_{A}{}^{C}\leq |U|^2 |X^T X| \leq  |U|^2|X^T||X|= |U|^2|X|^2,
\end{equation*}
or after taking the square root,
\begin{equation}  \label{mat-inq-B}
|XU| \leq |U||X|.
\end{equation}

\section{\label{calc}Calculus inequalities}
In this appendix, we collect, for the convenience of the reader, a number of calculus inequalities that are employed throughout this article. The proofs of the following inequalities are well known and may be found, for example, in 
the books \cite[Ch.~4]{AdamsFournier:2003}, \cite[Ch.~13, \S 1-3]{TaylorIII:1996} and \cite[Ch.~VI, \S 3]{Choquet_et_al:2000}. 
In this section, $U$ is an open subset of $\Tbb^n$.
\begin{thm}{\emph{[Sobolev's inequality]}} \label{Sobolev} Suppose 
 $k\in \Zbb_{\geq 0}$ and $0<\alpha \leq k-n/2\leq 1$. Then
$H^{k}(U)\subset C^{0,\alpha}(U)$ and 
\begin{equation*}
\norm{u}_{L^\infty}\lesssim \norm{u}_{C^{0,\alpha}}  \lesssim  \norm{u}_{H^{k}}
\end{equation*}
for all $u\in H^{k}(U)$.
\end{thm}

\begin{thm}{\emph{[Product and commutator inequalities]}} \label{Product} $\;$

\begin{enumerate}[(i)]
\item
Suppose $k\in \Zbb_{\geq 1}$ and $|\alpha|=k$.
Then
\begin{align*}
\norm{D^\alpha (uv)}_{L^2} \lesssim \norm{u}_{H^{k}}\norm{v}_{L^{\infty}} + \norm{u}_{L^{\infty}}\norm{v}_{H^{k}} \label{clacpropB.2.1}
\intertext{and}
\norm{[D^\alpha,u]v}_{L^2} \lesssim \norm{D u}_{L^{\infty}}\norm{v}_{H^{k-1}} + \norm{D u}_{
H^{k-1}}\norm{v}_{L^{\infty}}
\end{align*}
for all $u,v \in C^\infty(U)$.
\item[(ii)]  Suppose $k_1,k_2,k_3\in \Zbb_{\geq 0}$, $\;k_1,k_2\geq k_3$, and $k_1+k_2-k_3 > n/2$. Then
\begin{equation*}
\norm{uv}_{H^{k_3}} \lesssim \norm{u}_{H^{k_1}}\norm{v}_{H^{k_2}}
\end{equation*}
for all $u\in H^{k_1}(U)$ and $v\in H^{k_2}(U)$.
\end{enumerate}
\end{thm}

\begin{thm}{\emph{[Moser's inequality]}}  \label{Moser}
Suppose   $k\in \Zbb_{\geq 1}$, $0\leq s\leq k$, $|\alpha|=s$, $f\in C^k(V)$, where
$V$ is open and bounded in $\Rbb^N$ and contains $0$, and $f(0)=0$. Then
\begin{equation*}
\norm{D^\alpha f(u)}_{L^{2}} \leq C\bigl(\norm{f}_{C^k(\overline{V})}\bigr)(1+\norm{u}^{k-1}_{L^\infty})\norm{u}_{H^{k}}
\end{equation*}
for all $u \in C^0(U)\cap L^\infty(\Tbb^{n})\cap H^{k,p}(U)$ with
$u(x) \in V$ for all $x\in U$.
\end{thm}

\bibliographystyle{amsplain}
\bibliography{Tetrad_Formalism_v2}

\end{document}